\documentclass[11pt]{amsart}
\usepackage{fullpage}
\usepackage[foot]{amsaddr}
\usepackage{microtype}
\usepackage[OT1]{fontenc}
\usepackage{eulervm}
\usepackage[tt=false]{libertine} 
\usepackage{bbold}
\usepackage{amsmath}
\usepackage{amssymb}
\usepackage{amsthm}
\usepackage{thmtools} 
\usepackage[linesnumbered,boxed,ruled,vlined]{algorithm2e}
\usepackage{algpseudocode}
\usepackage{enumitem}
\usepackage{xifthen}
\usepackage{xspace}

\usepackage[margin=1cm]{caption} 
\usepackage{subfig}

\usepackage[thinlines]{easytable}

\usepackage[bookmarks=true,hypertexnames=false,pagebackref]{hyperref}
\hypersetup{colorlinks=true, citecolor=blue, linkcolor=red, urlcolor=blue}

\usepackage{pgfplots}
\pgfplotsset{compat=1.16}
\usepackage{tikz}
\usetikzlibrary{arrows,arrows.meta,backgrounds,calc,fit,decorations.pathreplacing,decorations.markings,shapes.geometric}

\tikzstyle{internal} = [draw, fill, shape=circle]
\tikzstyle{external} = [shape=circle]
\tikzstyle{square}   = [draw, fill, rectangle]
\tikzstyle{triangle} = [draw, fill, regular polygon, regular polygon sides=3, inner sep=3pt]
\tikzstyle{pentagon} = [draw, fill, regular polygon, regular polygon sides=5, inner sep=2pt, minimum size=14pt]
\tikzset{every fit/.append style=text badly centered}

\usetikzlibrary{positioning,chains,fit,shapes,calc}
\usetikzlibrary{trees}
\usetikzlibrary{decorations.pathreplacing}
\usetikzlibrary{decorations.pathmorphing}
\usetikzlibrary{decorations.markings}
\tikzset{>=latex} 

\usepackage{ifthen}

\usepackage{cleveref}

\usepackage[textsize=tiny]{todonotes}

\usepackage[normalem]{ulem}

\usepackage{mleftright}

\usepackage{cool}
\Style{DSymb={\mathrm d},DShorten=true,IntegrateDifferentialDSymb=\mathrm{d}}

\newcommand{\tp}[1]{{\left( #1 \right)}}

\renewcommand{\Pr}{\mathop{\mathrm{Pr}}\nolimits}

\def\*#1{\mathbf{#1}}
\def\+#1{\mathcal{#1}}
\def\-#1{\mathrm{#1}}
\def\=#1{\mathbb{#1}}
\def\^#1{\mathbb{#1}}

\newcommand{\norm}[2]{\ensuremath{\Vert #2 \Vert_{#1}}}
\newcommand{\abs}[1]{\ensuremath{\left\vert#1\right\vert}}

\newcommand{\defeq}{:=}

\renewcommand{\P}{\Pr}

\newcommand{\NP}{\textnormal{\textbf{NP}}}
\newcommand{\BIS}{\#\textnormal{\textbf{BIS}}}

\newcommand{\DTV}[2]{\-D_{\mathrm{TV}}\left({#1},{#2}\right)}

\newcommand{\nbd}[2]{N_{#1}^{#2}}

\newtheorem{theorem}{Theorem}

\newtheorem{lemma}[theorem]{Lemma}
\newtheorem{claim}[theorem]{Claim}
\newtheorem{observation}[theorem]{Observation}
\newtheorem{proposition}[theorem]{Proposition}

\theoremstyle{definition}
\newtheorem{condition}[theorem]{Condition}
\newtheorem{definition}[theorem]{Definition}

\theoremstyle{remark}
\newtheorem*{remark}{Remark}

\crefname{theorem}{Theorem}{Theorems}
\crefname{observation}{Observation}{Observations}
\crefname{claim}{Claim}{Claims}
\crefname{condition}{Condition}{Conditions}
\crefname{algorithm}{Algorithm}{Algorithms}
\crefname{property}{Property}{Properties}
\crefname{example}{Example}{Examples}
\crefname{fact}{Fact}{Facts}
\crefname{lemma}{Lemma}{Lemmas}
\crefname{corollary}{Corollary}{Corollaries}
\crefname{definition}{Definition}{Definitions}
\crefname{remark}{Remark}{Remarks}
\crefname{proposition}{Proposition}{Propositions}
\crefname{equation}{equation}{equations}
\crefname{enumi}{Case}{Case}
\creflabelformat{enumi}{(#2#1#3)}

\makeatletter
\def\prob#1#2#3{\goodbreak\begin{list}{}{\labelwidth\z@ \itemindent-\leftmargin
      \itemsep\z@  \topsep6\p@\@plus6\p@
      \let\makelabel\descriptionlabel}
  \item[\textbf{Name}]#1
  \item[\textbf{Instance}]#2
  \item[\textbf{Output}]#3
  \end{list}}
\makeatother

\makeatletter
\providecommand\@dotsep{5}
\def\listtodoname{Todo list}
\def\listoftodos{\@starttoc{tdo}\listtodoname}
\makeatother

\usepackage{nicefrac,comment}

\newboolean{doubleblind}
\setboolean{doubleblind}{false}

\title{Rapid mixing in positively weighted restricted Boltzmann machines}

\ifdoubleblind
\author{Author(s)}
\else

  \author{Weiming Feng, Heng Guo, Minji Yang}

\address[Heng Guo]{School of Informatics, University of Edinburgh, Informatics Forum, UK}
\address[Weiming Feng and Minji Yang]{School of Computing and Data Science, The University of Hong Kong, HK}

\email{wfeng@hku.hk}
\email{hguo@inf.ed.ac.uk}
\email{ymjessen02@connect.hku.hk}

\fi

\date{}
    
\begin{document}

\begin{abstract}
  \sloppy
  We show polylogarithmic mixing time bounds for the alternating-scan sampler for positively weighted restricted Boltzmann machines.
  This is done via analysing the same chain and the Glauber dynamics for ferromagnetic two-spin systems,
  where we obtain new mixing time bounds up to the critical thresholds.
\end{abstract}

\allowdisplaybreaks
\maketitle

\section{Introduction}

The restricted Boltzmann machine (RBM) \cite{AHS85,Smo86} is a popular model to represent many different types of data \cite{HOT06,SMH07,MH10}.
Its simple two-layer structure also makes it useful as a basic building block for deep belief networks \cite{HOT06}.
The development of RBMs is recognised as a main contribution for Geoffrey E.\ Hinton's Nobel prize in physics in 2024 \cite{nobel24}.
As it would distract from the main focus of our paper, we do not attempt to give a comprehensive overview of RBMs here.

The training of RBMs relies on estimating the gradient,
which is often done via the MCMC method.
One of the most popular Markov chains here is the alternating-scan sampler \cite{Hin02}, 
which updates the two layers of the variables alternately conditioned on the other layer.
The mixing time of this sampler (namely the time it takes to converge to its stationary distribution) is very important in learning RBMs, 
as emphasised in Hinton's practical guide \cite{Hin12}.

Despite RBMs' popularity, rigorous mixing time bounds of the alternating-scan sampler are rather sparse.
The only available results require either bounded interaction strengths \cite{Tos16} or special structures \cite{KQWW26}.
The lack of good bounds is perhaps for a good reason.
Via an equivalent formulation of anti-ferromagnetic two-spin systems, 
when parameters cross the critical threshold,
the mixing time in negatively weighted RBMs is exponentially large,
and in fact, in this case sampling and approximate counting are \NP-hard \cite{Sly10,SS14,GSV16}.
On the other hand, the contrastive divergence method \cite{Hin12} in practise typically runs the alternating-scan sampler for a constant number of rounds.
In this paper, we show a polylogarithmic mixing time bound for the alternating-scan sampler on positively weighted RBMs,
bypassing the bounded interaction strengths requirement and complementing the hardness for the negative weight case.

Next we introduce our main result more precisely.
A \emph{Boltzmann machine} \cite{AHS85} with a set $V$ of variables of size $n$
is specified by an $n$-by-$n$ symmetric interaction matrix $W=\{w_{uv}\}_{u,v\in V}$ and variable weights $\theta=\{\theta_v\}_{v\in V}$.
A configuration $\sigma:V\rightarrow \{0,1\}$ is associated with the Hamiltonian or the energy function:
\begin{align}\label{eqn:BM-Hamiltonian}
  E(\sigma) \defeq \sum_{u,v \in V} w_{uv} \sigma_u\sigma_v + \sum_{v\in V}\theta_v \sigma_v.
\end{align}
Without loss of generality we may assume that the diagonal entries of $W$ are all $0$.
The Gibbs distribution $\mu$ is defined as $\mu(\sigma) = \frac{e^{E(\sigma)}}{Z}$,
where $Z\defeq\sum_{\sigma\in\{0,1\}^V} e^{E(\sigma)}$ is the normalizing constant, namely the partition function.

A \emph{restricted Boltzmann machine} (RBM) \cite{Smo86} is one where the variables can be partitioned into two parts $V=V_0\uplus V_1$ (the visible and the hidden layers) such that $w_{uv}=0$ whenever $u\in V_0$ and $v\in V_1$.
We may also view an RBM as over a bipartite graph where the edge set $E$ represents nonzero interaction weights.

A popular algorithm to sample from RBMs is the aforementioned alternating-scan sampler \cite{Hin12},
which is a systematic scan variant of the Gibbs sampler where we scan the two partitions in order.
Starting from an arbitrary configuration $X \in \{0,1\}^n$.
For any $t \geq 1$, in the $t$-th step, 
it updates the current configuration $X$ as follows: 
pick the part $V_i$ with index $i = (t \mod 2)$ and resample the configuration on $V_i$ conditional on the current configuration of the other part $V_{1-i}$.
More formally, at step $t$, 
\begin{enumerate}
    \item pick the part $V_i$ with index $i = (t \mod 2)$;
    \item resample $X_{V_i} \sim \mu_{V_i}^{X(V_{1-i})}$, where $\mu_{V_i}^{X(V_{1-i})}$ is the marginal distribution of all variables in the part $V_i$ induced by $\mu$ conditioned on the configuration $X(V_{1-i})$ on the other part $V_{1-i}$;
\end{enumerate}
The mixing time of a Markov chain is defined as the number of steps until the configuration $X$ is close to the stationary distribution $\mu$ in total variation distance. 
Formally, let $P$ be the transition matrix of the Markov chain. 
Then, the mixing time is defined as
\begin{align}\label{eq:mixing-time}
    \forall \epsilon > 0, \quad t_{\textnormal{mix}}(\epsilon) = \max_{X_0 \in \{0,1\}^V}\min\left\{t \geq 0: \DTV{{P}^t(X_0,\cdot)}{\mu} < \epsilon\right\},
\end{align}
where $\DTV{\nu}{\mu}= \frac{1}{2}\sum_{x \in \{0,1\}^V} \abs{\nu(x) - \mu(x)}$ denotes the total variation distance and $X_0$ is called the starting configuration or state.

Now we can state our main result.

\begin{theorem}\label{thm:alternating-scan-mixing-BM}
  Let $c > 0$ be an arbitrary constant. 
  For any restricted Boltzmann machine $(W,\theta)$ with $n$ variables, 
  if for all $u,v$, either $w_{uv} \geq c$ or $w_{uv}=0$, and for all $v \in V$, $\theta_v \ge 0$, 
  then the alternating-scan sampler over the Gibbs distribution $\mu$ of the RBM has mixing time at most $O((\log n)^{C} \log \frac{1}{\epsilon})$, where $C = C(c) > 0$ is a constant depending on $c$.
\end{theorem}

We note that the lower bound $c>0$ is to avoid cases where, for example, some $w_{uv}=1/n$.
Certain technical conditions we rely on would break in such a case.
We believe that this is an artifact of our proof, and the theorem should hold with $c=0$.
On the other hand, the main strength of \Cref{thm:alternating-scan-mixing-BM} is that we do not need to assume any upper bound on $w_{uv}$'s.

Previously, Tosh \cite{Tos16} showed that the alternating-scan sampler mixes in logarithmic time when $\norm{1}{W}\norm{1}{W^{\texttt{t}}}< 4$ via a one-step coupling,
where $\norm{1}{\cdot}$ denotes the $1$-norm of matrices.
Kwon, Qin, Wang, and Wei \cite{KQWW26} considered the setting where $W_{uv}=c/n$ for any $u\in V_0$ and $v\in V_1$ for some $c$.
They obtained logarithmic mixing time as long as $c>-5.87$ via a drift and contraction coupling technique.
In contrast, \Cref{thm:alternating-scan-mixing-BM} does not have any upper bound on the interaction strength or assumption on the structure, and the proof technique is a significant departure from these two results.

Alternatively, rigorous efficient algorithms for positively weighted RBMs can be obtained via an equivalent formulation of ferromagnetic two-spin systems \cite{GoldbergJP03,LiuLZ14,GuoL18,GLL20}.
\Cref{thm:alternating-scan-mixing-BM} is also proved via this connection, so we will explain it next.

\subsection{Ferromagnetic two-spin systems}

Boltzmann machines are a special case of the so-called two-spin systems.
Let $G=(V,E)$ be a graph. 
For each edge $e \in E$, let $\beta_e,\gamma_e > 0$ be the edge activity at $e$.
For each vertex $v \in V$, let $\lambda_v \leq 1$ be the external field at $v$.
A two-spin system $(G,(\beta_e,\gamma_e)_{e \in E},(\lambda_v)_{v \in V})$ defines a Gibbs distribution $\mu$ over $\Omega =\{0,1\}^V$ such that 
\begin{align*}
    \forall \sigma \in \{0,1\}^V,\quad \mu(\sigma) \propto \prod_{v \in V: \sigma_v = 0}\lambda_v \prod_{\{u,v\}\in E: \sigma_u = \sigma_v = 0} \beta_e \prod_{\{u,v\}\in E: \sigma_u = \sigma_v = 1} \gamma_e.
\end{align*} A two-spin system is said to be ferromagnetic if $\beta_e\gamma_e \ge 1$ for all $e \in E$.

Positively weighted Boltzmann machines with parameters $(W,\theta)$ can be viewed as ferromagnetic two-spin systems over the complete graph via the following reparameterisation:
\begin{align*}
  \forall v\in V,\ \lambda_v = \exp(-\theta_v)
  \text{ and }\forall u,v\in V,\ \beta_{uv} = 1, \gamma_{uv} = \exp(w_{uv}).
\end{align*}
We may also remove edges with zero weights.
This way, restricted Boltzmann machines become ferromagnetic two-spin systems defined over bipartite graphs.

We consider the following family of positively weighted restricted Boltzmann machines.
\begin{definition}[$(\gamma,\lambda)$-RBM]\label{def:gamma-lambda-rbm}
    Let $\gamma > 1$ and $\lambda > 0$ be two constants.
    A ferromagnetic two-spin system $(G,(\beta_e,\gamma_e)_{e \in E},(\lambda_v)_{v \in V})$ on a bipartite graph $G=(V,E)$ is said to be a $(\gamma,\lambda)$-RBM 
    if $\lambda_v < \lambda$ for all $v \in V$ and $\beta_e = 1$, $\gamma_e \geq \gamma$ for all $e \in E$.
  \end{definition}
Restricted Boltzmann machines in \Cref{thm:alternating-scan-mixing-BM} are special cases of $(\gamma,\lambda)$-RBM in \Cref{def:ferromagnetic-two-spin-system} with $\beta = 1$, $\gamma = \exp(c)$, and $\lambda = 1+\epsilon$ for an arbitrarily small $\epsilon>0$. 
It is important that here $\gamma$ and $\lambda$ are constants, and we do not need to assume any of the $\gamma_e$, or $\lambda_v$ to be constants.
\Cref{thm:alternating-scan-mixing-BM} is in fact implied by the following more general result for the mixing time of the alternating-scan sampler on ferromagnetic two-spin systems in bipartite graphs, where $\lambda$ is allowed to take a larger value. 
\begin{theorem}\label{thm:alternating-scan-mixing-ferromagnetic-two-spin-system}
    Let $\gamma > 1$ and $\lambda < \sqrt{\gamma}$ be two constants.
    For any $(\gamma,\lambda)$-RBM over a bipartite graph with $n$ vertices, 
    the alternating-scan sampler on the Gibbs distribution has mixing time at most $O((\log n)^{C} \log \frac{1}{\epsilon})$, where $C = C(\gamma,\lambda) > 0$ is a constant depending on $(\gamma,\lambda)$.
\end{theorem}

For more general ferromagnetic two-spin systems,
we mainly consider the following families.
\begin{definition}[$(\beta,\gamma,\lambda)$-ferromagnetic two-spin systems]\label{def:ferromagnetic-two-spin-system}
  Let $\beta$, $\gamma$, and $\lambda$ be three constants such that $\beta \leq 1 < \gamma$, $\beta\gamma > 1$, and $\lambda > 0$.
  A ferromagnetic two-spin system $(G,(\beta_e,\gamma_e)_{e \in E},(\lambda_v)_{v \in V})$ is said to be a $(\beta,\gamma,\lambda)$-ferromagnetic two-spin system 
  if $\lambda_v < \lambda$ for all $v \in V$ and $\beta_e \leq \beta \leq 1 < \gamma \leq \gamma_e$, $\beta\gamma \geq \beta_e \gamma_e > 1$ for all $e \in E$.
\end{definition}

Note that $(\gamma,\lambda)$-RBMs and $(1,\gamma,\lambda)$-ferromagnetic two-spin systems are incomparable. On one hand, a $(\gamma,\lambda)$-RBM is not necessarily a $(1,\gamma,\lambda)$-ferromagnetic two-spin system, since the RBM definition does not impose $\beta_e\gamma_e \le \beta\gamma$, and $\beta_e\gamma_e = \gamma_e$ may exceed $\gamma$ for some $e\in E$. On the other hand, a $(1,\gamma,\lambda)$-ferromagnetic two-spin system is not necessarily a $(\gamma,\lambda)$-RBM, since $\beta_e$ is allowed to take any value in $[0,1]$ rather than only $\beta_e = 1$.


In addition to the alternating-scan sampler, we also analyse Glauber dynamics,
which is another fundamental Markov chain to sample from Gibbs distributions.
Starting from an arbitrary configuration $X \in \{0,1\}^V$, in each step, the Glauber dynamics updates the current configuration $X$ as follows:
\begin{itemize}
    \item pick a vertex $v$ uniformly at random from $V$;
    \item resample $X_v \sim \mu_{v}^{X(V \setminus \{v\})}$, where $\mu_{v}^{X(V \setminus \{v\})}$ is the marginal distribution on $v$ induced by $\mu$ conditioned on the configuration $X(V \setminus \{v\})$ on other variables $V \setminus \{v\}$;
\end{itemize}
We show that under the same conditions as in \Cref{thm:alternating-scan-mixing-ferromagnetic-two-spin-system},
Glauber dynamics mixes in near-linear time.

\begin{theorem}\label{thm:glauber-mixing-2}
  Let $\beta,\gamma,\lambda > 0$ be three constants such that $\beta \leq 1 < \gamma$, $\beta\gamma > 1$, and $\lambda < \lambda_0 \defeq \sqrt{\gamma/\beta}$. 
    For any $(\beta,\gamma,\lambda)$-ferromagnetic two-spin system with $n$ vertices, 
    the Glauber dynamics on the Gibbs distribution has mixing time at most $n (\log n)^{C} \log \frac{1}{\epsilon}$, where $C = C(\beta,\gamma,\lambda) > 0$ is a constant depending on $(\beta,\gamma,\lambda)$.
\end{theorem}

The same threshold $\lambda_0 = \sqrt{\gamma/\beta}$ also appeared in~\cite{GoldbergJP03}, where the authors showed that for ferromagnetic two-spin systems with uniform parameters $\lambda_v = \lambda$, $\beta_e = \beta$, and $\gamma_e = \gamma$, there exists a polynomial-time sampling algorithm if $\beta \gamma > 1$ and $\lambda \leq \lambda_0$. 
The condition for a polynomial-time sampling algorithm was later shown to be $\lambda \leq {\gamma}/{\beta}$ by~\cite{LiuLZ14}.
Both algorithms are obtained by reducing the problem of sampling from ferromagnetic two-spin systems to that of sampling from a \emph{ferromagnetic Ising model} with consistent external fields. Specifically, the resulting Ising model is a two-spin system $(G,(\beta_e,\gamma_e)_{e \in E},(\lambda_v)_{v \in V})$ such that every edge has interaction strength $\beta_e = \gamma_e = \sqrt{\beta \gamma} > 1$ and every vertex has external field $\lambda_v \leq 1$.
Jerrum and Sinclair \cite{JerrumS93} gave the first polynomial-time sampling algorithm to this Ising model.
After the reduction in \cite{LiuLZ14}, there is a constant gap between the external field $\lambda$ and $1$.
In this case, the best sampling algorithm runs in near-linear time as well \cite{CZ23}, via yet another connection \cite{ES88,GuoJ18,FengGW23} to the \emph{random cluster model} \cite{FK72}.

Our results in \Cref{thm:alternating-scan-mixing-ferromagnetic-two-spin-system} and \Cref{thm:glauber-mixing-2} are the first near-optimal mixing results for the alternating-scan sampler and Glauber dynamics on ferromagnetic two-spin systems with $\lambda < \lambda_0$, whereas all previous algorithms rely on a reduction to sampling from other models. 
From a technical perspective, our approach is completely different from the reduction technique used in~\cite{GoldbergJP03, LiuLZ14}. 
We develop a unified framework for analyzing the mixing of a family of \emph{heat-bath} and \emph{systematic scan} dynamics on ferromagnetic two-spin systems, which covers the alternating-scan sampler and Glauber dynamics as special cases.  
We give a proof overview in \Cref{sec:proof-overview}.

Another advantage of the direct mixing time bound in \Cref{thm:glauber-mixing-2} is that, unlike previous results, it allows us to extend our mixing time analysis beyond $\lambda_0$ to a larger threshold
\begin{align*}
  \lambda_c(\beta,\gamma) \defeq (\gamma/\beta)^{\frac{\sqrt{\beta \gamma}}{\sqrt{\beta \gamma}-1}} \geq \lambda_0(\beta,\gamma),
\end{align*}
which was previously identified as the potentially critical threshold for ferromagnetic 2-spin systems \cite{GuoL18}.\footnote{Roughly speaking, up to an integral gap, systems above this threshold are \BIS-hard \cite{LiuLZ14}, where \BIS{} is conjectured to be computationally hard \cite{DGGJ04}. }
In fact, Guo and Lu \cite{GuoL18} designed efficient sampling and approximate counting algorithms for ferromagnetic 2-spin systems below this threshold via correlation decay \cite{weitz2006counting}.
However, their algorithms run in time $O(n^C)$ where $C$ is a large constant depending on $(\beta,\gamma,\lambda)$.
Later, Guo, Liu, and Lu \cite{GLL20} designed another algorithm based on the zeros of polynomials method \cite{Bar16,PR17},
which works for all $\beta,\gamma$ such that $\beta\gamma>1$ but with a lower threshold for $\lambda$\footnote{Their threshold is roughly $\sqrt{\lambda_c}$.} and requires bounded degree graphs.
In any case, it has a similar $O(n^C)$ running time.
Our next result improves the exponent in the running times for sampling and approximate counting to absolute constants.
For sampling, our time bound is $\widetilde{O}(n^2)$.
\begin{theorem}\label{thm:glauber-mixing-1}
        Let $\beta,\gamma,\lambda > 0$ be three constants such that $\beta \leq 1 < \gamma$, $\beta\gamma > 1$, and $\lambda < \lambda_c \defeq (\gamma/\beta)^{\frac{\sqrt{\beta \gamma}}{\sqrt{\beta \gamma}-1}} $.  There exists a constant $C = C(\beta,\gamma,\lambda) > 0$ such that for any $(\beta,\gamma,\lambda)$-ferromagnetic two-spin system with $n$ vertices, the Glauber dynamics on the Gibbs distribution $\mu$ has mixing time
        \begin{itemize}
            \item at most $n^2 \cdot (\log n)^{C} \cdot \log \frac{1}{\epsilon}$ starting from the all-1 configuration;
            \item at most $n^3 \cdot (\log n)^{C} \cdot \log \frac{1}{\epsilon}$ starting from an arbitrary configuration.
        \end{itemize}
\end{theorem}

\begin{remark}
In the proof of \Cref{thm:glauber-mixing-1}, we show that the spectral gap of the Glauber dynamics is $\textsf{gap} = \widetilde{\Omega}(n^{-1})$, which is optimal.
Since $\mu(\boldsymbol{1}) = 2^{-\Omega(n)}$ for the all-1 configuration, this yields the upper bound $O\left(\frac{1}{\textsf{gap}} \log \frac{1}{\epsilon \mu(\boldsymbol{1})}\right)$
on the mixing time from the all-1 starting configuration. For the mixing time from an arbitrary starting configuration, the standard approach is to use the bound
$O\left(\frac{1}{\textsf{gap}} \log \frac{1}{\epsilon \mu_{\min}}\right)$,
where $\mu_{\min} = \min_{x \in \{0,1\}^V} \mu(x)$. However, for spin systems in \Cref{def:ferromagnetic-two-spin-system}, 
the parameters $\lambda_v$ and $\beta_e$ may be arbitrarily small, while $\gamma_e$ may be arbitrarily large, resulting in a potentially arbitrarily small $\mu_{\min}$. 
We resolve this issue by showing that the Glauber dynamics quickly reaches a warm-start configuration with high probability, and then bounding the mixing time from such a warm-start configuration (see \Cref{lemma:warm-start-configuration-general}). 
We remark that even if all parameters $\lambda_v, \beta_e, \gamma_e$ are assumed to be constants, $\mu_{\min}$ can still be as small as $\exp(-O(n^2))$. 
The reason is that the graph can be very dense and contain an $\Omega(n^2)$ number of edges. 
\end{remark}

Our result is also the first polynomial mixing time bound for Glauber dynamics on ferromagnetic two-spin systems with $\lambda < \lambda_c$ in general graphs. All previous polynomial mixing time results work only on \emph{bounded degree} graphs.
Let $\Delta$ denote the maximum degree of the graph.
Chen, Liu, and Vigoda first proved $n^{{e^{O(\Delta)}}}$ mixing time bound for Glauber dynamics~\cite{CLV23}, and later they improved the bound to $e^{e^{O(\Delta)}} n \log n$ ~\cite{CLV21}.
In contrast, our \Cref{thm:glauber-mixing-1} does not depend on $\Delta$.

\Cref{thm:glauber-mixing-1} is proved by combining \Cref{thm:glauber-mixing-2} with a mixing time boosting technique developed by Chen, Feng, Yin, and Zhang~\cite{CFYZ21}. Roughly speaking, by verifying a certain spectral independence condition \cite{ALO24} for ferromagnetic two-spin systems when $\lambda < \lambda_c$, we can reduce the analysis to the case $\lambda < \lambda_0$, which is handled by \Cref{thm:glauber-mixing-2}.
It is important that \Cref{thm:glauber-mixing-2} provides a direct mixing time bound rather than a reduction based sampling algorithm as in~\cite{GoldbergJP03, LiuLZ14}; otherwise, the mixing time boosting technique would not be applicable. The detailed proof is given in \Cref{sec:glauber-mixing-1}.

As mentioned before, we believe that the lower bound $c>0$ requirement can be removed in \Cref{thm:alternating-scan-mixing-BM},
but some new ideas are required to handle the case where, for example, some $w_{uv}=1/n$.
Another interesting open problem is to prove a near-optimal $\widetilde{O}(n)$ mixing time bound for ferromagnetic two-spin systems when $\lambda < \lambda_c$. Due to technical obstacles (see \Cref{sec:proof-overview}), we cannot directly extend the analysis of \Cref{thm:glauber-mixing-2} to this regime. A possible alternative is to use the refined mixing-time boosting techniques developed in~\cite{Chen0YZ22,ChenE22,FY26}. 
However, this approach requires a stronger \emph{entropic independence}~\cite{AJKPV22} condition, which is not known to hold for the class of ferromagnetic two-spin systems studied here.
More broadly, our proof framework applies to general ferromagnetic two-spin systems,
for which there is still a big gap between the known algorithmic~\cite{GuoL18,GLL20,shao2021contraction} and the hardness threshold \cite{LiuLZ14},
especially when $\beta,\gamma>1$.
In that case, worst-case correlation decay results, such as those in \cite{GuoL18}, no longer hold. 
We believe that our ``typical-case'' SSM (more detail in \Cref{sec:proof-overview}) is the first step on the right direction.

\section{Proof overview}\label{sec:proof-overview}
We give a proof overview for the mixing time of Glauber dynamics on ferromagnetic two-spin systems. For the simplicity of the overview, consider a ferromagnetic two-spin system $\mu$ defined on a graph $G=(V,E)$ with unified parameters, where $\lambda_v = \lambda$ for all $v \in V$ and $\beta_e = \beta, \gamma_e = \gamma$ for all $e \in E$ for constants $\lambda,\beta,\gamma$. We outline the proof of $n \cdot \mathrm{polylog}(n)$ mixing time bound in \Cref{thm:glauber-mixing-2} when $\lambda < \lambda_0 = \sqrt{\gamma/\beta}$. Other results can be proved as follows.
\begin{itemize}
    \item The proof technique of \Cref{thm:glauber-mixing-2} can be generalized to the alternating-scan sampler in \Cref{thm:alternating-scan-mixing-ferromagnetic-two-spin-system}.
    \item The mixing result in \Cref{thm:glauber-mixing-1} when $\lambda < \lambda_c$ can be proved by combining the mixing result in \Cref{thm:glauber-mixing-2} with the existing results in~\cite{CFYZ21}.
\end{itemize}

\subsection{All-to-one influence bound}
Let $\mu$ over $\{0,1\}^V$ be a Gibbs distribution defined on variable set $V$. For any two variables $u,v \in V$, the influence from $u$ on $v$ is defined as
\begin{align*}
    \Psi(u,v) \defeq \Pr[X_v = 1 \mid X_u = 1] - \Pr[X_v = 1 \mid X_u = 0].
\end{align*}
Anari, Liu, and Oveis Gharan \cite{ALO24} showed that if the maximum eigenvalue of the influence matrix $\Psi$ is bounded by a constant, then the Glauber dynamics mixes in polynomial time. The maximum eigenvalue of the influence matrix is bounded by the \emph{all-to-one influence} $\max_{v \in V} \sum_{u \in V} \Psi(u,v)$.
We show in \Cref{thm:all-to-one-influence} that if $\lambda < \lambda_c$, then the all-to-one influence is $O(1)$. The proof is inspired by the analysis in~\cite{ALO24}, where they analysed the all-to-one influence of the hardcore model in the uniqueness regime. Here, we need to deal with the ferromagnetic spin system in general graph with possibly unbounded degree. We use the correlation decay technique developed in~\cite{GuoL18} to prove the bound.

The all-to-one influence only gives an $n^{O(C)}$ mixing time bound, where the influence bound $C$ can be a very large constant. 
However, this is still useful in getting the local mixing bounds we need later.
To obtain our $n \cdot \mathrm{polylog}(n)$ mixing result in general graphs, we use a local mixing to global mixing argument based on the \emph{aggregate strong spatial mixing} (\emph{ASSM}) property.

\subsection{Mixing from typical-case ASSM}
A ferromagnetic two-spin system is a monotone system. 
Mossel and Sly~\cite{MS13} showed that the mixing of Glauber dynamics on monotone systems can be proved via the ASSM property. 
Let $v \in V$ and $S_v \subseteq V$  a subset of vertices containing $v$. Let $\partial S_v$ be the outer \emph{boundary} of $S_v$, which is the set of vertices not in $S_v$ but adjacent to $S_v$. Define the \emph{influence} of $u$ on $v$ by
\begin{align}\label{eq:influence-intro}
  \widehat{a}_u \defeq \max_{\sigma \in \{0,1\}^{\partial S_v}} \DTV{\mu^{\sigma^{u \gets 0}}_v}{\mu^{\sigma^{u \gets 1}}_v},
\end{align}
where $\sigma^{u \gets c}$ denotes the configuration on $\partial S_v$ obtained from $\sigma$ by changing the value of $u$ to $c$.
Mossel and Sly showed that if the ASSM property $\sum_{u \in \partial S_v} \widehat{a}_u \leq \frac{1}{20}$ holds and the mixing time of Glauber dynamics on the conditional distribution $\mu^{\sigma}_{S_v}$ is at most $T_{\text{local}}$ for any $\sigma \in \{0,1\}^{\partial S_v}$, then the mixing time of Glauber dynamics on $\mu$ is at most $O(T_{\text{local}} \cdot n \log n \cdot \max_{v \in V} \log |S_v \cup \partial S_v|)$.
Their result works for ferromagnetic two-spin systems on graphs with \emph{bounded degrees}. 
For the Ising model in the uniqueness regime, the ASSM property can be verified if the region $S_v$ is a ball centered at $v$ with radius $\ell_0 = O(1)$~\cite{MS13}.
Since the degrees are bounded, $|S_v \cup \partial S_v|$ is a constant, implying that $T_{\text{local}} = O(1)$. The overall mixing time of the Glauber dynamics on $\mu$ is $O(n \log n)$.

However, what we consider are general graphs with possibly unbounded degrees. Consider a star centered at $v$. If we choose $S_v=\{v\}$, then ASSM does not necessarily hold. If we choose $S_v$ as a ball centered at $v$ with radius $1$, the resulting $S_v$ is the whole $V$, and bounding local mixing $T_{\text{local}}$ is the same as bounding the mixing time of Glauber dynamics on $\mu$. To resolve these issues, we introduce a weaker version of the ASSM property. For each vertex $v \in V$, we algorithmically choose a region $S_v$ and also define a set of good boundary configurations $\Omega_{\partial S_v} \subseteq \{0,1\}^{\partial S_v}$. The specific choice of $S_v$ and $\Omega_{\partial S_v}$ will be given in later. We define a new influence bound $a_u$ as
\begin{align*}
  a_u \defeq \max_{\sigma \in \Omega_{\partial S_v}} \DTV{\mu^{\sigma^{u \gets 0}}_v}{\mu^{\sigma^{u \gets 1}}_v},
\end{align*}
Compared with~\eqref{eq:influence-intro}, the new influence considers only ``typical'' boundary conditions, namely those from $\Omega_{\partial S_v}$, on $\partial S_v$. Using the monotone coupling technique, we show that the mixing time of Glauber dynamics on $\mu$ is at most $O(T_{\textnormal{burn-in}} + T_{\textnormal{local}} \cdot n \log n \cdot \max_{v \in V} \log |S_v \cup \partial S_v|) = O(T_{\textnormal{burn-in}} + T_{\textnormal{local}} \cdot n \log^2 n)$ as long as the following conditions holds for two parameters $T_{\textnormal{burn-in}}$ and $T_{\textnormal{local}}$.
\begin{itemize}
  \item For the Glauber dynamics $(X_t)_{t \geq 0}$ on $\mu$,
    starting from an arbitrary $X_0 \in \{0,1\}^V$, for any $t \geq T_{\textnormal{burn-in}}$, any $v \in V$, with probability at least $1 - \frac{1}{\mathrm{poly}(n)}$, it holds that $X_t(\partial S_v) \in \Omega_{\partial S_v}$.
    \item ASSM holds for typical boundary conditions: $\sum_{u \in \partial S_v} a_u \leq \frac{1}{20}$ for all $v \in V$.
    \item For any vertex $v \in V$ and any $\sigma \in \{0,1\}^{\partial S_v}$, the Glauber dynamics on $\mu^\sigma_{S_v}$ has mixing time $T_{\textnormal{local}}$.
\end{itemize}

Compared with the result of Mossel and Sly, the key advantage is that we only require the ASSM property to hold under a typical-boundary condition after burn-in, while they require the ASSM property to hold for worst-case boundary conditions. For the star graph centered at $v$, we can simply take $S_v = \{v\}$ and let $\Omega_{\partial S_v}$ contains all configurations on neighbors of $v$ such that at least a constant fraction of them are assigned $1$. Note that the parameter setting is $\beta \leq 1 < \gamma$. If $\Omega(n)$ neighbors of $v$ are in state 1, then since $\gamma > 1$, the value on $v$ is almost fixed to be 1, so we can bound the sum of the influences. However, in the original definition of Mossel and Sly, the maximum influence for $w$ to $v$ is achieved when all other vertices are $0$. In this case, if $\beta = 1$, then each $a_u = \Omega(1)$, so the total influence is $\Omega(n)$.

\subsection{Typical-case ASSM for ferro spin systems}\label{sec:T-ASSM-overview}
To carry out the ideas in the previous section, 
we need to carefully choose the region $S_v$ and the set of good boundary configurations $\Omega_{\partial S_v}$ so that all the above conditions hold, which is the most technical part of the proof. 
We will guarantee that $|S_v| = \mathrm{polylog}(n)$.
Then, for the local mixing bound $T_{\textnormal{local}}$, the conditionally distribution is defined on $N = \mathrm{polylog}(n)$ vertices. Using the all-to-one influence bound and the result in~\cite{ALO24}, we have $T_{\textnormal{local}} \leq N^{O(1)} = \mathrm{polylog}(n)$.

We next give a detailed construction of the region $S_v$. To illustrate the idea, let us first focus on a special case when the graph $G$ is a tree. We run a DFS starting from the root $v$. Suppose the DFS procedure visits a vertex $w$. We first add $w$ into the region $S_v$. Next, let $u_0 = v, u_1, \ldots, u_k = w$ be the path from $v$ to $w$ in the tree. For each $u_i$, let $d_i$ denote the number of children of $u_i$ in the tree rooted at $v$. 
\begin{itemize}
    \item If $\sum_{i=1}^k d_i < D_1 = O(\log \log n)$, we recursively do the DFS on all children of $w$;
    \item If $\sum_{i=1}^k d_i \geq D_1$, we will \emph{not} recursively do the DFS on any child of $w$. 
      Instead, if the number of children of $w$ is less than $D_2 = (\log n)^3$, we add all these children into the region $S_v$ and terminate the exploration in this branch.
      Otherwise, we stop at $w$.
\end{itemize}
Overall, the DFS procedure will construct a region $S_v$, where the induced subgraph $T_{S_v} = G[S_v]$ is a subtree rooted at $v$. For any vertex $w \in S_v$, in the subtree $G[S_v]$, we can upper bound the degree sum of all vertices on the path from $v$ to $w$. Using this property, we can show that $|S_v| = \mathrm{polylog}(n)$.  

Let $\partial S_v$ be the outer boundary of $S_v$.
Define $\Omega_{\partial S_v}$ as the set of all boundary configurations $\sigma \in \{0,1\}^{\partial S_v}$ such that for any vertex $w \in S_v$ with $K$ neighbors in $\partial S_v$, if $K \geq D_2/3$, then at least $K/\log n = \Omega((\log n)^2)$ neighbors are assigned $1$ in $\sigma$. In other words, if $w$ has many neighbors in $\partial S_v$, then a significant proportion of them are assigned $1$. Since $\beta \leq 1$, when Glauber dynamics updates a vertex $u$, with a constant probability, the value on $u$ is updated to 1. After running Glauber dynamics for $T_{\textnormal{burn-in}} = O(n \log n)$ steps, a simple coupon collector and Chernoff bound argument shows that with high probability, the configuration on $\partial S_v$ is in $\Omega_{\partial S_v}$.

We next bound the sum $\sum_{u \in \partial S_v} a_u$. Fix a vertex $u \in \partial S_v$.
We first explain why a single influence $a_u$ is small. Then we give some high level ideas on how to bound the sum of the influences. To analyze the influence $a_u$, we need to consider a spin system with pinnings defined on the induced subgraph $G[S_v \cup \partial S_v]$. Using the self-reducibility property of ferromagnetic two-spin systems, we can remove the pinning and analyze a spin system on $T' = G[S_v \cup \{u\}]$ with some effective external fields on the inner boundary of $S_v$. Then, $a_u$ is the one-to-one influence from $u$ to $v$ in $T'$. Guo and Lu~\cite{GuoL18} showed the following \emph{computationally efficient} correlation decay result. Let $v_0 = v, v_1, \ldots, v_k = u$ be the path from $v$ to $u$ in the tree $T'$. Let $d'_i$ be the number of children of $v_i$ in $T'$. Then
\begin{align*}
   a_u \leq C_1  \exp\left(-\sum_{i=1}^{k-1}d'_i / C_2\right),
\end{align*}
for some sufficiently large constants $C_1, C_2 > 0$. Let $d_1,d_2,\ldots,d_{k}$ be the number of children of $v_i$ in the tree $G$ rooted at $v$. By the definition of $T' = G[S_v \cup \{u\}]$ and the construction of $S_v$, we have $d'_i = d_i$ for $1\leq i \leq k -2$ and $d'_{k-1} = 1$. 
Depending on how $v_{k-1}$ is added to $S_v$, there are two cases.
\begin{itemize}
    \item The vertex $v_{k-1}$ is added to $S_v$ because the DFS stops at the vertex $v_{k-2}$ and $v_{k-2}$ has less than $ D_2$ children. 
      However, stopping at $v_{k-2}$ means that $\sum_{i=1}^{k-2} d_i \geq D_1$.
      Thus $\sum_{i=1}^{k-1} d_i' \geq \sum_{i=1}^{k-2} d_i \geq D_1 = \Omega(\log \log n)$ and then $a_u$ is small.
    \item The vertex $v_{k-1}$ is added to $S_v$ because the DFS stops at the vertex $v_{k-1}$ and $v_{k-1}$ has at least $D_2$ children.
      Now, although $\sum_{i=1}^{k-1}d_i \geq D_1$, we have no lower bound on $\sum_{i=1}^{k-1} d_i'$ because $d'_{k-1}$ can be much smaller than $d_{k-1}$. However, in this case $v_{k-1}$ has many neighbors in $\partial S_v$ because $d_{k-1} \geq D_2$. By the definition of $\Omega_{\partial S_v}$, many neighbors of $v_{k-1}$ are assigned 1. Since the spin system is ferromagnetic, the value on $v_{k-1}$ is almost  fixed to be 1. The vertex $v_{k-1}$ \emph{blocks} the influence from $u$ to $v$ and $a_u$ is small.
\end{itemize}


To bound the sum of influences $\sum_{u \in \partial S_v} a_u$, we decompose the sum as $\sum_{k \geq 1} \sum_{u \in L_k(v) \cap \partial S_v}a_u \defeq \sum_{k \geq 1} \text{Inf}(k)$, where $L_k(v)$ denotes the set of vertices at level $k$ in the tree $G$ rooted at $v$ and $\text{Inf}(k)$ is the sum of influences at level $k$. 
We then use correlation decay analysis to bound $\text{Inf}(k)$ for each level $k$.
Compared to the all-to-one influence bound, which is proved using a similar methodology, a new challenge is the presence of the boundary conditions in $a_u$'s.
For two vertices $u$ and $u'$ in $\partial S_v$ at the same level $k$, 
the boundary conditions to achieve $a_{u}$ and $a_{u'}$ may be very different and some of the disagreements may be very close to the root $v$.
This makes the correlation decay analysis difficult to carry out.

To resolve this issue, we showed that, roughly speaking, if $\lambda < \lambda_0$ (for the definition of $\lambda_0$, recall \Cref{thm:alternating-scan-mixing-ferromagnetic-two-spin-system}), then we can assume that 
\begin{align}\label{eq:consistent-pinnings}
  \forall w \in L_{<k}(v) \cap \partial S_v, \quad \sigma(w) = \tau(w), \text{ where } L_{<k}(v) \defeq \cup_{j=1}^{k-1} L_j(v),
\end{align}
where $\sigma$ and $\tau$ are two boundary conditions that achieve $a_{u}$ and $a_{u'}$.
Hence, when analysing $\text{Inf}(k)$, we can assume all pinnings above level $k$ are consistent for all $u \in L_k(v) \cap \partial S_v$.
The disagreements only appear after level $k$.
Details of this argument are in \Cref{lemma:monotone-potential}.
With its help we then can apply the correlation decay analysis to establish ASSM. 
We remark that~\eqref{eq:consistent-pinnings} is the only place where we need to use the stronger condition $\lambda < \lambda_0$ in instead of $\lambda < \lambda_c$. 
If one can verify the typical-case ASSM property when $\lambda < \lambda_c$, then the above analysis framework gives an improved $\tilde{O}(n)$ mixing time to \Cref{thm:glauber-mixing-1}.


So far, all the discussion above assumes the graph $G$ itself is a tree. For a general graph $G$, the set $S_v$ can be constructed as follows.
We first construct a \emph{self-avoiding walk} (SAW) tree $T_{\text{SAW}}$ of the graph $G$ rooted at $v$ (a tree enumerating all self-avoiding walks from $v$ in graph $G$). Then, using the same construction as in the tree case, we construct the region $S_v^T$ for the SAW tree $T_{\text{SAW}}$ and then map all vertices in $S_v^T$ back to the original graph $G$ to obtain $S_v$. 
Details of this construction are in \Cref{sec:construct-region}.
Let $\partial S_v$ be the outer boundary of $S_v$ in $G$. The good boundary condition $\sigma \in \Omega_{\partial S_v}$ is defined similarly as above: if a vertex $w \in S_v$ has many neighbors in $\partial S_v$, then many of them are assigned 1 in $\sigma$.
To prove the ASSM property in $G$, we reduce the task to analyzing influences in the self-avoiding walk tree $T_{\text{SAW}}$. 
Using the ideas above, we show that every path in the SAW tree contributes a good decay of correlation, so that typical-case ASSM holds in general graphs.

\section{Preliminaries}

\subsection{Markov chain and mixing time}
Let $X_t$ be a Markov chain on a state space $\Omega$ with transition matrix $P$. 
We call a Markov chain \emph{irreducible} if for any two states $x,y \in \Omega$, there exists a positive integer $t$ such that $P^t(x,y) > 0$,
\emph{aperiodic} if for any $x \in \Omega$, $\gcd\{t \geq 1: P^t(x,x) > 0\} = 1$,
and \emph{reversible} with respect to a distribution $\mu$ if $\mu(x)P(x,y) = \mu(y)P(y,x)$ for all $x,y \in \Omega$. 
An irreducible, aperiodic, and reversible Markov chain has a unique stationary distribution $\mu$.
The mixing time is defined as
\begin{align*}
    t_{\textnormal{mix}}^P(\epsilon) = \max_{x \in \Omega} \min\{t \geq 0: \DTV{P^t(x,\cdot)}{\mu} < \epsilon\}.
\end{align*}
We often consider the mixing time when $\epsilon = 1/(4e)$ because of the following general bound
\begin{align}\label{eq:general-bound-on-mixing-time}
\forall \epsilon > 0, \quad t_{\textnormal{mix}}^P(\epsilon) \leq t_{\textnormal{mix}}^P\left(\frac{1}{4e}\right)\log \frac{1}{\epsilon}.
\end{align}

Let $\mu$ be a distribution over $\Omega = \{0,1\}^V$. 
Let $P$ be the Glauber dynamics on $\mu$. Then, the transition matrix $P$ is positive semi-definite with real non-negative eigenvalues $1 = \lambda_1 \geq \lambda_2 \geq \cdots \geq \lambda_{|\Omega|} \geq 0$. The spectral gap of $P$ is defined as $\gamma_{\text{GD}} = 1 - \lambda_2$. For any distribution $\nu$ over $\Omega$, it is well known that
\begin{align*}
D_{\chi^2}(\nu P \Vert \mu P) \leq \tp{1 - \gamma_{\text{GD}}}^t \cdot D_{\chi^2}(\nu \Vert \mu),
\end{align*}
where $D_{\chi^2}(\nu \Vert \mu) = \sum_{x \in \Omega} \frac{(\nu(x) - \mu(x))^2}{\mu(x)}$ is the chi-squared divergence between $\nu$ and $\mu$. The following relationship between the total variation distance and the chi-squared divergence holds:
\begin{align*}
\DTV{\nu}{\mu} \leq \sqrt{ D_{\chi^2}(\nu \Vert \mu)}.
\end{align*}
As a consequence, for the Glauber dynamics on $\mu$ starting from an arbitrary configuration $X_0 = \sigma$,
\begin{align*}
\DTV{X_t}{\mu} \leq \sqrt{ D_{\chi^2}(X_t \Vert \mu)} \leq \sqrt{ \tp{1 - \gamma_{\text{GD}}}^t \cdot D_{\chi^2}(X_0 \Vert \mu)} \leq \sqrt{ \tp{1 - \gamma_{\text{GD}}}^t \cdot \frac{1}{\mu(\sigma)}}.
\end{align*}
where $X_t$ is the distribution of the Glauber dynamics on $\mu$ after $t$ steps starting from $X_0$. Hence, the mixing time of the Glauber dynamics on $\mu$ is at most
\begin{align}\label{eq:mixing-time-glauber-gap}
t_{\textnormal{mix}}^P(\epsilon) \leq \frac{1}{\gamma_{\text{GD}}}\log \frac{1}{\epsilon^2 \mu(\sigma)}.
\end{align}
The ratio $\frac{1}{\gamma_{\text{GD}}}$ is called the relaxation time of the Glauber dynamics on $\mu$. For the other direction,
\begin{align}\label{eq:mixing-time-glauber-gap-other-direction}
\forall \epsilon > 0, \quad t_{\textnormal{mix}}^{P}(\epsilon) \geq \tp{\frac{1}{\gamma_{\text{GD}}} - 1} \log \frac{1}{2\epsilon}.
\end{align}

Next, consider a Gibbs distribution $\mu$ defined on a bipartite graph $G=(V_0,V_1,E)$. Let $Q$ be the alternating-scan sampler on $\mu$. Formally, let $P_0$ denote the transition matrix of updating the configuration on $V_0$ conditional on the current configuration of the other part $V_1$, and let $P_1$ denote the transition matrix of updating the configuration on $V_1$ conditional on the current configuration of the other part $V_0$. Then, the transition matrix $Q$ of the alternating-scan sampler is defined as
\begin{align*}
Q = P_1P_0.
\end{align*}
When $\mu$ is the Gibbs distribution of a restricted Boltzmann machine, the Markov chain $Q$ is irreducible, aperiodic, and has the unique stationary distribution $\mu$. However, $Q$ may not be reversible with respect to $\mu$. Let the multiplicative reversiblization be $R(Q) = QQ^*$, where $Q^*$ is defined by
\begin{align*}
Q^*(\sigma,\tau) = \frac{\mu(\tau)}{\mu(\sigma)} Q(\tau,\sigma).
\end{align*}
Then $R(Q)$ is reversible with respect to $\mu$.
Furthermore, all eigenvalues of $R(Q)$ are real and non-negative~\cite{fill1991eigenvalue}.
The relaxation time of the alternating-scan sampler is defined by
\begin{align*}
T_{\text{rel}}(Q) = \frac{1}{1 - \sqrt{1 - \gamma(R(Q))}},
\end{align*}
where $\gamma(R(Q)) = 1 - \lambda_2(R(Q))$, and $\lambda_2(R(Q))$ is the second largest eigenvalue of $R(Q)$.

\begin{proposition}[\text{\cite[Theorem 1]{GuoKZ18}}]\label{prop:relaxation-time-alternating-scan}
    For a RBM on a bipartite graph with Gibbs distribution $\mu$,
    \begin{align*}
     T_{\text{rel}}(Q) \leq \frac{2}{\gamma_{\text{GD}}},
    \end{align*}
where $\gamma_{\text{GD}}$ is the spectral gap of the Glauber dynamics on $\mu$.
\end{proposition}
\begin{remark}
Theorem 1 in \cite{GuoKZ18} considers the spectral gap of the lazy version of the Glauber dynamics on $\mu$, which is $\frac{1}{2}I + \frac{1}{2}P$, where $I$ is the identity matrix. Hence, we add a factor of 2 in the above proposition.
\end{remark}

The mixing time of the alternating-scan sampler on $\mu$ can be bounded by the following proposition.

\begin{proposition}[\text{\cite[Theorem 3]{GuoKZ18}}]\label{prop:mixing-time-alternating-scan}
For the alternating-scan sampler $Q$ on an RBM, starting from a configuration $\sigma \in \{0,1\}^V$, after running $Q$ for $T_{\text{rel}}(Q) \log \frac{4e^2}{\epsilon^2 \mu(\sigma)}$ steps, the total variation distance between the resulting distribution and the stationary distribution is at most $\epsilon$.
\end{proposition}
\begin{remark}
The mixing time upper bound stated in \cite[Theorem 3]{GuoKZ18} is $T_{\text{rel}}(Q) \log \frac{4e^2}{\mu_{\min}}$, 
where $\mu_{\min} \defeq \min_{\sigma \in \{0,1\}^V} \mu(\sigma)$ and they define the mixing time by setting $\epsilon=1/(2e)$. 
To get \Cref{prop:mixing-time-alternating-scan}, generalising from $1/(2e)$ to an arbitrary $\epsilon$ is straightforward,
and the proof of \cite[Theorem 2.1]{fill1991eigenvalue}, which the proof in \cite{GuoKZ18} is based on, already deals with $\mu(\sigma)$ instead of $\mu_{\min}$.
\end{remark}

\subsection{Self-reducibility}
Let $G = (V,E)$ be a graph. Let $\mu$ be the Gibbs distribution of a ferromagnetic two-spin system on $G$ with parameters $(\beta_e,\gamma_e)_{e \in E},(\lambda_v)_{v \in V}$.

Fix a subset $\Lambda \subseteq V$. Let $\sigma \in \{0,1\}^{V \setminus \Lambda}$ be a configuration on $V \setminus \Lambda$. We use $\mu^\sigma$ to denote the distribution of $X \sim \mu$ conditional on $X(\Lambda) = \sigma$. The pinning $\sigma$ induces a conditional distribution $\mu^\sigma_{\Lambda}$ on $\Lambda$ given $\sigma$. Note that $\mu^\sigma_{\Lambda}$ is a Gibbs distribution of a ferromagnetic two-spin system on $G[\Lambda]$ with edge activities $(\beta_e,\gamma_e)_{e \in G[\Lambda]}$. For all vertices $v \in \Lambda$, the vertex activity at $v$ is updated to $\lambda_v'=\lambda_v \prod_{e \in E_c} \beta_e \prod_{e \in E_1} \frac{1}{\gamma_e} \leq \lambda_v$, where $E_c$ is the set of edges $\{v,u\}$ for $u \in V \setminus \Lambda$ and $\sigma_u = c$ for $c \in \{0,1\}$.

\begin{observation}[Self-reducibility under pinning]\label{obs:self-reducibility}
  Let $\beta \leq 1 < \gamma$, $\beta\gamma > 1$ and $\lambda < \lambda_c(\beta,\gamma)$. For any $(\beta,\gamma,\lambda)$-ferromagnetic two-spin system with Gibbs distribution $\mu$ in $G=(V,E)$, for any pinning $\sigma$ on a subset $\Lambda \subseteq V$, $\mu^\sigma_\Lambda$ is also the Gibbs distribution of a $(\beta,\gamma,\lambda)$-ferromagnetic two-spin system on $G[\Lambda]$. 
  Similarly, for any $(\gamma,\lambda)$-RBM with Gibbs distribution $\mu$, $\mu^\sigma_\Lambda$ is also the Gibbs distribution of a $(\gamma,\lambda)$-RBM on $G[\Lambda]$.
\end{observation}

\subsection{Self-avoiding walk tree}\label{sec:saw-tree}
Let $G = (V,E)$ be a graph.
Assume that there is a total ordering of all vertices $V$ in $G$.
The self-avoiding walk (SAW) tree is defined as follows.
\begin{definition}[\text{SAW tree~\cite{weitz2006counting}}]\label{def:saw-tree}
 Let $G=(V,E)$ be a graph. For any vertex $v\in V$, the SAW tree $T_{\textnormal{SAW}}(G,v)$ rooted at $v$ enumerates all SAWs from $v$ such that every path $v_0-v_1-\cdots-v_\ell$ from root to leaf satisfies that either it is a SAW that ends at $v_\ell$ (namely the degree $\deg_G(v_\ell)$ of $v_\ell$ is $1$) or it is a SAW that ends at a cycle-closing vertex $v_\ell$ ($v_0-v_1-\cdots-v_{\ell-1}$ is a SAW and $v_\ell = v_i$ for some $0\leq i \leq \ell - 2$).
\end{definition}

In addition, we also need to consider SAW trees when a boundary is present.
Let $S\subseteq V$ be a set of boundary vertices.
The SAW tree $T = T_{\textnormal{SAW}}(G,v,S)$ rooted at $v$ with boundary $S$ is same as $T = T_{\textnormal{SAW}}(G,v)$ defined in \Cref{def:saw-tree},
except that any SAW stops immediately after reaching a boundary vertex $u \in S$, in which case $u$ is the last vertex in that SAW.
Thus, $ T_{\textnormal{SAW}}(G,v)$ is the same as $ T_{\textnormal{SAW}}(G,v,\emptyset)$.

The following two observations are straightforward to verify from the definition.
\begin{observation}\label{obs:non-leaf-vertices-in-saw-tree}
  For any non-leaf vertex $u$ in $T_{\textnormal{SAW}}(G,v,S)$, the degree of $u$ in $T$ is the same as the degree of its preimage $f(u)$ in $G$. 
\end{observation}

\begin{observation}\label{obs:leaf-vertices-in-saw-tree-with-boundary}
    Any leaf $u$ in $T_{\textnormal{SAW}}(G,v,S)$ falls into three disjoint types: (1) $u$ is a copy of some vertex in the boundary $S$; (2) $u$ is a cycle-closing vertex; (3) $u$ has degree one in $G$ and is not a copy of any vertex in $S$. As a corollary, any cycle-closing vertex $u$ cannot be a copy of any vertex in $S$.
\end{observation}

Consider a spin system on graph $G$ with parameters $(\beta_e,\gamma_e)_{e \in E},(\lambda_v)_{v \in V}$ and the Gibbs distribution $\mu$. Fix a vertex $v$ and a pinning $\sigma \in \{0,1\}^S$ over boundary $S$. To analyse the conditional marginal distribution $\mu_v^\sigma$, we need to use the following construction of SAW trees with pinnings.

\begin{definition}[SAW tree with pinning]\label{def:saw-tree-with-pinning}
    Let $\sigma \in \{0,1\}^S$ be a partial pinning on $S$, where $S \subseteq V$ is a set of boundary vertices.
    The SAW tree $T_{\textnormal{SAW}}(G,v,\sigma)$ rooted at $v$ with pinning $\sigma$ is constructed as follows.
    \begin{enumerate}
        \item Construct the SAW tree $T = T_{\textnormal{SAW}}(G,v,S)$ with boundary $S$.
        \item For any leaf vertex in $T$ that is a copy of some $u \in S$, pin its value to be $\sigma(u)$.
        \item For any cycle-closing leaf vertex $v_\ell$ in $T$, say $v_\ell = v_i$ for some $0\leq i \leq \ell - 2$ in the SAW, we pin the value of $v_\ell$ to be $0$ if $v_{i+1} > v_\ell$ and pin the value of $v_\ell$ to be $1$ if $v_{i+1} < v_\ell$ according to the total order of $V$.
    \end{enumerate}
    \end{definition}
         
    By \Cref{obs:leaf-vertices-in-saw-tree-with-boundary}, if some leaf vertex $u$ in $T_{\textnormal{SAW}}(G,v,S)$ gets pinned in the second step of \Cref{def:saw-tree-with-pinning}, then the pinning on $u$ will not be changed in the third step because $u$ cannot be a cycle-closing vertex.

    Let $T = T_{\textnormal{SAW}}(G,v,\sigma)$.
    Denote $T = (V_T,E_T)$, where $V_T$ is all vertices in $T$ and $E_T$ are all edges in $T$.
    By \Cref{def:saw-tree}, some leaf vertices of $T$ are cycle-closing vertices and we define
    \begin{align}\label{eq:gamma}
        \Gamma \defeq \{w \in V_T: w \text{ is a cycle-closing leaf vertex of } T\}.
    \end{align}
    We remark that $\Gamma \subseteq V_T$ is determined by the tuple $(G,v,S)$ and all vertices in $\Gamma$ are leaf vertices of $T$. We use $\rho_\Gamma$ to denote the pinning on all cycle-closing leaf vertices of $T$.

For a vertex $w$ in graph $G$, it may have multiple copies in $T$. We use $\text{copy}(w)$ to denote the set of all copies of $w$ in $T$.
Define the set of all copies of vertices in $S$ as
\begin{align}\label{eq:lambda}
  \bar{S} \defeq \bigcup_{w \in S} \text{copy}(w).
\end{align}
By the construction of $T$, $\bar{S}$ is a subset of leaf vertices in $T$. We use $\sigma_{\bar{S}}$ to denote the pinning on all vertices in $\bar{S}$. 
Note that $\sigma_{\bar{S}}$ is determined by the pinning $\sigma \in \{0,1\}^S$.

Every vertex in $T$ is a copy of some vertex in $G$ and every edge in $T$ is a copy of some edge in $G$. We can naturally define a Gibbs distribution on $T$ by inheriting the parameters of the two-spin systems on $G$. Denote the Gibbs distribution on $T$ as $\pi$. 
Let $\pi^{\bar{\sigma}}$ be the Gibbs distribution on $T$ with pinning $\bar{\sigma} = \rho_\Gamma \cup \sigma_{\bar{S}}$.  
The main point of all these constructions is the following well-known result by Weitz~\cite{weitz2006counting}. 
    
 \begin{proposition}[\text{\cite{weitz2006counting}}]\label{prop:marginal-distributions-are-identical}   
    For the root vertex $v$, two marginal distributions $\mu_v^\sigma$ and $\pi_v^{\bar{\sigma}}$ are identical. 
\end{proposition}

\subsection{Tree recursion and potential function}
Consider a SAW tree $T$ rooted at $v$ with pinning $\bar{\sigma}$ on a subset of leaf vertices. For each vertex $w \in T$, let $T_w$ be the sub-tree of $T$ rooted at $w$.
Consider the spin system induced by the sub-tree $T_w$ on the vertices in $T_w$.
Let $p_w(0)$ and $p_w(1)$ be the marginal probabilities of $w$ being 0 and 1 in the Gibbs distribution induced by the sub-tree $T_w$ respectively. Define
\begin{align}\label{eq:R-w}
    R_w = \frac{p_w(0)}{p_w(1)}.
\end{align}
If the value of $w$ is pinned to be 0, then $p_w(1) = 0$ and $R_w = \infty$. This happens only at leaves. 

Let $u$ be a vertex in $T$. Let $u_1,u_2,\ldots,u_d$ be the children of $u$. The tree recursion function $F_u: [0,\infty]^d \to \mathbb{R}$ at the vertex $u$ is defined as
\begin{align}\label{eq:tree-recursion}
F_u(x_1,x_2,\ldots,x_d) = \lambda_u \prod_{i=1}^d \frac{\beta_{u,u_i} x_i +1}{x_i+\gamma_{u,u_i}}.
\end{align}
Weitz~\cite{weitz2006counting} shows a well-known recursion relation 
\begin{align*}
    R_u = F_u(R_{u_1},\ldots,R_{u_d}).
\end{align*}

Let $\beta \leq 1 < \gamma$, $\beta \gamma > 1$ and $\lambda >0$ be three parameters. 
Now, let us consider a $(\beta,\gamma,\lambda)$-ferromagnetic two-spin system on graph $G$ in \Cref{def:ferromagnetic-two-spin-system}.
Guo and Lu~\cite{GuoL18} used a potential function method to analyze the recursion function. 
By~\eqref{eq:tree-recursion}, the image space of $F_u$ is within $[0,\lambda)$.
Let $\Phi: [0,\lambda) \to \mathbb{R}$ be a differentiable and increasing potential function. Instead of analyzing the recursion of $R_w$, they analyze the recursion of $\Phi(R_w)$.
The tree recursion in~\eqref{eq:tree-recursion} at vertex $u$ with potential function $\Phi$ is
\begin{align*}
F^\Phi_u (y_1,y_2,\ldots,y_d) = (\Phi \circ F_u \circ \Phi^{-1})(y_1,y_2,\ldots,y_d),
\end{align*}
where $y= \Phi(x)$, $y_i$ = $\Phi(x_i)$ and all $x_i \in [0,\lambda)$.

The potential function used by Guo and Lu~\cite{GuoL18} for ferromagnetic two-spin systems is given implicitly via its derivative $\phi(x) = \Phi'(x)$,
which is
\begin{align}\label{eq:phi-definition}
  \phi(x)\defeq \min\left \{\frac{1}{x\log\frac{\lambda}{x}},\frac{1}{t}\right \}, \text{ where } t = t(\beta,\gamma,\lambda) > 0\text{ is a constant}.
\end{align}

Let $\gamma > 1$ be a constant. For $(\lambda,\gamma)$-RBMs, which is a special case of ferromagnetic two-spin systems, we can define $\phi(x)$ in the same form as
\begin{align}\label{eq:phi-definition-1}
    \phi(x)\defeq \min\left \{\frac{1}{x\log\frac{\lambda}{x}},\frac{1}{t}\right \}, \text{ where } t = t(\gamma,\lambda) > 0\text{ is a constant}.
\end{align}

The following condition is easy to prove using the definition of $\phi(x)$.
\begin{condition}\label{condition:phi-bound}
There exist constants $C_{\max}(t,\lambda) > 0$ and $C_{\min}(t,\lambda) > 0$ such that
$$\forall x \in [0,\lambda), \quad C_{\min} \leq \phi(x) \leq C_{\max}.$$
\end{condition}

The specific definition of the constant $t$ can be found in \cite{GuoL18}. The potential function is then
\begin{align}\label{eq:Phi-definition}
  \Phi(x) = \int_0^x \phi(s)\, ds.
\end{align}
We also require that the potential function $\Phi(x)$ satisfies the following property.

\begin{condition}\label{condition:potential-function-bound}
There exist a constant $0< \alpha = \alpha(\beta,\gamma,\lambda) < 1$ such that for all $x_1,\ldots,x_d \in (0, \lambda)$,
    \begin{align*}
    C_{\phi,d}(\boldsymbol{x}) \defeq \phi(F_u(\boldsymbol{x})) \sum_{i=1}^d \left\vert  \frac{\partial F_u}{\partial x_i}(\boldsymbol{x})  \right\vert \frac{1}{\phi(x_i)}  \leq 1 - \alpha.
    \end{align*}
\end{condition}

\begin{condition}\label{condition:trivial-bound-for-C}
There exist constants $C_{\text{trl}} > 0$ and $C_{\text{decay}} > 0$ such that for any $x_1,x_2,\ldots,x_d \in (0, \lambda)$,
\begin{align*}
    \phi(F_u(\boldsymbol{x})) \left\vert  \frac{\partial F_u}{\partial x_i}(\boldsymbol{x})  \right\vert \frac{1}{\phi(x_i)}
    \leq C_{\text{trl}} \lambda_u \exp(-C_{\text{decay}} d).
\end{align*}
\end{condition}

In \cite{GuoL18}, \Cref{condition:potential-function-bound} is proved for uniform parameters, namely, the same $\beta,\gamma$ for all edges and the same $\lambda$ for all vertices. For non-uniform parameters $(\lambda_v)_{v \in V}$ and $(\beta_e,\gamma_e)_{e \in E}$ and RBMs, the proof is similar. The proof of the following lemma can be found in Appendix \ref{sec:decay-general}. 
\begin{lemma}[\text{\cite{GuoL18}}]\label{lemma:potential-function-bound-GL}
    Let $\beta \leq 1 < \gamma$, $\beta \gamma > 1$ and $\lambda < \lambda_c(\beta,\gamma) \defeq (\gamma/\beta)^{\frac{\sqrt{\beta \gamma}}{\sqrt{\beta \gamma}-1}}$ be three parameters.
    Consider the recursion function $F_u$ in \eqref{eq:tree-recursion} with $\lambda_u < \lambda$ and for any edge $e = \{u,u_i\}$, $\beta_e \leq \beta \leq 1 < \gamma \leq \gamma_e$, $\beta\gamma \geq \beta_e \gamma_e > 1$.
    Then, \Cref{condition:phi-bound} and \Cref{condition:potential-function-bound} hold.

    Similarly, let $\gamma > 1$ and $\lambda < \lambda_c(\gamma,1) \defeq \gamma^{\frac{\sqrt{\gamma}}{\sqrt{\gamma}-1}} $ be two constants. Consider the recursion function $F_u$ in \eqref{eq:tree-recursion} with $\lambda_u < \lambda$ and for any edge $e = \{u,u_i\}$, $\beta_e = 1 < \gamma \leq \gamma_e$. Then, \Cref{condition:phi-bound} and \Cref{condition:potential-function-bound} hold.
\end{lemma}

    In addition, we also have the following trivial bound for each term in the sum $C_{\phi,d}(\boldsymbol{x})$, which verifies \Cref{condition:trivial-bound-for-C}. We remark that the following lemma does \emph{not} require that $\beta_e\gamma_e \leq \beta \gamma$ for all $e \in E$. Hence, it verifies \Cref{condition:trivial-bound-for-C} for both $(\beta,\gamma,\lambda)$-ferromagnetic two-spin systems and $(\gamma,\lambda)$-RBMs.

    \begin{lemma}\label{lemma:trivial-bound-for-C-ver}
        Let $\gamma > 1$ and $\lambda > 0$ be parameters.
        Consider the recursion function $F_u$ in \eqref{eq:tree-recursion} with $\lambda_u < \lambda$ and for any edge $e = \{u,u_i\}$, $\beta_e \leq 1 < \gamma \leq \gamma_e$ and $\beta_e\gamma_e > 1$ for all $e \in E$. For any $1\leq i \leq d$, it holds that for any $x_1,x_2,\ldots,x_d \in (0, \lambda)$,
    \begin{align*}
        \phi(F_u(\boldsymbol{x})) \left\vert  \frac{\partial F_u}{\partial x_i}(\boldsymbol{x})  \right\vert \frac{1}{\phi(x_i)} &\leq C_{\text{trl}} \cdot  \lambda_u  \tp{\frac{\lambda + 1}{\lambda + \gamma}}^{d-1}
        \leq C_{\text{trl}} \lambda_u \exp(-C_{\text{decay}} d),
    \end{align*}
    where $C_{\text{trl}} > 0$ and $C_{\text{decay}}  > 0$ are two constants.
    \end{lemma}
    \begin{proof}
    By \Cref{condition:phi-bound}, we have $\phi(F_u(\boldsymbol{x})) \leq C_{\max}$ and $\frac{1}{\phi(x_i)} \leq \frac{1}{C_{\min}}$. Further, since $\beta_e \leq 1$, 
    \begin{align*}
     \abs{\frac{\partial F_u}{\partial x_i}(\boldsymbol{x})} = \lambda_u \frac{\beta_{u,u_i}\gamma_{u,u_i} - 1}{(x_i + \gamma_{u,u_i})^2} \prod_{1 \leq j \leq d: j \neq i} \frac{\beta_{u,u_j} x_j +1}{x_j+\gamma_{u,u_j}} &\leq \lambda_u \frac{\gamma_{u,u_i} - 1}{\gamma_{u,u_i}^2} \tp{\frac{\lambda + 1}{\lambda + \gamma}}^{d-1}\\
     \text{(by $\gamma_{u,u_i} > 1$)}\quad &\leq \lambda_u \tp{\frac{\lambda + 1}{\lambda + \gamma}}^{d-1}. \qedhere
    \end{align*}
    \end{proof}

    Using the potential function and the above property, Guo and Lu~\cite{GuoL18} showed the following strong spatial mixing (SSM) result for $(\beta,\gamma,\lambda)$-ferromagnetic two-spin systems.

    \begin{definition}[SSM property]\label{def:ssm-property}
        A spin system on graph $G=(V,E)$ with Gibbs distribution $\mu$ satisfies the strong spatial mixing (SSM) property if there exist constants $A > 0$ and $0 < B < 1$ such that for any two configurations $\sigma$ and $\tau$ in a subset $\Lambda \subseteq V$, where $\sigma$ and $\tau$ differ only at subset $D \subseteq \Lambda$, then for any vertex $v \not\in \Lambda$, it holds that 
        \begin{align*}
        \abs{\frac{\mu^\sigma_v(0)}{\mu^\sigma_v(1)} - \frac{\mu^\tau_v(0)}{\mu^\tau_v(1)}} \leq A (1 - B)^\ell,
        \end{align*}
        where $\ell = \min_{u \in D} d(u,v)$ is the distance from $v$ to the closest vertex in $D$.
    \end{definition}

\begin{lemma}[\text{\cite{GuoL18}}]\label{lemma:ssm-property}
    Let $\beta \leq 1 < \gamma$, $\beta \gamma > 1$ and $\lambda < \lambda_c(\beta,\gamma) \defeq (\gamma/\beta)^{\frac{\sqrt{\beta \gamma}}{\sqrt{\beta \gamma}-1}}$ be three constants.
    Any $(\beta,\gamma,\lambda)$-ferromagnetic two-spin system satisfies the SSM property.
\end{lemma}

Using the potential function, we can prove the SSM property for $(\gamma,\lambda)$-RBMs.
\begin{lemma}\label{lemma:ssm-property-rbm}
    Let $\gamma > 1, \lambda < \lambda_c(\gamma,1) \defeq \gamma^{\frac{\sqrt{\gamma}}{{\sqrt{\gamma}-1}}} $ be constants.
    Any $(\gamma,\lambda)$-RBM satisfies the SSM property.
\end{lemma}

Given the bound of derivative in \Cref{lemma:potential-function-bound-GL}, \Cref{lemma:ssm-property-rbm} follows from the same potential-function argument used to prove \Cref{lemma:ssm-property} in \cite{GuoL18}. 


\section{All-to-one influence bound}

We start by establishing the all-to-one influence bound.
The analysis here is also useful later to establish ASSM in \Cref{sec:assm-in-saw-tree}.

\begin{definition}[All-to-one influence]\label{def:all-to-one-influence}
    Let $\mu$ be a distribution over $\{0,1\}^V$. We say that $\mu$ has $C_{\text{inf}}$-bounded all-to-one influence if, for every vertex $v \in V$,
\begin{align*}
    \sum_{u \in V\setminus\{v\}}\abs{\Pr_{X\sim \mu}[X(v) = 0 \mid X(u) = 0] -
    \Pr_{X\sim \mu}[X(v) = 0 \mid X(u) = 1]} \leq C_{\text{inf}}.
\end{align*}
\end{definition}


\begin{theorem}\label{thm:all-to-one-influence}
    Let $\beta \leq 1 < \gamma$, $\beta\gamma > 1$, and $\lambda < \lambda_c(\beta,\gamma) \defeq (\gamma/\beta)^{\frac{\sqrt{\beta \gamma}}{\sqrt{\beta \gamma}-1}}$. 
    Let $\mu$ be the Gibbs distribution for a $(\beta,\gamma,\lambda)$-ferromagnetic two-spin system on $G=(V,E)$.
    It has $C_{\textnormal{inf}}$-bounded all-to-one influence, where $C_{\textnormal{inf}} = C_{\textnormal{inf}}(\beta,\gamma,\lambda) > 0$ is a constant depending only on $\beta,\gamma,\lambda$.

    Let $\gamma > 1, \lambda < \lambda_c(1,\gamma) \defeq \gamma^{\frac{\sqrt{\gamma}}{{\sqrt{\gamma}-1}}} $ be constants.
    Let $\mu$ be the Gibbs distribution for a $(\gamma,\lambda)$-RBM on $G=(V,E)$.
    It has $C_{\textnormal{inf}}$-bounded all-to-one influence, where $C_{\textnormal{inf}} = C_{\textnormal{inf}}(\gamma,\lambda) > 0$ is a constant depending only on $\gamma,\lambda$.
\end{theorem}


We treat RBMs as a special case of ferromagnetic two-spin systems. We will highlight the differences between the two cases if necessary.
To prove this theorem, consider the SAW tree $T = T_{\textnormal{SAW}}(G,v,\emptyset)$ rooted at $v$. The cycle-closing leaves of $T$ have fixed pinned values. We use the self-reducibility property in \Cref{obs:self-reducibility} to remove all cycle-closing leaves from the SAW tree and update the external fields at their neighbours.
Thus, without loss of generality, we can assume there is no pinning on $T$. 
Let $\pi$ denote the Gibbs distribution on $T = (V_T,E_T)$, where the parameters are inherited from $\mu$. Fix a vertex $w \in V$. Let $S = \text{copy}(w)$ be the set of all copies of $w$ in $T$. By \Cref{prop:marginal-distributions-are-identical},  $\mu_v^{w \gets c}$ is identical to $\pi_v^{S \gets c}$ for $c \in \{0,1\}$, where $S \gets c$ is the pinning on $S$ such that all $x \in S$ are pinned to be $c$.

Note that tree is also a bipartite graph. The spin system on the tree is also a $(\beta,\gamma,\lambda)$-ferromagnetic two-spin system (for the two-spin system case) or a $(\gamma,\lambda)$-RBM (for the RBM case). Using \Cref{lemma:potential-function-bound-GL}, \Cref{lemma:trivial-bound-for-C-ver}, \Cref{lemma:ssm-property}, and \Cref{lemma:ssm-property-rbm}, the following condition holds for the spin system on the tree, where $t$ is the constant in $\phi(x)$ in \eqref{eq:phi-definition} and \eqref{eq:phi-definition-1} for the two-spin system case and RBM case respectively.
\begin{condition}\label{condition}
    Let $\lambda,\gamma,t > 0$ be constants, where $\gamma > 1$.
    The ferromagnetic two-spin system on the SAW tree $T = (V_T,E_T)$ satisfies the following conditions:
  \begin{itemize}
    \item for all $e \in E_T$, $\beta_e\gamma_e > 1$ and $\beta_e \leq 1 < \gamma \leq  \gamma_e$, and for all $v \in V_T$ $\lambda_v < \lambda$;
    \item there exists a potential function $\Phi$ with derivative $\phi$ in forms of \eqref{eq:phi-definition} and \eqref{eq:phi-definition-1} with $t$ and $\lambda$ satisfying:
       \begin{itemize}
         \item \Cref{condition:phi-bound} (boundedness) with constants $C_{\max}$ and $C_{\min}$,
         \item \Cref{condition:potential-function-bound} (correlation decay) with constant $\alpha$,
         \item and  \Cref{condition:trivial-bound-for-C} (large degree decay) with constants $C_{\text{trl}}$ and $C_{\text{decay}}$;
       \end{itemize}
      \item the strong spatial mixing (SSM) property holds with constants $A> 0$ and $0 < B < 1$.
    \end{itemize}
\end{condition}
By \Cref{lemma:potential-function-bound-GL}, \Cref{lemma:trivial-bound-for-C-ver}, \Cref{lemma:ssm-property}, and \Cref{lemma:ssm-property-rbm}, all the constants in \Cref{condition} depend only on $\beta,\gamma,\lambda$ for the two-spin system case and $\gamma,\lambda$ for the RBM case.

For any vertex $u \in V_T$, let $R_u$ be the marginal ratio at $u$ defined in \eqref{eq:R-w}. The ratio $R_u$ can be computed recursively using the tree recursion function $F_u$ in \eqref{eq:tree-recursion} in a bottom-up manner. From this perspective, $T$ can also be viewed as a computation tree for the ratio $R_u$.

\begin{definition} [Pinning on the computation tree]\label{def:pinning-on-computation-tree}
  Let $u \in V_T$ and $S$ be a subset of vertices in the subtree of $u$, where $u \notin S$. Let $\sigma: S \to [0,\infty]$ be a pinning on $S$ (of ratios). For each $x \in S$, we remove all the descendants of $x$ and fix the value $R_x = \sigma(x)$. Then, all pinnings are on the leaves of the subtree rooted at $u$. For all other leaf vertices $x'$, we set $R_{x'} = \lambda_{x'}$ as the definition of $R_{x'}$ in \eqref{eq:R-w}. We use $R^\sigma_u$ to denote the marginal ratio at $u$ computed via tree recursion in a bottom-up manner.

We also use the notation $R^\sigma_u$ even if $\sigma$ contains pinning outside the subtree of $u$. In this case, $R^{\sigma}_u = R^{\bar{\sigma}}_u$, where $\bar{\sigma}$ is the pinning obtained from $\sigma$ by removing the pinning outside the subtree of $u$.
\end{definition}

By definition, it is straightforward to verify that $R^{S \gets \infty}_v = \frac{\mu_v^{w \gets 0}(0)}{\mu_v^{w \gets 0}(1)}$ and $R^{S \gets 0}_v = \frac{\mu_v^{w \gets 1}(0)}{\mu_v^{w \gets 1}(1)}$, where $S$ is the set of all copies of $w$ in $T$.
Note that for the computation tree, pinnings are with respect to the ratio $R$ instead of the state, although it is easy to translate between the two.
To emphasize that the pinning is on all copies of $w$, we denote 
\begin{align*}
R_v^{w^0} = R_v^{S \gets \infty} \quad \text{and} \quad R_v^{w^1} = R_v^{S \gets 0}.
\end{align*}

The following lemma is straightforward.
\begin{lemma}\label{lemma:dtv<R}
    The influence of $w$ on $v$ can be bounded by  
    \begin{align*}
        \DTV{\mu_v^{w \gets 0}}{\mu_v^{w \gets 1}} \leq \abs{R_v^{w^0}-R_v^{w^1}},
    \end{align*}
\end{lemma}
\begin{proof}
    By \Cref{prop:marginal-distributions-are-identical}, $\mu_v^{w \gets c}$ coincides with $\pi_v^{S \gets c}$ for $c \in \{0,1\}$, where $S = \text{copy}(w)$ and $R_v^{w^0} = R_v^{S \gets \infty}$, $R_v^{w^1} = R_v^{S \gets 0}$ as in the notation above. So $\DTV{\mu_v^{w \gets 0}}{\mu_v^{w \gets 1}} = \DTV{\pi_v^{S \gets \infty}}{\pi_v^{S \gets 0}}$. The marginals at $v$ are Bernoulli: $\pi_v^{S \gets \infty}(1) = 1/(1+R_v^{w^0})$ and $\pi_v^{S \gets 0}(1) = 1/(1+R_v^{w^1})$. Thus
    \begin{align*}
    \DTV{\pi_v^{S \gets \infty}}{\pi_v^{S \gets 0}} = \left | \frac{1}{1+R_v^{w^0}} - \frac{1}{1+R_v^{w^1}} \right | = \frac{\bigl| R_v^{w^0} - R_v^{w^1} \bigr|}{(1+R_v^{w^0})(1+R_v^{w^1})} \leq \bigl| R_v^{w^0} - R_v^{w^1} \bigr|. &\qedhere
    \end{align*}
\end{proof}

In $R^{w^0}_v$ and $R^{w^1}_v$, a set of vertices is pinned to $0$ or $1$. Next, we decompose the influence into the sum of influences contributed by individual vertices in this set. We define the following notion of influence from one vertex in the computation tree. A similar definition and analysis for the hardcore model appears in~\cite{ALO24}, but we need a more careful definition for ferromagnetic two-spin systems. Define the set of vertices at level $k$ by
\begin{align*}
\forall k \in \mathbb{N}, \quad L_k(u) = \{v \in V_T: d(v,u) = k\},
\end{align*}
where $d(v,u)$ is the distance from $v$ to $u$ in the SAW tree $T$. A vertex $u'$ is called a sibling of $u$ if $u'$ has the same parent as $u$.

\begin{definition}[Influence from one vertex in the computation tree]\label{def:influence-from-one-vertex}
Let $u \in L_k(v)$ be a vertex in the computational tree $T$ at level $k$. Define the influence of $u$ on $v$ as
\begin{align*}
I_v^u = \sup_{\sigma \in \+S} \left | R_v^{\sigma \land u \gets \infty} - R_v^{\sigma \land u \gets 0} \right |,
\end{align*}
where $\+S$ contains all pinnings $\sigma: L_k(v) \setminus \{u\} \to [0,\infty]$ satisfying that for all siblings $u'$ of $u$, $\sigma(u') \in (0,\lambda)$.
\end{definition}

Compared to the definition in~\cite{ALO24}, our definition explicitly constrains the siblings of $u$. We next prove the following influence bound using the technique in~\cite{ALO24}.

\begin{lemma}\label{lemma:influence-bound-one-vertex}
    The influence satisfies
    \begin{align*}
    |R_v^{w^0} - R_v^{w^1}| \leq 2\sum_{u \in \text{copy}(w)} I_v^u.
    \end{align*}
\end{lemma}
\begin{proof}
Let $u_1, \ldots, u_m$ be the vertices in $\text{copy}(w)$ in the increasing order of the distance to root $v$, which means $d(v,u_i)\leq d(v,u_j)$ for $1\leq i<j\leq m$.
Let $S_i:=\{u_i,\cdots,u_m\}$ for $1\leq i\leq m$.
For $j$ from $0$ to $m$, we inductively show that:
\begin{align*}
    |R_v^{w^0} - R_v^{w^1}|\leq 2\sum_{i=1}^{j} I_v^{u_i}+|R_v^{S_{j+1}\gets \infty} - R_v^{S_{j+1}\gets 0}|,
\end{align*}
where $S_{m+1}=\emptyset$ and $|R_v^{S_{m+1}\gets \infty}-R_v^{S_{m+1}\gets 0}|=0$. When $j=0$, the inequality holds trivially. Assume that the inequality holds for $j$ for some $0\leq j<m$. We next show that the inequality also holds for $j+1$. 
By the triangle inequality, we have
\begin{align*}
|R_v^{S_{j+1}\gets \infty} - R_v^{S_{j+1}\gets 0}|\leq |R_v^{S_{j+1}\gets \infty} - R_v^{S_{j+2}\gets \infty}|+|R_v^{S_{j+2}\gets \infty} - R_v^{S_{j+2}\gets 0}|+|R_v^{S_{j+2}\gets 0} - R_v^{S_{j+1}\gets 0}|.   
\end{align*}
To verify the $j+1$ case, using the induction hypothesis on $j$, it suffices to show that the first and third terms are each bounded by $I_v^{u_{j+1}}$. We only prove this for the first term, since the third term is analogous.
By the monotonicity of the recursion function,
\begin{align*}
|R_v^{S_{j+1}\gets \infty} - R_v^{S_{j+2}\gets \infty}|\leq |R_v^{S_{j+2}\gets \infty\land u_{j+1}\gets\infty } - R_v^{S_{j+2}\gets \infty\land u_{j+1}\gets 0}|.
\end{align*}
Because $d(v,u_{j+1}) \leq d(v,u_i)$ for all $i > j+1$, all pinnings on $S_{j+2}$ induce pinnings on $L_k(v) \setminus \{u_{j+1}\}$, where $k$ is the level of $u_{j+1}$ in $T$. Moreover, all siblings of $u_{j+1}$ are not in $\text{copy}(w)$, and thus they are unpinned. By the definition of the tree recursion, when we compute the tree recursion from bottom to top, all siblings $u'$ of $u_{j+1}$ obtain a value in $(0,\lambda)$, which is the induced pinning on $u'$. 
For all other vertices $u'' \in L_k(v) \setminus \{u_{j+1}\}$ that is not a sibling of $u_{j+1}$, the ratio on $u''$ computed via the tree recursion can be any value in $[0,\infty]$.
Therefore, by the definition of $I_v^{u_{j+1}}$ in \Cref{def:influence-from-one-vertex}, we have
\begin{align*}
|R_v^{S_{j+1}\gets \infty} - R_v^{S_{j+2}\gets \infty}|\leq I_v^{u_{j+1}}.
\end{align*}
The same argument gives $|R_v^{S_{j+2}\gets 0} - R_v^{S_{j+1}\gets 0}|\leq I_v^{u_{j+1}}$. This proves the $j+1$ case and hence the lemma.
\end{proof}

Using~\Cref{lemma:influence-bound-one-vertex} and \Cref{lemma:dtv<R}, we have the following bound 
\begin{align}\label{eq:bound-on-dtv}
    \sum_{w \in V \setminus \{v\}} \DTV{\mu_v^{w \gets 0}}{\mu_v^{w \gets 1}} \leq 2\sum_{w \in V \setminus \{v\}} \sum_{u \in \text{copy}(w)} I_v^u = 2\sum_{k \geq 1} \sum_{w \in L_k(v)} I_v^w.
\end{align}

Next, fix an integer $k \geq 1$. We bound the sum of influences over all vertices in $L_k(v)$. We also work with the potential function $\Phi$ defined in \eqref{eq:phi-definition}. Fix a vertex $w \in L_k(v)$. Let $\sigma^w$ be a pinning on $L_k(v) \setminus \{w\}$ that attains (or is arbitrarily close to) the supremum in the definition of $I_v^w$. We emphasize that $\sigma^w$ depends on $w$. Instead of directly bounding $I_v^w$, we bound the potential difference
$|\Phi(R_v^{\sigma^w \land w \gets \infty}) - \Phi(R_v^{\sigma^w \land w \gets 0})|$.
We use the following general relation.

\begin{lemma}\label{lemma:dtv-bound-potential}
For any two $x^0,x^1 \in (0,\lambda)$, we have
\begin{align*}
   \frac{1}{C_{\max}} \left | \Phi(x^0)-\Phi(x^1) \right | \leq \abs{x^0-x^1} \leq \frac{1}{C_{\min}} \left | \Phi(x^0)-\Phi(x^1) \right |,
\end{align*}
where $C_{\max}$ and $C_{\min}$ are constants defined in \Cref{condition:phi-bound}.
\end{lemma}


\begin{proof}
For the potential $\Phi$ with derivative $\phi = \Phi'$ from~\eqref{eq:phi-definition}, the mean value theorem gives
\[
|\Phi(x^0) - \Phi(x^1)| = \phi(\eta)\,|x^0 - x^1|
\]
for some $\eta$ between $x^0$ and $x^1$. By \Cref{condition:phi-bound}, for any $z \in (0,\lambda)$, we have $\phi(z) \geq C_{\min}$ and $\phi(z) \leq C_{\max}$. The lemma can be proved using the following equation:
\begin{align*}
|x^0 - x^1| = \frac{|\Phi(x^0) - \Phi(x^1)|}{\phi(\eta)} &\qedhere
\end{align*}
\end{proof}

Now, our task is reduced to bound the difference of the potential $\Phi(R_v^{\sigma^w \land w \gets \infty}) - \Phi(R_v^{\sigma^w \land w \gets 0})$. In \Cref{subsec:general-results-for-correlation-decay}, we give some general influence decay results. In \Cref{subsec:proof-of-the-influence-bound}, we apply these general results to prove the influence bound.

\subsection{General influence decay results}\label{subsec:general-results-for-correlation-decay}
Next, we present general results for proving the influence bound, which will also be used later to prove aggregate strong spatial mixing. Consider a ferromagnetic two-spin system $\+S$ on a tree $T$, rooted at $v$. For each vertex $w \in L_k(v)$, let $\sigma^w$ be a pinning on $L_k(v) \setminus \{w\}$. Different vertices $w$ may correspond to different pinnings $\sigma^w$. Define the potential-based influence from $w$ to the root $v$ as
\begin{align}\label{eq:K-v-w-potential}
    K_v^w = \abs{\Phi(R_v^{\sigma^w \land w \gets \infty}) - \Phi(R_v^{\sigma^w \land w \gets 0})}.
\end{align}
More generally, for any vertex $u$ on the path between $w$ and $v$, define the influence $K_u^w$ of $w$ on $u$ by
\begin{align}\label{eq:K-u-w-potential}
K_u^w = \abs{\Phi(R_u^{\sigma^w \land w \gets \infty}) - \Phi(R_u^{\sigma^w \land w \gets 0})},
\end{align}
where $R_u^{\sigma^w \land w \gets \infty}$ and $R_u^{\sigma^w \land w \gets 0}$ are ratios computed by tree recursion in $\+S$, with $\sigma^w$ restricted to the subtree rooted at $u$.

The following two general influence decay results hold. 
\begin{lemma}\label{lemma:general-influence-decay-1}
    Suppose $\+S$ in $T$ is a ferromagnetic two-spin system satisfying \Cref{condition}.
    Let $u \in L_\ell(v)$ be a vertex at level $\ell$, where $0 \le \ell \leq k-2$. Let $u_1,u_2,\ldots,u_d$ be the children of $u$. Then
\begin{align*}
\sum_{w \in L_{k-\ell}(u)} K_u^w 
 \leq C_\text{trl} \lambda_u d \exp(-C_{\text{decay}}d)\max_{1 \leq i \leq d} \sum_{w \in L_{k-\ell-1}(u_i)} K_{u_i}^w,
\end{align*}
where $L_j(u)$ denotes the set of vertices at level $j$ in the subtree rooted at $u$.
\end{lemma}




\begin{lemma}\label{lemma:general-influence-decay-2}
Suppose $\+S$ in $T$ is a ferromagnetic two-spin system satisfying \Cref{condition}.
There exist constants $\ell_0$ and $0< \delta < 1$ such that if $k > \ell_0$, then for any $0 \leq \ell \leq k - \ell_0$, for any vertex $u \in L_\ell(v)$ with children $u_1,\cdots,u_d$, it holds that 
\begin{align*}
  \sum_{w \in L_{k-\ell}(u)} K_u^w \leq (1 - \delta) \max_{1 \leq i \leq d} \sum_{w \in L_{k-\ell-1}(u_i)} K_{u_i}^w.
\end{align*}
\end{lemma}

These two lemmas can be proved by combining the techniques developed in~\cite{GuoL18,ALO24}. Compared to the proof in \cite{ALO24} for the hardcore model, our proof needs to carefully analyze the potential function $\Phi$ and use the decay results in \Cref{condition:potential-function-bound} and \Cref{condition:trivial-bound-for-C} to control the influence decay.

\begin{proof}[Proof of \Cref{lemma:general-influence-decay-1}]
    We have $L_{k-\ell}(u) = \bigcup_{i=1}^d L_{k-\ell-1}(u_i)$ (disjoint). Fix a $w \in L_{k-\ell-1}(u_i)$, where $w$ lies in the subtree of $u_i$. For each $j \neq i$, define the marginal ratio $z_j^w$ at $u_j$ as $z_j^w = R_{u_j}^{\sigma^w}$. For the subtree rooted at $u_i$, define two ratios $z_i^{w,0}$ and $z_i^{w,\infty}$ as $z_i^{w,0} = R_{u_i}^{\sigma^w \land w \gets 0}$ and $z_i^{w,\infty} = R_{u_i}^{\sigma^w \land w \gets \infty}$. Then, two ratios $R_u^{\sigma^w \land w \gets 0}$ and $R_u^{\sigma^w \land w \gets \infty}$ can be written as
    \begin{align*}
    R_u^{\sigma^w \land w \gets 0} &= F_u(z_1^w, \ldots, z_{i-1}^w, z_i^{w,0}, z_{i+1}^w, \ldots, z_d^w)\\
    R_u^{\sigma^w \land w \gets \infty} &= F_u(z_1^w, \ldots, z_{i-1}^w, z_i^{w,\infty}, z_{i+1}^w, \ldots, z_d^w).
    \end{align*}
    Let $y_j^w = \Phi(z_j^w)$ for $j \neq i$, $y_i^{w,0} = \Phi(z_i^{w,0})$, $y_i^{w,\infty} = \Phi(z_i^{w,\infty})$. The potential recursion is
    \begin{align*}
    \Phi(R_u^{\sigma^w \land w \gets 0}) &= (\Phi\circ F_u \circ \Phi^{-1})(y_1^w, \ldots, y_{i-1}^w, y_i^{w,0}, y_{i+1}^w, \ldots, y_d^w)\\
    \Phi(R_u^{\sigma^w \land w \gets \infty}) &= (\Phi\circ F_u \circ \Phi^{-1})(y_1^w, \ldots, y_{i-1}^w, y_i^{w,\infty}, y_{i+1}^w, \ldots, y_d^w).
    \end{align*}
    By definition, $K_u^w = \bigl| \Phi(R_u^{\sigma^w \land w \gets 0}) - \Phi(R_u^{\sigma^w \land w \gets \infty}) \bigr|$. Applying the mean value theorem to the map $y_i \mapsto (\Phi\circ F_u\circ \Phi^{-1})(y_1^w,\ldots,y_i^w,\ldots,y_d^w)$ (with $y_j^w$ for $j \neq i$ fixed), there exists $\tilde{y}_i^w$ between $y_i^{w,0}$ and $y_i^{w,\infty}$ such that
\begin{align*}
K_u^w = \left | \frac{\partial (\Phi\circ F_u\circ \Phi^{-1})}{\partial y_i}(y_1^w,\ldots,\tilde{y}_i^w,\ldots,y_d^w) \right | \cdot \bigl| y_i^{w,0} - y_i^{w,\infty} \bigr|.
\end{align*}
Let $\tilde{z}_i^w = \Phi^{-1}(\tilde{y}_i^w)$; then $\tilde{z}_i^w$ lies between $z_i^{w,0}$ and $z_i^{w,\infty}$. Compute the partial derivative $\frac{\partial (\Phi\circ F_u\circ \Phi^{-1})}{\partial y_i}$ by the chain rule. With $\boldsymbol{z}^w = (z_1^w,\ldots,z_{i-1}^w,\tilde{z}_i^w,z_{i+1}^w,\ldots,z_d^w)$ we have
\begin{align}\label{eq:K-u-w-potential-pf}
K_u^w &= \frac{ \phi(F_u(\boldsymbol{z}^w))}{\phi(\tilde{z}_i^w)} \left | \frac{\partial F_u}{\partial z_i}(\boldsymbol{z}^w) \right |  \bigl| y_i^{w,0} - y_i^{w,\infty} \bigr|\leq \frac{\phi(F_u(\boldsymbol{z}^w))}{\phi(\tilde{z}_i^w)} \left | \frac{\partial F_u}{\partial z_i}(\boldsymbol{z}^w) \right | \cdot K_{u_i}^w,
\end{align}
where the last equation holds because $\bigl| y_i^{w,0} - y_i^{w,\infty} \bigr| = \bigl| \Phi(z_i^{w,0}) - \Phi(z_i^{w,\infty}) \bigr| = K_{u_i}^w$.
Summing over $w \in L_{k-\ell}(u)$, we have
\begin{align}\label{eq:sum-K-u}
\sum_{w \in L_{k-\ell}(u)} K_u^w \leq \sum_{i=1}^d \sum_{w \in L_{k-\ell-1}(u_i)} \frac{\phi(F_u(\boldsymbol{z}^w))}{\phi(\tilde{z}_i^w)} \left | \frac{\partial F_u}{\partial z_i}(\boldsymbol{z}^w) \right | \cdot K_{u_i}^w.
\end{align}

By the assumption of the lemma, $\ell \leq k - 2$. Hence, all $z_j^w$ for $j \neq i$ and $\tilde{z}_i^w$ are in the range $(0,\lambda)$. Using \Cref{condition:trivial-bound-for-C}, we have the following bound 
\begin{align*}
    \sum_{w \in L_{k-\ell}(u)} K_u^w  &\leq \sum_{i=1}^d C_\text{trl} \lambda_u \exp(-C_{\text{decay}} d)  \sum_{w \in L_{k-\ell-1}(u_i)} K_{u_i}^w\\
    &\leq C_\text{trl} \lambda_u d \exp(-C_{\text{decay}} d) \max_{1 \leq i \leq d} \sum_{w \in L_{k-\ell-1}(u_i)} K_{u_i}^w. \qedhere
\end{align*}
\end{proof}

\begin{proof}[Proof of \Cref{lemma:general-influence-decay-2}]
    We start from~\eqref{eq:sum-K-u}.
    For each $i$, the coefficient of $K_{u_i}^w$ depends on $w$ through $\boldsymbol{z}^w$ (every $z_j^w$ for depends on $w$). If we could use a single $\boldsymbol{z} = (z_1,\ldots,z_d)$ for all $w$, then \Cref{condition:potential-function-bound} would give $\phi(F_u(\boldsymbol{z})) \sum_{i=1}^d \bigl| \frac{\partial F_u}{\partial z_i}(\boldsymbol{z}) \bigr| \frac{1}{\phi(z_i)} < 1 - \alpha$, so we can use the lemma to bound $K_u^w$. But here $\boldsymbol{z}^w$ depends on $w$. 
Using the technique in \cite{ALO24},
we resolve this when $\ell \leq k - \ell_0$ using the SSM property in \Cref{condition}.

We will use the SSM property for ratio pinnings, although it is stated in \Cref{def:ssm-property} for $0/1$ pinnings. This follows from monotonicity of the recursion. Since $\beta_e\gamma_e>1$, each $F_u$ is coordinate-wise increasing. Therefore, for any ratio pinning $\rho$ on a level and any ancestor $x$, the ratio $R_x^\rho$ lies between the two ratios obtained by replacing every pinned ratio in $\rho$ by $0$ and by $\infty$. Hence the SSM bound in \Cref{condition} also bounds the effect of arbitrary ratio pinnings at distance $r$ by $A(1-B)^r \leq A\exp(-Br)$.

Define the pinning $\tau$ on $L_k(v)$ such that $\tau$ fixes all vertices in $L_k(v)$ to be 0. 
Define 
\begin{align*}
z_i = R_{u_i}^{\tau} \text{ and } \boldsymbol{z} = (z_1,\ldots,z_d).
\end{align*}
For $w \in L_{k-\ell-1}(u_i)$, the distance from $w$ to $u_i$ is $k - \ell-1 \geq \ell_0 - 1$. By the preceding extension of SSM, $\Vert \boldsymbol{z}^w - \boldsymbol{z} \Vert_\infty \leq \eta$ with $\eta = A \exp(-B (\ell_0 -1))$.
Furthermore, using the same bound at vertex $u$, $|F_u(\boldsymbol{z}^w) - F_u(\boldsymbol{z})| \leq \eta$. 
Define 
\begin{align}\label{eq:C-definition}
C(\boldsymbol{a}) \defeq \frac{\phi(F_u(\boldsymbol{a}))}{\phi(a_i)} \left | \frac{\partial F_u}{\partial z_i}(\boldsymbol{a}) \right |, \qquad \text{so that} \qquad \frac{C(\boldsymbol{z}^w)}{C(\boldsymbol{z})} = \frac{\phi(F_u(\boldsymbol{z}^w))}{\phi(F_u(\boldsymbol{z}))} \cdot \frac{\phi(z_i)}{\phi(\tilde{z}_i^w)} \cdot \frac{\bigl| \frac{\partial F_u}{\partial z_i}(\boldsymbol{z}^w) \bigr|}{\bigl| \frac{\partial F_u}{\partial z_i}(\boldsymbol{z}) \bigr|}.
\end{align}

To analyze the above ratio, we need to use the following lemma.
\begin{lemma}\label{lemma:phi-ratio-bound}
Recall $\phi(x) = \min \{ \frac{1}{t}, \frac{1}{x \log \frac{\lambda}{x}} \}$ 
for some constant $t$.
For any two numbers $a,b \in (0,\lambda)$ with $|a - b| \leq \eta$, it holds that $\frac{\phi(a)}{\phi(b)} \leq 1 + O_{t,\lambda} (\eta)$.
\end{lemma}
\begin{proof}
Note that $x \log \frac{\lambda}{x} \leq \frac{\lambda}{e}$ for all $x \in (0,\lambda)$.
Also note that if $t \geq \frac{\lambda}{e}$, then $\frac{1}{x \log \frac{\lambda}{x}} \geq \frac{1}{t}$ for all $x \in (0,\lambda)$. In this case, $\phi(x) = 1/t$ is a constant and the lemma holds trivially.

Let us assume $t < \frac{\lambda}{e}$. Then, there are two roots to $x \log \frac{\lambda}{x} = t$ in $(0,\lambda)$, denoted by $x_1<x_2$. We have
\begin{align*}
\phi(x) = \begin{cases}
\frac{1}{t} & \text{if } x \in (0,x_1], \\
\frac{1}{x \log \frac{\lambda}{x}} & \text{if } x \in (x_1,x_2), \\
\frac{1}{t} & \text{if } x \in [x_2,\lambda).
\end{cases}
\end{align*}
Since $t$ is a constant, 
we have $x_1$ and $x_2$ are also constants depending on $t$ and $\lambda$. For $x \in (x_1,x_2)$, the derivative $|\phi'(x)|$ is bounded by a constant $c$ depending only on   $t,\lambda$. Hence, the ratio can be bounded by 
\begin{align*}
\frac{\phi(a)}{\phi(b)} \leq 1 + \frac{|\phi(a)-\phi(b)|}{\phi(b)} \leq 1 + \frac{c |a - b|}{C_{\min}} = 1 + O_{t,\lambda} (\eta),
\end{align*}
where $C_{\min} = C_{\min}(t,\lambda)$ is the constant in \Cref{condition:phi-bound}.
\end{proof}

Using \Cref{lemma:phi-ratio-bound}, we can bound the first two terms in \eqref{eq:C-definition} as
\begin{align*}
    \frac{\phi(F_u(\boldsymbol{z}^w))}{\phi(F_u(\boldsymbol{z}))} \cdot \frac{\phi(z_i)}{\phi(\tilde{z}_i^w)}  = \tp{1 + O_{t,\lambda} (\eta)}^2.
\end{align*}
Now, for the last term, recall that $\beta_i = \beta_{u,u_i}$ and $\gamma_i = \gamma_{u,u_i}$, we can write the ratio as 
\begin{align*}
    \frac{\bigl| \frac{\partial F_u}{\partial z_i}(\boldsymbol{z}^w) \bigr|}{\bigl| \frac{\partial F_u}{\partial z_i}(\boldsymbol{z}) \bigr|} = \frac{F_u(\boldsymbol{z}^w)}{F_u(\boldsymbol{z})} \cdot \frac{(\beta_i z_i + 1)(z_i + \gamma_i)}{(\beta_i \tilde{z}_i^w + 1)(\tilde{z}_i^w + \gamma_i)}.
\end{align*}
Let $\beta_j = \beta_{u,u_j}$ and $\gamma_j = \gamma_{u,u_j}$ for all $j \in [d]$. For two numbers $a,b \in (0,\lambda)$ and $|a-b| \leq \eta$, we have
\begin{align*}
\tp{\frac{\beta_j a + 1}{a + \gamma_j}} / \tp{\frac{\beta_j b + 1}{b + \gamma_j}} &\leq 1 + \frac{(\beta_j \gamma_j - 1) |a - b|}{(a + \gamma_j)(\beta_j b + 1)} \leq  1 + |a - b| = 1 + \eta,
\end{align*}
where the last inequality holds because $\beta_j \leq 1$ and $\frac{\beta_j \gamma_j - 1}{(a + \gamma_j)(\beta_j b + 1)} \leq \frac{\gamma_j-1}{\gamma_j} < 1$.  
Similarly, we have the following bound:
\begin{align*}
\frac{(\beta_i a + 1)(a+\gamma_i)}{(\beta_i b + 1)(b+\gamma_i)} &\leq 1 + \frac{(\beta_i (a+b) + \beta_i \gamma_i - 1)|a-b|}{(\beta_i b + 1)(b+\gamma_i)}\\
& \leq 1+\frac{(2\lambda + \gamma_i - 1)|a-b|}{\gamma_i}\leq 1 +O_{\lambda} (\eta),
\end{align*}
where the last inequality holds because $\frac{2\lambda + \gamma_i - 1}{\gamma_i} < 2\lambda + 1$ for all $\gamma_i > 1$.
Using the above two bounds, the last term in \eqref{eq:C-definition} can be bounded as
\begin{align*}
    \frac{\bigl| \frac{\partial F_u}{\partial z_i}(\boldsymbol{z}^w) \bigr|}{\bigl| \frac{\partial F_u}{\partial z_i}(\boldsymbol{z}) \bigr|} \leq \tp{1 + O_{\lambda} (\eta)}^{d+1}.
\end{align*}
Finally, by putting all the bounds together, we have
\begin{align*}
    \frac{C(\boldsymbol{z}^w)}{C(\boldsymbol{z})} = \frac{\phi(F_u(\boldsymbol{z}^w))}{\phi(F_u(\boldsymbol{z}))} \cdot \frac{\phi(z_i)}{\phi(\tilde{z}_i^w)} \cdot \frac{\bigl| \frac{\partial F_u}{\partial z_i}(\boldsymbol{z}^w) \bigr|}{\bigl| \frac{\partial F_u}{\partial z_i}(\boldsymbol{z}) \bigr|} \leq \tp{1 + O_{t,\lambda} (\eta)}^{d+3}.
\end{align*}
The sum of the influence in~\eqref{eq:sum-K-u} now can be bounded by
\begin{align}\label{eq:sum-K-u-bound}
    \sum_{w \in L_{k-\ell}(u)} K_u^w &\leq \sum_{i=1}^d \sum_{w \in L_{k-\ell-1}(u_i)} \frac{\phi(F_u(\boldsymbol{z}^w))}{\phi(\tilde{z}_i^w)} \left | \frac{\partial F_u}{\partial z_i}(\boldsymbol{z}^w) \right | \cdot K_{u_i}^w \\
    &\leq \tp{1 + O_{t,\lambda} (\eta)}^{d+3} \sum_{i=1}^d \frac{\phi(F_u(\boldsymbol{z}))}{\phi(z_i)} \left | \frac{\partial F_u}{\partial z_i}(\boldsymbol{z}) \right | \sum_{w \in L_{k-\ell-1}(u_i)} K_{u_i}^w \notag \\
    &\leq \tp{1 + O_{t,\lambda} (\eta)}^{d+3}\tp{\sum_{i=1}^d \frac{\phi(F_u(\boldsymbol{z}))}{\phi(z_i)} \left | \frac{\partial F_u}{\partial z_i}(\boldsymbol{z}) \right |} \cdot \tp{\max_{i \in [d]} \sum_{w \in L_{k-\ell-1}(u_i)} K_{u_i}^w}.\notag
\end{align}
For the middle term in the above formula, using \Cref{condition:trivial-bound-for-C} and \Cref{condition:potential-function-bound}, we have 
\begin{align*}
    \sum_{i=1}^d \frac{\phi(F_u(\boldsymbol{z}))}{\phi(z_i)} \left | \frac{\partial F_u}{\partial z_i}(\boldsymbol{z}) \right | \leq \min \left\{ 1 - \alpha,  C_{\text{trl}} \cdot  d\lambda_u  \exp(-C_{\text{decay}} d) \right\} = \min \left\{1 - \alpha, C_1 \exp(-C_2 d)\right\},
\end{align*}
where $\alpha < 1$ is the constant in \Cref{condition:potential-function-bound} and $C_{\text{trl}}, C_{\text{decay}}$ are the constants in \Cref{condition:trivial-bound-for-C}.
Note that $\lambda_u \leq \lambda$ for the constant $\lambda$ in \Cref{condition}, so the second bound is upper bounded by $dC_1 \exp(-C_2 d)$ for some constants $C_1, C_2 > 0$. We can choose sufficiently large constants $d_0$ and $\ell_0$ such that the following holds. 
If $d > d_0$, we use
\begin{align*}
     (1+O_{t,\lambda} (\eta))^{d+3} \cdot d C_1 \exp(-C_2 d)\leq d C_1 (1+O_{t,\lambda} (\eta))^{3} \cdot \exp((-C_2+ O_{t,\lambda} (\eta)) d).
\end{align*}
By choosing $\ell_0$ large enough, we can make sure that $\eta = A \exp(-B(\ell_0 - 1))$ is sufficiently small so that $-C_2 + O_{t,\lambda}(\eta) < -C_2/2$. 
Since $d \ge d_0$, by taking the constant $d_0$ sufficiently large, the whole term is bounded by $1 - \alpha^2$. 
If $d \leq d_0$, then
\begin{align*}
     (1+O_{t,\lambda} (\eta))^{d+3} \cdot (1-\alpha) \leq   (1+O_{t,\lambda} (\eta))^{d_0+3} \cdot (1-\alpha) \leq 1 - \alpha^2,
\end{align*}
where the last inequality holds by choosing $\ell_0$ large enough so that $\eta$ is small enough and the $(1+O_{t,\lambda} (\eta))^{d_0+3}$ term is at most $1+\alpha$. Combining the two cases, the lemma holds with $\delta = \alpha^2$.
\end{proof}

\subsection{Proof of the influence bound} \label{subsec:proof-of-the-influence-bound}

We are now ready to prove the influence bound. Using~\eqref{eq:bound-on-dtv}, we bound the sum of the influence level by level. Fix an integer $k \geq 1$, to bound the sum $\sum_{w \in L_k(v)} K_v^w$, we apply \Cref{lemma:general-influence-decay-1} and \Cref{lemma:general-influence-decay-2}. Formally, we first truncate the tree $T$ and only keep levels up to $k$ to form a new tree $T_k$. By definition of $I_v^w$, for every $w$, it fixes the pinning on the $k$-th level of $T_k$. Hence, we can only consider the tree $T_k$ when analysing the influence. Using \Cref{lemma:general-influence-decay-1} and \Cref{lemma:general-influence-decay-2} recursively along a maximizing branch, the applications of \Cref{lemma:general-influence-decay-2} contribute a factor $(1-\delta)^{\Omega(k)}$, and the remaining $O(\ell_0)$ applications of \Cref{lemma:general-influence-decay-1} contribute only a constant factor, since $C_\text{trl}\lambda_x d_x\exp(-C_{\text{decay}}d_x)=O(1)$ uniformly over all vertices $x$. Thus we reach a vertex $u$ at level $k-1$ with children $u_1,u_2,\ldots,u_d$. Formally,
\begin{align*}
    \sum_{w \in L_k(v)} K_v^w &\leq (1-\delta)^{\max\{0,k-\ell_0+1\}} \cdot \tp{C_\text{trl} \lambda_u d \exp(-C_{\text{decay}}d)}^{\max\{0,\min\{\ell_0-2,k-2\}\}} \cdot \sum_{i=1}^d K_{u}^{u_i} \\
    &\leq O(1) \cdot (1 - \delta)^k \cdot \sum_{i=1}^d K_{u}^{u_i}.
\end{align*}
Note that in these recursive applications, the underlying systems are induced subtrees of the original computation tree, with possible pinnings absorbed by self-reducibility. Hence the edge parameters are inherited and the external fields can only decrease, so the systems remain $(\beta,\gamma,\lambda)$-ferromagnetic two-spin systems, respectively $(\gamma,\lambda)$-RBMs. By the verification preceding \Cref{condition}, they also satisfy \Cref{condition}.
Finally, we bound each $K_{u}^{u_i}$. By definition of the influence in~\Cref{def:influence-from-one-vertex}, we can write the influence as 
\begin{align*}
    K_{u}^{u_i} = \left\vert \Phi(R_u^{\sigma^i \land u_i \gets \infty}) - \Phi(R_u^{\sigma^i \land u_i \gets 0}) \right\vert,
\end{align*}
where $\sigma^i$ is a pinning on all $u_j$ with $j \neq i$ and $\sigma^i(u_j) \in (0,\lambda)$ for all $j \neq i$. A simple calculation shows 
\begin{align*}
\Vert R^{\sigma^i \land u_i \gets \infty} - R^{\sigma^i \land u_i \gets 0} \Vert &\leq \lambda \prod_{j: j\neq i}{\frac{\beta_{u,u_j} \lambda + 1}{\lambda + \gamma_{u,u_j}}} \cdot \tp{ \frac{\beta_{u,u_i}\gamma_{u,u_i} - 1}{\gamma_{u,u_i}}}\\
\text{($\beta_{e} \leq 1,\gamma_{e} \geq \gamma > 1$ for all edges $e$)}\quad&\leq \lambda \prod_{j: j\neq i}{\frac{ \lambda + 1}{\lambda + \gamma}} \\
\text{($\lambda > 0$ and $\gamma > 1$ are constants)}\quad &\leq \lambda \exp(-\Omega(d)).
\end{align*}
Using \Cref{lemma:dtv-bound-potential}, we have
\begin{align*}
\sum_{i=1}^d K_{u}^{u_i} \leq \sum_{i=1}^d O(1) \cdot \Vert R^{\sigma^i \land u_i \gets \infty} - R^{\sigma^i \land u_i \gets 0} \Vert \leq O(1) \cdot d \cdot \exp(-\Omega(d))= O(1).
\end{align*}
Finally, combining~\eqref{eq:bound-on-dtv}, \Cref{lemma:dtv-bound-potential}, and the above bounds, the total influence is bounded by
\begin{align*}
    \sum_{w \in V \setminus \{v\}} \DTV{\mu_v^{w \gets 0}}{\mu_v^{w \gets 1}} &\leq 2\sum_{k \geq 1} \sum_{w \in L_k(v)} I_v^w \leq O(1) \sum_{k \geq 1} \sum_{w \in L_k(v)} K_v^w\\
    &\leq O(1) \sum_{k \geq 1} (1 - \delta)^k = O(1).
\end{align*}



\section{Mixing from typical-case aggregate strong spatial mixing}

The ferromagnetic two-spin systems are monotone systems.
To make this notion precise, recall that $\mu^\sigma$ denotes the distribution of $X \sim \mu$ conditional on $X(\Lambda) = \sigma$,
where $\Lambda \subseteq V$ is a subset of vertices and $\sigma \in \{0,1\}^{\Lambda}$ is a configuration on $\Lambda$. 
Define a partial ordering $\preceq$ as follows. 
For any $\Lambda \subseteq V$, any two configurations $\sigma,\tau \in \{0,1\}^{\Lambda}$, 
\begin{align}
  \sigma \preceq \tau \quad \Leftrightarrow \quad \sigma_v \leq \tau_v \quad \forall v \in \Lambda.\label{eq:partial-order-on-configurations}
\end{align}

\begin{definition}[Monotone spin systems]
 A two-spin system is said to be monotone if for any $\Lambda \subseteq V$, any two configurations $\sigma,\tau \in \{0,1\}^{\Lambda}$, if $\sigma \preceq \tau$, then $\mu^\sigma$ is stochastically dominated by $\mu^\tau$, which means that there exists a coupling $(X,Y)$ such that $X \sim \mu^\sigma$ and $Y \sim \mu^\tau$ and $\Pr[X \preceq Y] = 1$.
\end{definition}

As a well-known fact, any Gibbs distribution of ferromagnetic two-spin system is a monotone spin system. We provide a proof for the sake of completeness in Appendix \ref{app:censoring}.
\begin{proposition}\label{prop:monotone-of-ferro-ising}
    Any Gibbs distribution of ferromagnetic two-spin system is a monotone spin system.
\end{proposition}

We study the block dynamics on two-spin systems with Gibbs distribution $\mu$. Let $\+B = \{B_1,B_2,\ldots,B_r\}$ be a set of blocks, where each block $B_i \subseteq V$ and $\cup_{i=1}^r B_i = V$. 
We consider two kinds of block dynamics: heat-bath block dynamics and systematic scan block dynamics. 

Starting from an initial configuration $X\in \Omega = \{0,1\}^V$, in each step, the heat-bath block dynamics updates the current configuration $X$ as follows:
\begin{itemize}
    \item pick a block $B$ uniformly at random from $\+B$;
    \item resample $X(B)\sim \mu_{B}^{X(V \setminus B)}$, where $\mu_{B}^{X(V \setminus B)}$ is the marginal distribution on $B$ induced by $\mu$ conditioned on the configuration $X(V \setminus B)$ on other variables $V \setminus B$ outside of $B$.
\end{itemize}
The systematic scan block dynamics updates the current configuration $X$ as follows: for each update step,
\begin{itemize}
  \item scan all the blocks $B_i$ for $i$ from 1 to $r$ in order, and resample the configuration on $B_i$ conditional on the current configuration of other variables: $X(B_i) \sim \mu_{B_i}^{X(V \setminus B_i)}$.
\end{itemize}

For each block $B_i$, let $P_{B_i}$ denote the transition matrix of updating the configuration on $B_i$ conditional on the current configuration of other variables. The transition matrix of heat-bath block dynamics is then 
\begin{align*}
  P_{\text{HB}} = \frac{1}{r}\sum_{i=1}^r P_{B_i},
\end{align*}
and the transition matrix of systematic scan block dynamics is 
\begin{align*}
  P_{\text{Scan}} = P_{B_r} \cdot P_{B_{r-1}} \cdots P_{B_1}.
\end{align*}
The result in this section works for both the heat-bath block dynamics and the systematic scan block dynamics.
In the rest of the proof in this section, we use the phrase ``block dynamics'' to refer to both the heat-bath block dynamics and the systematic scan block dynamics.

As before, the mixing time of block dynamics is defined as the number of steps until the configuration $X$ is close to the stationary distribution $\mu$ in total variation distance. Formally, let $P:\Omega \times \Omega \to [0,1]$ be the transition matrix of the block dynamics. Then, the mixing time is defined as
\begin{align*}
\forall \epsilon > 0, \quad t_{\textnormal{mix}}^{P}(\epsilon) = \max_{\sigma \in \Omega}\min\left\{t \geq 0: \DTV{P^t(\sigma,\cdot)}{\mu} < \epsilon\right\}.
\end{align*}

Monotone systems admit monotone grand couplings.
The following standard result applies to $P_{\text{HB}}$ and $P_{\text{Scan}}$. For the sake of completeness, we provide a proof in Appendix \ref{app:censoring}.
\begin{proposition}[Monotone grand coupling of block dynamics]\label{prop:monotone-coupling}
Let $\mu$ be a Gibbs distribution of a ferromagnetic two-spin system on graph $G=(V,E)$. Let $P$ be a block dynamics on $\mu$. Then, there exists a monotone coupling function $f: \Omega \times [0,1] \to \Omega$ such that for any $\sigma \in \Omega$, real vector $r \in [0,1]^{n+1}$ uniformly at random, $\sigma \to \tau$ where $\tau = f(\sigma,r)$ follows the law of $P$. Furthermore, for any $\sigma \preceq \sigma'$, it holds that 
\begin{align*}
\Pr_r[f(\sigma,r) \preceq f(\sigma',r)] = 1.
\end{align*}
\end{proposition}

To analyse this grand coupling, due to the monotonicity, it suffices to consider two chains starting from all-one configuration $\*1$ and all-zero configuration $\*0$.

\begin{definition}\label{def:monotone-coupling}
Let $(r_t)_{t \geq 1}$ be a sequence of independent uniformly random real vectors in $[0,1]^{n+1}$. Let $X^+_0$ be the all-ones configuration and $X^-_0$ be the all-zeros configuration. Define the monotone coupling $(X^+_t,X^-_t)_{t \geq 0}$ as for any $t \geq 1$, $X^+_t = f(X^+_{t-1},r_t)$ and $X^-_t = f(X^-_{t-1},r_t)$, where $f(\cdot,\cdot)$ is the monotone coupling function in \Cref{prop:monotone-coupling}.
\end{definition}

In addition, to facilitate the analysis later, define the following censored block dynamics.

\begin{definition}[Censored block dynamics]
    Let $\mu$ be the Gibbs distribution of a ferromagnetic two-spin system on graph $G=(V,E)$.
    Let $P:\Omega \times \Omega \to [0,1]$ be the transition matrix of a block dynamics on $\mu$ with a set of blocks $\+B = \{B_1,B_2,\ldots,B_r\}$. For any subset $S \subseteq V$, any pinning $\sigma \in \{0,1\}^{V \setminus S}$,  the censored block dynamics $P_S$ on $\mu_S^\sigma$ is defined as follows. 
    \begin{itemize}
        \item The Markov chain starts from an arbitrary $X \in \{0,1\}^V$ with $X(V \setminus S) = \sigma$.
    \end{itemize}
    For the heat-bath block dynamics, in each step, 
    \begin{itemize}
        \item sample $B \in \+B$ uniformly at random, and resample the configuration on $B \cap S$ conditional on the current configuration of other variables: $X(B \cap S) \sim \mu_{B \cap S}^{X(V \setminus (B \cap S))}$.
    \end{itemize}
    For the systematic scan block dynamics, in each step,
    \begin{itemize}
        \item scan all the blocks $B_i$ for $i$ from 1 to $r$ in order, and resample the configuration on $B_i \cap S$ conditional on the current configuration of other variables: $X(B_i \cap S) \sim \mu_{B_i \cap S}^{X(V \setminus (B_i \cap S))}$.
    \end{itemize}
    \end{definition}
    
    The censored block dynamics $P_S^{\textnormal{censored}}$ only updates the configuration on $S$ while keeping the configuration on $V \setminus S$ fixed.
    Intuitively, updates outside of $S$ are ``censored''.
    During the whole process, the configuration on $V \setminus S$ is fixed as $\sigma$.
    Let $(X_t)_{t \geq 0}$ be the Markov chain generated by $P_S^{\textnormal{censored}}$ on $\mu_S^\sigma$. 
    As before, the mixing time of censored block dynamics $P_S^{\textnormal{censored}}$ on $\mu_S^\sigma$ is
    \begin{align*}
      \forall \epsilon>0, \quad t^{P_S^{\textnormal{censored}},\mu_S^\sigma}_{\textnormal{mix}}(\epsilon) =\max_{X_0:X_0(V \setminus S) = \sigma}\min\left\{t \geq 0: \DTV{(P_S^{\textnormal{censored}})^t(X_0,\cdot)}{\mu_S^\sigma} < \epsilon\right\}.
    \end{align*}

    The key to our proof is the notion of good neighbourhood and boundary conditions, which facilitates typical-case ASSM.
Let $S \subseteq V$ be a subset of vertices. 
The outer boundary $\partial S$ of $S$ is the set of vertices $v \in V \setminus S$ such that there exists an edge $\{u,v\} \in E$ with $u \in S$.

\begin{definition}\label{def:correlation-decay}
  For any $v \in V$, we call a neighbourhood $S_v \ni v$ and a set of boundary conditions $\Omega_{\partial S_v} \subseteq \{0,1\}^{\partial S_v}$ \emph{good} with local mixing time $T_{\text{local}}$ if the following three properties hold:
\begin{itemize}
  \item \textbf{Closed under shortest paths.} For any $\sigma,\tau \in \Omega_{\partial S_v}$, 
    there exists a path of good boundary configurations $\eta_0,\eta_1,\ldots,\eta_t \in \Omega_{\partial S_v}$ such that $\eta_0 = \sigma$, $\eta_t = \tau$, 
    and for any $1 \leq i \leq t$, $\eta_i$ and $\eta_{i+1}$ differ only at one vertex, where $t = \abs{\{\sigma(u) \neq \tau(u): u \in \partial S_v\}}$ is the Hamming distance between $\sigma$ and $\tau$.
  \item \textbf{ASSM under good boundary conditions.} For any $u \in \partial S_v$, define the influence of $u$ on $v$ as
    \begin{align}\label{eq:influence}
      a_u \defeq \max_{\sigma \in \Omega_{\partial S_v}} \DTV{\mu^{\sigma^{u \gets 0}}_v}{\mu^{\sigma^{u \gets 1}}_v},
    \end{align}
    where $\sigma^{u \gets c}$ denotes the configuration on $\partial S_v$ obtained from $\sigma$ by changing the value of $u$ to $c$. Then, the following aggregate strong spatial mixing (ASSM) property holds
    \begin{align}\label{eq:assm-with-good-boundary}
      \sum_{u \in \partial S_v} a_u \leq \frac{1}{20}.
    \end{align}
  \item \textbf{Local mixing.} For any outside configuration $\sigma \in \{0,1\}^{V \setminus S_v}$, 
    the censored block dynamics $P_{S_v}^{\textnormal{censored}}$ on $\mu^\sigma_{S_v}$ has mixing time $t^{P_{S_v}^{\textnormal{censored}},\mu_{S_v}^\sigma}_{\textnormal{mix}}(\frac{1}{4e}) \leq T_{\text{local}}$.
\end{itemize}
\end{definition}

Now we are ready to show the main theorem of this section.

\begin{theorem}\label{thm:mixing}
Let $\mu$ be the Gibbs distribution of a ferromagnetic two-spin system on graph $G=(V,E)$. Let $P$ be a block dynamics on $\mu$ with a set $\+B$ of blocks. 
Let $T_{\textnormal{local}} > 0$ and $T_{\textnormal{burn-in}} > 0$ be two integers. 
Suppose for any $v \in V$, there exists $S_v \subseteq V$ and $\Omega_{\partial S_v} \subseteq \{0,1\}^{\partial S_v}$ such that
\begin{itemize}
  \item $(S_v,\Omega_{\partial S_v})$ is good with local mixing time $T_{\text{local}}$ as in \Cref{def:correlation-decay};
    \item the monotone coupling $(X_t^+,X_t^-)_{t \geq 0}$ of $P$ in \Cref{def:monotone-coupling} satisfies that 
    for any $t \geq T_{\textnormal{burn-in}}$,
    \begin{align}\label{eq:burn-in}
         \Pr[X^+_t(\partial S_v) \notin \Omega_{\partial S_v} \lor X^-_t(\partial S_v) \notin \Omega_{\partial S_v}] \leq \frac{1}{n^3},
    \end{align}
    where $n = |V|$ is the number of vertices.
\end{itemize}
Then the mixing time of block dynamics $P$ is bounded by
\begin{align}\label{eq:mixing-time-bound}
t_{\textnormal{mix}}^P\left(\frac{1}{4e}\right) = O\left(T_{\textnormal{burn-in}} + T_{\textnormal{local}} \cdot \max_{v \in V}\log |R_v|\cdot \log n\right), \quad \text{where } R_v = S_v \cup \partial S_v.
\end{align}
\end{theorem}

In \eqref{eq:mixing-time-bound} we set $\epsilon=1/(4e)$ for convenience later.
It is standard to extend it to general $\epsilon>0$.
The proof of \Cref{thm:mixing} follows similar lines as in \cite{MS13}.

\begin{proof}[Proof of \Cref{thm:mixing}]
Let $(X_t^+,X_t^-)_{t \geq 0}$ be the monotone coupling of $P$ in \Cref{def:monotone-coupling}. 
Define $T_{\textnormal{phase}} \defeq T_{\textnormal{local}} \cdot \max_{v \in V} \log \left(20 |R_v|\right)$. We show that for any integer $k \geq 1$, it holds that
\begin{align}\label{eq:recursion}
    &\max_{v \in V} \Pr\left[X^+_{T_{\textnormal{burn-in}} + (k+1) \cdot T_{\textnormal{phase}}}(v) \neq X_{T_{\textnormal{burn-in}} + (k+1) \cdot T_{\textnormal{phase}}}^-(v)\right]\notag\\
    \leq\,& \frac{1}{2} \max_{v \in V} \Pr\left[X^+_{T_{\textnormal{burn-in}} + k \cdot T_{\textnormal{phase}}}(v) \neq X_{T_{\textnormal{burn-in}} + k \cdot T_{\textnormal{phase}}}^-(v)\right] + \frac{1}{n^2}.
\end{align}
Solving the recursion in \eqref{eq:recursion}, after $T \defeq T_{\textnormal{burn-in}} + O(T_{\textnormal{phase}}\log n)$ steps, 
\begin{align*}
\max_{v \in V} \Pr\left[X^+_T(v) \neq X^-_T(v)\right] \leq \tp{\frac{1}{2}}^{O(\log n)} + \frac{2}{n^2} \leq \frac{3}{n^2}. 
\end{align*}
By a union bound over all $v \in V$, it holds that $\Pr[X^+_T \neq X^-_T] \leq \frac{3}{n} \leq \frac{1}{4e}$.
This holds for two chains starting from the all-one configuration $\*1$ and all-zero configuration $\*0$.
By monotonicity, namely~\Cref{prop:monotone-coupling}, starting from an arbitrary pair of initial configurations, the two chains can be coupled successfully with probability at least $1-\frac{1}{4e}$. Therefore, by the standard coupling argument, the mixing time bound in~\eqref{eq:mixing-time-bound} is proved. Our task is reduced to verify the recursion in~\eqref{eq:recursion}.

Fix an integer $k \geq 0$. Let $s = k \cdot T_{\textnormal{phase}} + T_{\textnormal{burn-in}}$.
Fix a vertex $v \in V$ and the corresponding region $S_v\subseteq V$.
We construct another two instances of Markov chains $(Y_j^+,Y_j^-)_{j \geq 0}$ by the following process:
\begin{itemize}
    \item for $0 \leq j \leq s$, let $(Y_j^+,Y_j^-) = (X_j^+,X_j^-)$;
    \item for $j > s$, the two processes $Y_{j-1}^+ \to Y_j^+$  and $Y_{j-1}^- \to Y_j^-$ both follow the transition rule of the censored block dynamics $P_{S_v}^{\textnormal{censored}}$.
\end{itemize}

For two random variables $X$ and $Y$ over $\{0,1\}^V$, we say that the distribution of $X$ is stochastically dominated by the distribution of $Y$, denoted by $X \preceq_D Y$, if there exists a coupling $(X,Y)$ such that $X \preceq Y$ with probability 1, where the partial order $\preceq$ is defined in \eqref{eq:partial-order-on-configurations}.
The following result holds for the censored block dynamics $P_{S_v}^{\textnormal{censored}}$. 
A similar result appeared in~\cite[Theorem 7]{BlancaCV20}. For the sake of completeness, we provide a brief proof in Appendix \ref{app:censoring}.
\begin{claim}\label{claim:censoring}
  The following stochastic dominance relations hold:
  \[\forall j \geq 0, \quad Y_j^-\preceq_D X_j^- \preceq_D X_j^+\preceq_D Y_j^+.\] 
\end{claim}

\Cref{claim:censoring} states a stochastic dominance relation among four random variables $X_j^-, X_j^+, Y_j^-, Y_j^+$. The statement itself only involves the marginal distribution of four random variables. For instance, the distribution of $Y^-_j$ is stochastic dominated by the distribution of $X_j^-$. The claim itself states nothing about the joint distribution of four random variables.

Recall that $R_v=S_v\cup \partial S_v$. Let $t = s + T_{\textnormal{phase}}$.
Since $(X_t^+,X_t^-)_{t \geq 0}$ forms a monotone coupling, we have $X_t^+(v) \geq X_t^-(v)$ with probability 1.
To upper bound the probability of $X_t^+(v) \neq X_t^-(v)$, we only need to upper bound $Pr[X_t^+(v) = 1] - Pr[X_t^-(v) = 1]$.
The stochastic dominance relations in \Cref{claim:censoring} shows
\begin{align*}
\Pr[X_t^-(v) = 1] \geq \Pr[Y_t^-(v) = 1] \text{ and } \Pr[X_t^+(v) = 1] \leq \Pr[Y_t^+(v) = 1].
\end{align*}
Therefore, we have the following upper bound:
\begin{align}\label{eq:upper-bound-on-X-t-plus-v-neq-X-t-minus-v}
\Pr[X_t^+(v) \neq X_t^-(v)] = Pr[X_t^+(v) = 1] - Pr[X_t^-(v) = 1] \leq \Pr[Y_t^+(v) = 1] - \Pr[Y_t^-(v) = 1].
\end{align}
For any two configurations $\sigma^+,\sigma^-\in \{0,1\}^{R_v}$, let $\mathcal{C}(\sigma^+,\sigma^-)$ be the event
$X_s^+(R_v)=\sigma^+$ and $X_s^-(R_v)=\sigma^-$. 
We only consider $\sigma^+,\sigma^-$ such that $\+C(\sigma^+,\sigma^-)$ happens with a positive probability.
For $t > s$, we will 
upper bound the difference between the probabilities of $Y_t^+(v) = 1$ and $Y_t^-(v) = 1$
conditioned on $\mathcal{C}(\sigma^+,\sigma^-)$. 
Let $\tau^+ = \sigma^+( \partial S_v)$ and $\tau^- = \sigma^-(\partial S_v)$ be the configurations on the boundary $\partial S_v$ induced by $\sigma^+$ and $\sigma^-$ respectively.
We also define $\+C(\tau^+,\tau^-)$ be the event $X_s^+(\partial S_v) = \tau^+$ and $X_s^-(\partial S_v) = \tau^-$.
By the triangle inequality, we have 
\begin{align}
\begin{split}\label{eq:triangle}
    &\abs{\Pr[Y_t^+(v) = 1 \mid \mathcal{C}(\sigma^+,\sigma^-)] - \Pr[Y_t^-(v) = 1 \mid \mathcal{C}(\sigma^+,\sigma^-)]}\\
    \leq\,& \left | \P[Y_t^+(v)=1 \mid \mathcal{C}(\sigma^+,\sigma^-)] -\mu_v^{\tau^+}(1) \right | + \left | \mu_v^{\tau^+}(1)-\mu_v^{\tau^-}(1)\right |\\
    &\quad + \left | \P[Y_t^-(v)=1 \mid \mathcal{C}(\sigma^+,\sigma^-)] -\mu_v^{\tau^-}(1) \right |,
\end{split}
\end{align}
By the law of total probability and the triangle inequality, the probability of $Y^+_t (v) \neq Y_t^-(v)$ is at most
\begin{align}\label{eq:chain-rule}
    & \Pr[Y_t^+(v) = 1] - \Pr[Y_t^-(v) = 1] \notag\\
    =\,& \sum_{(\sigma^+,\sigma^-):\sigma^+ \neq \sigma^-} \Pr[\mathcal{C}(\sigma^+,\sigma^-)] \cdot (\P[Y_t^+(v) = 1 \mid \mathcal{C}(\sigma^+,\sigma^-)] - \P[Y_t^-(v) = 1 \mid \mathcal{C}(\sigma^+,\sigma^-)])\notag\\
    \leq\,& \sum_{(\sigma^+,\sigma^-):\sigma^+ \neq \sigma^-} \Pr[\mathcal{C}(\sigma^+,\sigma^-)] \cdot |\P[Y_t^+(v) = 1 \mid \mathcal{C}(\sigma^+,\sigma^-)] - \P[Y_t^-(v) = 1 \mid \mathcal{C}(\sigma^+,\sigma^-)]|
\end{align}
Note that the sum above enumerates only pairs of distinct feasible boundary configurations $\sigma^+,\sigma^- \in \{0,1\}^{R_v}$, namely $\sigma^+ \neq \sigma^-$.
This is because, when $\sigma^+ = \sigma^-$, using the conditional independence property of spin systems, two Markov chains $Y^+_t(v)$ and $Y^-_t(v)$ are exactly the same stochastic processes (the same starting configuration and same transition matrix) inside $S_v$, and therefore $\Pr[Y_t^+(v) = 1 \mid \mathcal{C}(\sigma^+,\sigma^-)] = \Pr[Y_t^-(v) = 1 \mid \mathcal{C}(\sigma^+,\sigma^-)]$. 
Combining~\eqref{eq:triangle} and~\eqref{eq:chain-rule}, we have
\begin{align}
\begin{split}\label{eq:sum-of-three}
    &\Pr[Y_t^+(v) = 1] - \Pr[Y_t^-(v) = 1]\\
     \leq\,& \sum_{(\sigma^+,\sigma^-):\sigma^+\neq \sigma^-} \Pr[\mathcal{C}(\sigma^+,\sigma^-)] \left | \P[Y_t^+(v)=1 \mid \mathcal{C}(\sigma^+,\sigma^-)] -\mu_v^{\tau^+}(1) \right |\\
    &+\sum_{(\sigma^+,\sigma^-):\sigma^+\neq \sigma^-} \Pr[\mathcal{C}(\sigma^+,\sigma^-)] \left | \mu_v^{\tau^+}(1)-\mu_v^{\tau^-}(1)\right |\\
    &+\sum_{(\sigma^+,\sigma^-):\sigma^+\neq \sigma^-} \Pr[\mathcal{C}(\sigma^+,\sigma^-)] \left | \P[Y_t^-(v)=1 \mid \mathcal{C}(\sigma^+,\sigma^-)] -\mu_v^{\tau^-}(1) \right |.    
\end{split}
\end{align}

Consider the first and the third terms in~\eqref{eq:sum-of-three}.
Note that $Y_t^+(v)$ and  $Y_t^-(v)$ both follow the censored transition matrix $P_{S_v}^{\textnormal{censored}}$. The configuration outside $S_v$ is fixed in the censored process, and the configuration inside $S_v$ converges to the conditional marginal distribution $\mu_v^{\tau^+}$ and $\mu_v^{\tau^-}$ respectively. 
Therefore, by the local mixing property of \Cref{def:correlation-decay} and since $ t- s = T_{\textnormal{local}} \cdot \max_{v \in V}\log (20 | R_v|)$, by \eqref{eq:general-bound-on-mixing-time},
\begin{align*}
    \forall \sigma^+,\sigma^- \in \{0,1\}^{R_v},\quad
\begin{split}
    \left | \P[Y_t^+(v)=1\mid\mathcal{C}(\sigma^+,\sigma^-)] -\mu_v^{\tau^+}(1) \right | &\leq \frac{1}{20|R_v|};\\
    \left | \P[Y_t^-(v)=1\mid\mathcal{C}(\sigma^+,\sigma^-)] -\mu_v^{\tau^-}(1) \right | &\leq \frac{1}{20|R_v|}.
\end{split}
\end{align*}
Therefore the first and the third terms in~\eqref{eq:sum-of-three} can be bounded by
\begin{align}
\begin{split}\label{eq:term-1-and-3}
    &\sum_{(\sigma^+,\sigma^-):\sigma^+\neq \sigma^-} \Pr[\mathcal{C}(\sigma^+,\sigma^-)] \left | \P[Y_t^+(v)=1 \mid \mathcal{C}(\sigma^+,\sigma^-)] -\mu_v^{\tau^+}(1) \right |\\
    &\qquad +\sum_{(\sigma^+,\sigma^-):\sigma^+\neq \sigma^-} \Pr[\mathcal{C}(\sigma^+,\sigma^-)] \left | \P[Y_t^-(v)=1 \mid \mathcal{C}(\sigma^+,\sigma^-)] -\mu_v^{\tau^-}(1) \right |\\
    \leq& \sum_{(\sigma^+,\sigma^-):\sigma^+\neq \sigma^-}  \frac{\Pr[\mathcal{C}(\sigma^+,\sigma^-)]}{20|R_v|} + \sum_{(\sigma^+,\sigma^-):\sigma^+\neq \sigma^-}  \frac{\Pr[\mathcal{C}(\sigma^+,\sigma^-)]}{20|R_v|}\\
    =&\frac{\Pr[X_s^+(R_v)\neq X_s^-(R_v)]}{10|R_v|}\leq \frac{\max_{u\in V} \Pr[X_s^+(u)\neq X_s^-(u)]}{10}.
\end{split}
\end{align}

To bound the second term in~\eqref{eq:sum-of-three}, we first have that
\begin{align*}
&\sum_{(\sigma^+,\sigma^-):\sigma^+\neq \sigma^-} \Pr[\mathcal{C}(\sigma^+,\sigma^-)] \left | \mu_v^{\tau^+}(1)-\mu_v^{\tau^-}(1)\right |\\
=\,& \sum_{(\tau^+,\tau^-) \in \{0,1\}^{\partial S_v} \times \{0,1\}^{\partial S_v}:\tau^+\neq \tau^-}\P[\mathcal{C}(\tau^+,\tau^-)]\left | \mu_v^{\tau^+}(1)-\mu_v^{\tau^-}(1)\right |,
\end{align*}
because whenever $\tau^+ = \sigma^+(\partial S_v)=\sigma^-(\partial S_v) = \tau^-$, it holds that $\mu_v^{\tau^+}(1) = \mu_v^{\tau^-}(1)$.
We then construct a path $\eta_0,\eta_1,\ldots,\eta_t \in \{0,1\}^{\partial S_v}$ such that $\eta_0 = \tau^+$, $\eta_t = \tau^-$, and for any $1 \leq i \leq t$, $\eta_i$ and $\eta_{i+1}$ differ only at one vertex, where $t = \{\tau^+(u) \neq \tau^-(u): u \in \partial S_v\}$ is the Hamming distance between $\tau^+$ and $\tau^-$.
There are two cases depending on whether both $\tau^+$ and $\tau^-$ are in $\Omega_{\partial S_v}$. 
If so, by the first property of \Cref{def:correlation-decay}, we can further assume that $\eta_i \in \Omega_{\partial S_v}$ for all $0 \leq i \leq t$.
Then,
\begin{align*}
    &\sum_{\tau^+\neq \tau^-}\P[\mathcal{C}(\tau^+,\tau^-)]\left | \mu_v^{\tau^+}(1)-\mu_v^{\tau^-}(1)\right |\leq \sum_{\tau^+\neq \tau^-}\P[\mathcal{C}(\tau^+,\tau^-)]\sum_{i=1}^{t}\left|\mu_v^{\eta_{i-1}}(1)-\mu_v^{\eta_i}(1) \right|\\
    \leq& \sum_{\tau^-\neq \tau^+}\P[\mathcal{C}(\tau^+,\tau^-)]\sum_{u\in \partial S_v}\mathbb{1}\{\tau^+(u)\neq \tau^-(u) \}(\mathbb{1}\{\tau^+,\tau^-\in \Omega_{\partial S_v}\}a_u+\mathbb{1}\{\tau^+\text{ or }\tau^-\notin \Omega_{\partial S_v}\}\cdot 1),
\end{align*}
where in the last inequality, we split the two cases. 
If both $\tau^+$ and $\tau^-$ are in $\Omega_{\partial S_v}$, then $\eta_i \in \Omega_{\partial S_v}$ for all $0 \leq i \leq t$.
It implies that the difference between $\mu_v^{\eta_{i-1}}(1)$ and $\mu_v^{\eta_i}(1)$ is at most $a_u$, where $u$ is the vertex that $\eta_{i-1}$ and $\eta_{i}$ differ on and $a_u$ is defined in \eqref{eq:influence}. 
Otherwise $\tau^+$ or $\tau^-$ is not in $\Omega_{\partial S_v}$, then the difference between $\mu_v^{\eta_{i-1}}(1)$ and $\mu_v^{\eta_i}(1)$ is at most 1. Rearranging the terms, we have
\begin{align*}
  &\sum_{u\in \partial S_v}a_u\cdot \sum_{\tau^+ \neq \tau^-}\P[\mathcal{C}(\tau^+,\tau^-)]\cdot \mathbb{1}\{\tau^+(u)\neq \tau^-(u) \}\cdot\mathbb{1}\{\tau^+, \tau^-\in \Omega_{\partial S_v}\}\\
    &\quad +\sum_{u\in \partial S_v}\sum_{\tau^+ \neq \tau^-}\P[\mathcal{C}(\tau^+,\tau^-)]\cdot\mathbb{1}\{\tau^+(u)\neq \tau^-(u) \}\cdot \mathbb{1}\{\tau^+\text{ or }\tau^-\notin \Omega_{\partial S_v}\}\\
    \leq\,&\sum_{u\in \partial S_v}a_u \P\left[X_s^+(u)\neq X_s^-(u)\right] + \sum_{u\in \partial S_v} \Pr\left[X^-_s(\partial S_v) \notin \Omega_{\partial S_v}\text{ or }X^+_s(\partial S_v) \notin \Omega_{\partial S_v}\right]\\
    \leq\,& \max_{u \in V} \Pr[X_s^+(u)\neq X_s^-(u)] \sum_{u\in \partial S_v} a_u + \abs{\partial S_v} \cdot \frac{1}{n^3},
\end{align*}
where we used $\mathbb{1}\{\tau^+, \tau^-\in \Omega_{\partial S_v}\}\le 1$ in the first inequality, and the condition~in \eqref{eq:burn-in} in the second.
Finally, the sum of $a_u$ can be bounded by the ASSM property in \Cref{def:correlation-decay}.
As $\abs{\partial S_v}\le n$, the second term in~\eqref{eq:sum-of-three} can be bounded by
\begin{align}\label{eq:term-2}
    \sum_{(\sigma^+,\sigma^-):\sigma^+\neq \sigma^-} \Pr[\mathcal{C}(\sigma^+,\sigma^-)] \left | \mu_v^{\tau^+}(1)-\mu_v^{\tau^-}(1)\right | \leq \frac{\max_{u \in V} \Pr[X_s^+(u)\neq X_s^-(u)]}{20} + \frac{1}{n^2}.
\end{align}

Combining~\eqref{eq:upper-bound-on-X-t-plus-v-neq-X-t-minus-v}, ~\eqref{eq:sum-of-three},~\eqref{eq:term-1-and-3}, and~\eqref{eq:term-2}, we have
for all $v \in V$,
\begin{align*}
\P[X_t^+(v)\neq X_t^-(v)] \leq \frac{3\max_{u\in V} \P[X_s^+(u)\neq X_s^-(u)]}{20} + \frac{1}{n^2}.
\end{align*}
Taking the maximum over $v\in V$ proves~\eqref{eq:recursion}.
\end{proof}

\begin{remark}[Relaxing the local mixing condition]
  In \Cref{def:correlation-decay}, the local mixing condition is assumed for an arbitrary outside configuration $\sigma \in \{0,1\}^{V \setminus S_v}$. For applications considered in this paper, we can verify this strong assumption of local mixing. However, the proof technique above works fine with a relaxed condition of local mixing, where we consider only $\sigma \in \{0,1\}^{V \setminus S_v}$ such that $\sigma(\partial S_v) \in \Omega_{\partial S_v}$ instead of an arbitrary outside configuration. The mixing result in \Cref{thm:mixing} still holds under this relaxed local mixing condition.
\end{remark}

\section{Construct the good neighbourhood}\label{sec:construct-region}
In this section, we show how to construct the good neighbourhood.
Let $G=(V,E)$ be a graph.
For any $v\in V$ we construct the good neighbourhood $S_v\subseteq V$ such that $v\in S_v$.
We first need some definitions.
%
Recall \Cref{def:saw-tree}, the SAW tree.
Let $\text{cld}_T(u)$ be the set of children of $u$ in a tree $T$. For a SAW tree $T = T_{\textnormal{SAW}}(G,v,\partial S_v)$ rooted at $v$, and for any vertex $u\in T$ that is a copy of some vertex in $S_v$, define 
\begin{align}\label{eq:F_T}
F_T(u):=|\{w\in \text{cld}_T(u): w \text{ \small is a copy of some vertex in } S_v \text{ and } w \text{ \small is not a cycle-closing vertex in } T\}|.
\end{align}

\begin{lemma}\label{lemma:saw-tree-construction-nonsimple}
Let $G=(V,E)$ be a graph.  
Let $1 \leq D_1 \leq D_2$ be two integer parameters.
For any vertex $v\in V$, there exists $S_v\subseteq V$ with $v\in S_v$ such that $|S_v| \leq \exp(D_1)\cdot D_2$ and the following property holds for the SAW tree $T = T_{\textnormal{SAW}}(G,v,\partial S_v)$.  
For any leaf vertex $w$ in $T$ such that $w$ is a copy of some vertex in $\partial S_v$, at least one of the following two conditions holds:
\begin{itemize}
\item Let $v = u_1,u_2,\cdots,u_k,w$ be the path from the root $v$ to $w$ in $T$, where $k \geq 1$ is the distance between $v$ and $w$ in $T$. It holds that $\sum_{i=1}^{k-1} F_T(u_i) \geq D_1$;
\item there exists an ancestor $u$ of $w$ such that the number of non-cycle-closing children of $u$ is at least $D_2$. 
\end{itemize}
\end{lemma}

\begin{proof}[Proof of \Cref{lemma:saw-tree-construction-nonsimple}]
Fix $v \in V$ and we construct the region $S_v$ as follows. Consider the SAW tree $T_\emptyset=T_{\textnormal{SAW}}(G,v,\emptyset)$. By removing all cycle-closing vertices in $T_\emptyset$, we obtain a tree $T'$.
We use a DFS starting from the root $v$ to first construct a region $Q_v$ as in \Cref{alg:Q_v-construction}. (\Cref{alg:Q_v-construction} is the same as the procedure for trees described in \Cref{sec:T-ASSM-overview}.) 
In the algorithm, for each vertex $u\in T'$, 
\begin{align*}
  \text{degsum}(u)\defeq\sum_{w\in \text{path}(v,u)}|\text{cld}_{T'}(w)|,
\end{align*}
where $\text{path}(v,u)$ is the set of vertices on the path from $v$ to $u$ in $T'$, including $v$ and $u$.
\begin{algorithm}
    \caption{Construction of the region $Q_v$}\label{alg:Q_v-construction}
    \SetKwProg{Proc}{Procedure}{}{}
    Initialize $Q_v =\emptyset$\;
    $\textnormal{DFS}(v)$\;
    \Return $Q_v$\;
    \Proc{$\textnormal{DFS}(u)$}{
        $Q_v \gets Q_v\cup \{u\}$\;
        \If{$u$ is a leaf in $T'$}{
            \Return\;
        }
        \ElseIf{$\text{degsum}(u) \geq D_1$\label{step-last}}{
            \If{$ |\text{cld}_{T'}(u)| < D_2$}{
                $Q_v \gets Q_v\cup \text{cld}_{T'}(u)$\label{step-add}\;
            }
            \Return\;
        }
        \Else{
            \For{each child $w$ of $u$\label{step-DFS}}{
                $\textnormal{DFS}(w)$\;
            }
        }
    }
\end{algorithm}  

After constructing $Q_v$ by \Cref{alg:Q_v-construction}, define
\begin{align*}
  S_v\defeq\{u\in G:\exists u'\in Q_v \text{ such that } u' \text{ is a copy of } u\}.
\end{align*}

Let $T = T_{\textnormal{SAW}}(G,v,\partial S_v)$ be the SAW tree rooted at $v$ with boundary $\partial S_v$.
We first show that for each $w\in T$ that is a copy of some vertex in $\partial S_v$, at least one of the two conditions in the lemma holds.

Let $v = u_1,u_2,\cdots,u_k, w$ be the path from the root $v$ to $w$ in $T$, where $k$ is the distance between $v$ and $w$. If there exists an ancestor $u$ of $w$ such that the number of non-cycle-closing children of $u$ in $T$ is at least $D_2$, then the second condition holds. Otherwise, for any ancestor $u_i$ of $w$, it must hold that the number of non-cycle-closing children of any $u_i$ is less than $D_2$ in $T$. Recall the tree $T'$ obtained from $T_\emptyset = T_{\textnormal{SAW}}(G,v,\emptyset)$ by removing all cycle-closing vertices. 
Since none of the $\{u_i\}_{i\in[k]}$ is cycle-closing, the path $v= u_1,u_2,\cdots,u_k$ must be present in $T'$ as well.
Since $u_i$ is not a leaf vertex in $T$, it has the same set of children in $T$ as in $T_\emptyset$. Hence, the non-cycle-closing children of $u_i$ in $T$ are exactly the children of $u_i$ in $T'$. Therefore, $|\text{cld}_{T'}(u_i)| < D_2$ for all $1 \leq i \leq k$.
Consider the DFS procedure in $T'$. When we do the DFS along the path $u_1,u_2,\cdots,u_k$ in $T'$, the $\text{DFS}$ procedure must stop at some $u_j$ for $1 \leq j \leq k-1$ because:
\begin{itemize}
    \item the DFS procedure must have stopped at some $u_j$ for $1 \leq j \leq k$. Otherwise, $w$ is added to $Q_v$ and then $w$ cannot be a copy of some vertex in $\partial S_v$;
    \item furthermore, the DFS procedure cannot stop at $u_k$. Otherwise, since $u_k$ is not a leaf vertex in $T'$, $u_k$ must satisfy the condition in Line \ref{step-last}. Note that $\text{cld}_{T'}(u_k) < D_2$. Then, all children of $u_k$, including $w$, are added to the set $Q_v$. This contradicts the assumption that $w$ is a copy of some vertex in $\partial S_v$.
\end{itemize} 
Note that the $\text{DFS}$ procedure can stop only when it reaches the condition in Line~\ref{step-last}, because $u_1,u_2,\cdots,u_{k-1}$ are not leaves in $T'$. 
Furthermore, since $|\text{cld}_{T'}(u_i)|< D_2$ for all $1 \leq i \leq k-1$, the $\text{DFS}$ procedure can stop only after executing Line~\ref{step-add} at some $u_j$ for $1\leq j \leq k-1$. Therefore, 
$$\sum_{i=1}^{k-1} F_T(u_i) \geq \sum_{i=1}^{j} F_T(u_i) = \text{degsum}(u_j) \geq D_1,$$ 
where the equality follows from the fact that all children of $u_i$ in $T'$ are added to $Q_v$ for $1\leq i\leq j$ (for $i < j$, we run DFS on all children of $u_i$, and for $i = j$, since $\abs{\text{cld}_{T'}(u_j)} < D_2$, we add all children of $u_j$ to $Q_v$ directly). Hence, all children of $u_i$ in $T'$ are copies of some vertices in $S_v$, and none of them is cycle-closing by the definition of $T'$. Therefore, for each $1 \le i \le j$, we have $F_T(u_i)=|\text{cld}_{T'}(u_i)|$, which gives the equality above. This implies that the first condition holds. 

Finally, we bound the size of $S_v$. Since there is a surjection from $Q_v$ to $S_v$, we have $|S_v| \leq |Q_v|$. 
Consider the following optimisation problem.
Let $g(m)$ be the maximum number of vertices in a tree $T_0$ such that: 
for any leaf $u$ in $T_0$, 
\begin{align*}
  \sum_{w\in \text{path}(v,u),w\neq u}|\text{cld}_{T_0}(w)| < m.
\end{align*} 
In other words, $g(m)$ denotes the size of the largest tree $T_0$ satisfying the condition above with parameter $m$. By definition, $g(1) = 1$.
We claim that the following recursive relation holds:
\begin{align*}
    g(m) =\max_{d\in [1,m-1]} \{1 + d\cdot g(m-d)\}.
\end{align*}
Indeed, for any tree $T_0$ satisfying the requirement with parameter $m$, let $d$ be the number of children of the root $v$ in $T_0$. Then each subtree rooted at a child of $v$ satisfies the same condition with parameter $m-d$, and hence each such subtree contains at most $g(m-d)$ vertices.

We prove $g(m)\leq \exp(m)$ by induction. The base case $g(1) = 1 \leq \exp(1)$ holds. Assume $g(m') \leq \exp(m')$ for all $m' < m$. Then
\begin{align*}
 \forall d \in [1,m-1], \quad  1+d \cdot g(m-d) &\leq 1 + d\cdot \exp(m-d)\\
&\leq 1+(\exp(d)-1)\exp(m-d)\\
&= \exp(m) + 1 - \exp(m-d) \leq \exp(m),
\end{align*}
where we use $\exp(d) \geq 1 + d$ for all $d\in \mathbb{R}$ in the second inequality.

Back to the size of $|Q_v|$.
If we omit the children added to $Q_v$ in Line~\ref{step-add}, then the remaining DFS tree has at most $g(D_1) \leq \exp(D_1)$ vertices by the optimisation problem analysed above. 
Each time Line~\ref{step-add} is executed, at most $D_2$ children are added to $Q_v$, and these added vertices do not trigger further DFS calls. Since Line~\ref{step-add} can be executed at most once for each vertex in the remaining DFS tree, we obtain
\[
|Q_v| \leq g(D_1)\cdot D_2 \leq \exp(D_1)\cdot D_2.
\]
Therefore, $|S_v| \leq |Q_v| \leq \exp(D_1)\cdot D_2$.
\end{proof}



\section{Reducing the ASSM property from graphs to SAW trees}
In this section, we verify the conditions in \Cref{def:correlation-decay} for the neighbourhood constructed in \Cref{lemma:saw-tree-construction-nonsimple}. 
Fix a vertex $v \in V$ in the graph $G=(V,E)$.
Let $S_v \subseteq V$ be the region constructed by \Cref{lemma:saw-tree-construction-nonsimple} with the following parameters
\begin{align}\label{eq:parameters-for-S_v-construction}
    D_1 \defeq C_D \cdot \log \log n, \quad D_2 \defeq (\log n)^3,
\end{align}
where $C_D$ is a constant. The value of $C_D$ will be determined in~\eqref{eq:C_D-definition}.
Recall $\partial S_v$, the outer boundary of $S_v$. For any $u \in S_v$, define the boundary-neighbors of $u$ as 
\begin{align}\label{eq:boundary-neighbors}
    \nbd{\partial S_v}{G}(u) \defeq \{w \in \partial S_v: (u,w) \in E\}.
\end{align}

\begin{definition}[Good boundary condition]\label{def:good-boundary-configuration}
  We say a configuration $\sigma \in \{0,1\}^{\partial S_v}$ is \emph{good} if for any $u \in S_v$ with $|\nbd{\partial S_v}{G}(u)|>D_2/3$, it satisfies
\begin{align*}
    |\{w\in \nbd{\partial S_v}{G}(u): \sigma(w) = 1\}|\geq |\nbd{\partial S_v}{G}(u)|/(\log n) + 2.
\end{align*}  
Let $\Omega_{\partial S_v}$ denote the set of all good boundary conditions.
\end{definition}

Good boundary conditions admit typical-case ASSM.



\begin{lemma}\label{lemma:assm-with-good-boundary}
    Let $\beta \leq 1 < \gamma$, $\beta\gamma > 1$ and $\lambda < \lambda_0(\beta,\gamma) \defeq \sqrt{\gamma/\beta}$ be three constants. Let $\mu$ be the Gibbs distribution of a $(\beta,\gamma,\lambda)$-ferromagnetic two-spin system with parameters $(\beta_e,\gamma_e)_{e \in E},(\lambda_v)_{v \in V}$ on $G = (V,E)$ as in \Cref{def:ferromagnetic-two-spin-system}.
For any $u \in \partial S_v$, let $a_u$ be the influence of $u$ on $v$ in the distribution $\mu$, defined as in \eqref{eq:influence}, where the boundary condition set $\Omega_{\partial S_v}$ is given by \Cref{def:good-boundary-configuration}. 
Then, 
$\sum_{u\in \partial S_v}a_u \leq \frac{1}{20}$.

Let $\gamma > 1$ and $\lambda < \lambda_0(1,\gamma) \defeq \sqrt{\gamma}$ be two constants. Let $\mu$ be the Gibbs distribution of a $(\gamma,\lambda)$-RBM with parameters $(\gamma_e)_{e \in E},(\lambda_v)_{v \in V}$ on $G = (V,E)$ as in \Cref{def:gamma-lambda-rbm}. The same result $\sum_{u\in \partial S_v}a_u \leq \frac{1}{20}$ holds.
\end{lemma}





In this section, we carry out the first step in the proof of \Cref{lemma:assm-with-good-boundary}. We reduce the problem to verifying a similar ASSM statement on the SAW tree $T$ instead of on the original graph $G$.
Next, in \Cref{sec:assm-in-saw-tree}, we prove the ASSM property on $T$. 
\Cref{lemma:assm-with-good-boundary} follows from combining the two steps. 

Let $\sigma \in \{0,1\}^{\partial S_v}$ be a good boundary condition.
Let $T = T_{\textnormal{SAW}}(G,v,\sigma) = (V_T,E_T)$ be the SAW tree with boundary $\partial S_v$ defined in \Cref{def:saw-tree-with-pinning}.
We first recall some notation and background on the SAW tree $T$.
Let $\Gamma \subseteq V_T$ be the set of cycle-closing leaf vertices of $T$, and let $\rho_\Gamma$ be the pinning on $\Gamma$. Let $\Lambda$ be the set of all leaf vertices in $T$ that are copies of vertices in $\partial S_v$. Let $\sigma_{\Lambda}$ be the pinning on $\Lambda$ inherited from $\sigma$. We use $\bar{\sigma} \defeq \rho_\Gamma \cup \sigma_{\Lambda}$ to denote the total pinning on $\Gamma \cup \Lambda$. Note that all vertices in $\Gamma \cup \Lambda$ are leaves of $T$. Let $\pi$ be the Gibbs distribution on $T$ obtained by inheriting the parameters of $\mu$ on $G$. By \Cref{prop:marginal-distributions-are-identical}, the marginals $\mu_v^\sigma$ and $\pi_v^{\bar{\sigma}}$ are identical.

We next prune the SAW tree $T$ by removing all cycle-closing leaf vertices. Using the self-reducibility property in \Cref{obs:self-reducibility}, we can remove all cycle-closing leaf vertices from $T$ and modify the external fields at their neighbors accordingly. 
From now on, we use $T = (V_T,E_T)$ to denote the pruned SAW tree and $\pi$ to denote the Gibbs distribution on this pruned tree.
Note that for the two-spin system case, $\pi$ is still a Gibbs distribution of a $(\beta,\gamma,\lambda)$-ferromagnetic two-spin system on $T = (V_T,E_T)$. For the RBM case, $\pi$ is still a Gibbs distribution of a $(\gamma,\lambda)$-RBM on $T = (V_T,E_T)$.

As in~\eqref{eq:assm-with-good-boundary}, we want to prove $\sum_{u\in \partial S_v}a_u\leq \frac{1}{20}$, where $a_u = \max_{\sigma \in \Omega_{\partial S_v}} \DTV{\mu^{\sigma^{u \gets 0}}_v}{\mu^{\sigma^{u \gets 1}}_v}$.
Our goal is to reduce this to verifying a similar ASSM statement on $T$.
For this purpose, we extend the definitions of boundary-neighbors and good boundary conditions from the graph $G$ to the SAW tree $T$.
For every vertex $w \in V_T \setminus \Lambda$, similar to \eqref{eq:boundary-neighbors}, define 
\begin{align*}
    \nbd{\Lambda}{T}(w) \defeq \{u \in \Lambda: \{w,u\} \in E_T\}.
\end{align*}
Intuitively, one can view $\Lambda$ as the boundary of $T$. Then $\nbd{\Lambda}{T}(w)$ is the set of boundary-neighbors of $w$ in $T$.
We next define a good boundary condition on $T$.
Note that in the pruned SAW tree $T$, the pinning is defined only on  $\Lambda$, because $\Gamma$ has been removed.
We introduce the following notion of a good boundary condition for the SAW tree $T$, analogous to \Cref{def:good-boundary-configuration}.  

\begin{definition}[Good boundary for the SAW tree]\label{def:good-boundary-configuration-for-saw-tree}
We say a configuration $\tau \in \{0,1\}^{\Lambda}$ is a good boundary condition if for any $w \notin \Lambda$ with $|\nbd{\Lambda}{T}(w)|>D_2/3$, it satisfies
\begin{align}\label{eq:good-boundary-configuration-for-saw-tree}
    |\{u\in \nbd{\Lambda}{T}(w): \tau(u) = 1\}|\geq |\nbd{\Lambda}{T}(w)|/(\log n) + 1.
\end{align}
We use $\Omega_{\Lambda}$ to denote the set of all good boundary conditions on $T$.
\end{definition}



Finally, for any vertex $w \in \Lambda$, define the influence of $w$ on $v$ in the distribution $\pi$ by
\begin{align}\label{eq:b-w}
    b_w = \max_{\tau \in \Omega_{\Lambda}} \DTV{\pi^{\tau^{w \gets 0}}_v}{\pi^{\tau^{w \gets 1}}_v}.
\end{align}

We show the following relationship between the influence bounds in $G$ and $T$.

\begin{lemma}\label{lemma:mapping-influence-bounds}
The influence bounds in $G$ and $T$ satisfy
\begin{align*}
\sum_{u \in \partial S_v} a_u \leq \sum_{w \in \Lambda} b_w.
\end{align*}
\end{lemma}

\begin{proof}
Since $\sum_{w \in \Lambda} b_w=\sum_{u \in \partial S_v} \sum_{w\in \text{copy}(u)}b_w$, it suffices to show that for any $u \in \partial S_v$,
\begin{align}\label{eq:need-to-show-mapping}
    a_u \leq \sum_{w\in \text{copy}(u)}b_w. 
\end{align}
For a pinning $\sigma \in \Omega_{\partial S_v}$, the corresponding pinning on $T$ is $\sigma_\Lambda$, and
\begin{align*}
   \DTV{\mu^{\sigma^{u \gets 0}}_v}{\mu^{\sigma^{u \gets 1}}_v} = \DTV{\pi^{\sigma_{\Lambda}^{\text{copy}(u) \gets 0}}_v}{\pi^{\sigma_{\Lambda}^{\text{copy}(u) \gets 1}}_v},
\end{align*}
where $\sigma_{\Lambda}^{\text{copy}(u) \gets c}$ is the pinning on $T$ obtained from $\sigma_\Lambda$ by changing the value of all copies of $u$ to $c$.

List all copies of $u$ in $T$ as $\text{copy}(u) = \{u_1,\cdots,u_k\}$. By the triangle inequality, we can write 
\begin{align}\label{eq:triangle-split}
  \DTV{\mu^{\sigma^{u \gets 0}}_v}{\mu^{\sigma^{u \gets 1}}_v} \leq \sum_{i=1}^k \DTV{\pi^{\sigma_{\Lambda,i-1}}_v}{\pi^{\sigma_{\Lambda,i}}_v},
\end{align}
where, for any $i\ge 1$, $\sigma_{\Lambda,i}$ is obtained from $\sigma_\Lambda$ by changing the values of $u_1,\cdots,u_{i}$ to $0$ and the values of $u_{i+1},\cdots,u_k$ to $1$.
Note that $\sigma_{\Lambda,i-1}$ and $\sigma_{\Lambda,i}$ differ only at the single vertex $u_i$.

Next, we show that $\sigma_{\Lambda,i} \in \Omega_{\Lambda}$ for all $i = 1,\cdots,k$. 
Consider any vertex $w\notin \Lambda$. The vertex $w$ is a copy of some vertex $w'\in S_v$. 
By the construction of $T$ in \Cref{def:saw-tree-with-pinning}, each vertex $x$ in $\nbd{\Lambda}{T}(w)$ corresponds bijectively to a vertex $y$ in $\nbd{\partial S_v}{G}(w')$, and $x$ is a copy of $y$.
Thus, for any $w\notin \Lambda$ with $|\nbd{\Lambda}{T}(w)|>D_2/3$, we can find $w'\in S_v$ such that $w$ is a copy of $w'$ and $|\nbd{\partial S_v}{G}(w')|=|\nbd{\Lambda}{T}(w)|>D_2/3$. Moreover,
\begin{align}\label{eq:bar-sigma}
    |\{x\in \nbd{\Lambda}{T}(w): \sigma_{\Lambda}(x) = 1\}| &= |\{y\in \nbd{\partial S_v}{G}(w'): \sigma(y) = 1\}| \notag\\
    &\geq |\nbd{\partial S_v}{G}(w')|/(\log n) + 2 \notag\\
    &= |\nbd{\Lambda}{T}(w)|/(\log n) + 2,
\end{align}
where the inequality holds because $\sigma \in \Omega_{\partial S_v}$ in \Cref{def:good-boundary-configuration}.
For $\sigma_{\Lambda,i}$, the only difference from $\sigma_\Lambda$ is that the values of some copies of $u$ are changed. 
In the SAW tree, no two copies of $u$ can be children of the same vertex, so
$|\{x\in \nbd{\Lambda}{T}(w): \sigma_{\Lambda,i}(x) = 1\}|\geq |\{x\in \nbd{\Lambda}{T}(w): \sigma_{\Lambda}(x) = 1\}|-1$.
Hence, for any $\sigma \in \Omega_{\partial S_v}$, combining~\eqref{eq:good-boundary-configuration-for-saw-tree} and~\eqref{eq:bar-sigma}, we obtain $\sigma_{\Lambda,i} \in \Omega_{\Lambda}$. Since the definition of $b_w$ in~\eqref{eq:b-w} ranges over all pinnings in $\Omega_{\Lambda}$, we have
\begin{align*}
    a_u=\max_{\sigma \in \Omega_{\partial S_v}} \DTV{\mu^{\sigma^{u \gets 0}}_v}{\mu^{\sigma^{u \gets 1}}_v} &\leq \sum_{i=1}^k \DTV{\pi^{\sigma_{\Lambda,i}}_v}{\pi^{\sigma_{\Lambda,i-1}}_v} \leq \sum_{w\in \text{copy}(u)}b_w.
\end{align*}
Summing over all $u \in \partial S_v$ proves the lemma.
\end{proof}

\section{ASSM on the SAW tree}\label{sec:assm-in-saw-tree}

We now prove the ASSM property on the SAW tree. Fix a vertex $v \in V$ and a region $S_v$.
Throughout this section, we treat the $(\beta,\gamma,\lambda)$-ferromagnetic two-spin system case and the $(\gamma,\lambda)$-RBM case in parallel. Constants hidden in $O(\cdot)$ notation depend only on $\beta,\gamma,\lambda$ in the two-spin system case and only on $\gamma,\lambda$ in the RBM case. In estimates involving the parameter $\beta$, we take $\beta=1$ for RBMs.
Given a good boundary condition $\sigma \in \Omega_{\partial S_v}$, we construct the SAW tree $T = T_{\textnormal{SAW}}(G,v,\sigma)$ and prune all cycle-closing vertices in $T$. Recall that $\Lambda$ consists of all copies of vertices in $\partial S_v$.
To prove \Cref{lemma:assm-with-good-boundary}, by~\Cref{lemma:mapping-influence-bounds}, we need to show that
\begin{align*}
    \sum_{w\in \Lambda} b_w = \sum_{w \in \Lambda} \max_{\tau \in \Omega_{\Lambda}} \DTV{\pi^{\tau^{w \gets 0}}_v}{\pi^{\tau^{w \gets 1}}_v}  \leq \frac{1}{20},
\end{align*}
where $\pi$ is the Gibbs distribution on the SAW tree.

For each vertex $w \in \Lambda$, let $\tau^w$ be the pinning of $\Lambda$ in $\Omega_{\Lambda}$ that maximizes the total variation distance $\DTV{\pi^{\tau^{w \gets 0}}_v}{\pi^{\tau^{w \gets 1}}_v}$. We write a superscript $w$ to emphasize that the pinning $\tau^w$ depends on $w$. 
In the analysis, we view the SAW tree $T$ as a computation tree and use the tree recursion to compute the marginal ratio at the root $v$. For each vertex $w$, define the corresponding ratio pinning $\rho^w: \Lambda \setminus \{w\} \to [0,\infty]$ such that 
\begin{align}\label{eq:rho-w}
\forall u \in \Lambda \setminus \{w\}, \quad \rho^w(u)=\begin{cases}
\infty & \text{if } \tau^w(u) = 0;\\
0 & \text{if } \tau^w(u) = 1.
\end{cases}
\end{align}
Consider two ratios $R_v^{\rho^w \land w \gets \infty}$ and $R_v^{\rho^w \land w \gets 0}$ at $v$ under the two pinnings $\rho^w \land w \gets \infty$ and $\rho^w \land w \gets 0$, respectively, where the ratio is computed via the tree recursion (see \Cref{def:pinning-on-computation-tree}). Using the same proof as in \Cref{lemma:dtv<R}, it is straightforward to show that
\begin{align*}
b_w &= \left | \frac{1}{1+R_v^{\rho^w \land w \gets \infty}} - \frac{1}{1+R_v^{\rho^w \land w \gets 0}} \right | = \frac{\abs{R_v^{\rho^w \land w \gets \infty}-R_v^{\rho^w \land w \gets 0}}}{(1+R_v^{\rho^w \land w \gets \infty})(1+R_v^{\rho^w \land w \gets 0})} \\
&\leq \abs{R_v^{\rho^w \land w \gets \infty}-R_v^{\rho^w \land w \gets 0}}.
\end{align*}

Hence, it suffices to bound the difference $\abs{R_v^{\rho^w \land w \gets \infty}-R_v^{\rho^w \land w \gets 0}}$.
However, for different vertices $w$, the pinnings $\rho^w: \Lambda \setminus \{w\} \to [0,\infty]$ can be different. We show that we can modify each pinning $\rho^w$ to a pinning $\sigma^w$ such that $\sigma^w$ is similar to $\sigma^{w'}$ whenever two vertices $w$ and $w'$ lie on the same level of the SAW tree $T$. Recall that $L_k(v)$ is the set of all descendants of $v$ at distance $k$ from $v$ in the SAW tree $T$. 

\begin{definition}\label{def:universal-pinning}
A pinning $\sigma^*: \Lambda \to \{0,\infty\}$ is defined as follows. For each non-leaf vertex $u$,
\begin{itemize}
  \item if $|N_\Lambda^T(u)| \leq D_2 / 3$, we set $\sigma^*(w) = \infty$ for all $w \in \Lambda$ that are children of $u$;
  \item if $|N_\Lambda^T(u)| > D_2 / 3$, let $w_1,w_2,\ldots,w_d \in \Lambda$ be the children of $u$ in the SAW tree $T$, where $d = |N_\Lambda^T(u)|$. Let $\gamma_i = \gamma_{u,w_i}$ and $\beta_i = \beta_{u,w_i}$. Suppose all $w_i$ are sorted in increasing order of $\beta_i\gamma_i$ (breaking ties arbitrarily). For the first $\lfloor |N^T_\Lambda(u)|/(\log n)  \rfloor$ children, we set $\sigma^*(w_i) = 0$. For the remaining children, we set $\sigma^*(w_i) = \infty$.
\end{itemize}
\end{definition}

Intuitively, the pinning $\sigma^*$ is a pinning in $\Omega_\Lambda$ that maximizes the ratio $R^{\sigma^*}_v$. To see this, since we consider a ferromagnetic two-spin system, setting $\sigma^*(w) = \infty$ for all $w$ would maximize the ratio $R^{\sigma^*}_v$. However, by \Cref{def:good-boundary-configuration-for-saw-tree}, if $|\nbd{\Lambda}{T}(u)| > D_2 / 3$, then there is a restriction on the pinning at children of $u$. Hence, in the above definition, we need to pay special attention when $|\nbd{\Lambda}{T}(u)| > D_2 / 3$.

The following lemma plays a key role in the analysis.
For any $k$, let $L_{<k}(v) = \cup_{0 \leq j < k}L_j(v)$.
\begin{lemma}\label{lemma:monotone-potential}
Let $\beta \leq 1 < \gamma$, $\beta\gamma > 1$ and $\lambda < \lambda_0(\beta,\gamma) \defeq \sqrt{\gamma/\beta}$.
Consider a ferromagnetic two-spin system $\+S$ on a tree $T = (V_T,E_T)$ rooted at $v$ with Gibbs distribution $\pi$ and parameters $(\beta_e,\gamma_e)_{e \in E_T}$ and $(\lambda_v)_{v \in V_T}$.
Suppose $\beta_e \gamma_e > 1$, $\beta_e \leq \beta \leq 1 \leq \gamma \leq \gamma_e$, and $\lambda_v < \lambda $ for all $e \in E_T$ and $v \in V_T$.
Let $w\in\Lambda$ be a vertex. Suppose $w \in L_k(v)$ for some $k \in \mathbb{N}$. 
For any pinning $\rho^w$ obtained from $\tau_w\in\Omega_{\Lambda}$ as in \eqref{eq:rho-w}, there exists a pinning $\sigma^w: (L_k(v) \setminus \{w\}) \cup (\Lambda \cap L_{<k}(v)) \to [0,\infty]$ such that
\begin{itemize}
\item for all vertices $u \in \Lambda$ with $u \in L_{k'}(v)$ for $k' < k$, $\sigma^w(u) = \sigma^*(u)$;
\item for all siblings $u$ of the vertex $w$, $\sigma^w(u) = \rho^w(u)$ if $u \in \Lambda$ and $\sigma^w(u) \in (0,\lambda)$ if $u \notin \Lambda$.
\end{itemize}
Then the following inequality holds:
\begin{align}\label{eq:lemma:monotone-inequality}
 \abs{R_v^{\rho^w \land w \gets \infty}-R_v^{\rho^w \land w \gets 0}} \leq \abs{R_v^{\sigma^w \land w \gets \infty}-R_v^{\sigma^w \land w \gets 0}}.
\end{align}
\end{lemma}
We remark that \Cref{lemma:monotone-potential} does not require $\beta_e \gamma_e < \beta\gamma$ for all edges $e \in E_T$. Hence, the lemma works for both $(\beta,\gamma,\lambda)$-ferromagnetic two-spin systems and $(\gamma,\lambda)$-RBMs (with $\beta = 1$).
The proof of \Cref{lemma:monotone-potential} is given in \Cref{subsec:proof-of-lemma:monotone-potential}. Using this lemma, we can bound the sum of the influences at each level. For each integer $k \geq 1$, the sum of influences can be bounded as 
\begin{align*}
\sum_{w \in L_k(v) \cap \Lambda} \abs{R_v^{\rho^w \land w \gets \infty}-R_v^{\rho^w \land w \gets 0}} \leq \sum_{w \in L_k(v) \cap \Lambda} \abs{R_v^{\sigma^w \land w \gets \infty}-R_v^{\sigma^w \land w \gets 0}} \defeq \text{Inf}(k).
\end{align*}
By the definition of the pinning $\sigma^w$, for any $w \in L_k(v) \cap \Lambda$, the restriction of $\sigma^w$ to $L_{<k}(v)$ is the same. The only difference lies in the pinning on vertices at level $k$. Hence, we reduce the task of proving aggregate strong spatial mixing to the problem analyzed in \Cref{subsec:general-results-for-correlation-decay}.
We also remark that \Cref{lemma:monotone-potential} is the only place where we use the stronger condition $\lambda < \lambda_0$.

Define the following spin system for each level $k \geq 1$.
\begin{definition}\label{def:ferromagnetic-two-spin-system-for-level-k}
Let $k \geq 1$ be an integer. Define a spin system $\+S_k$ as follows. 
\begin{itemize}
    \item Truncate the SAW tree $T$ to keep the first $k$ levels. Let $T_k$ be the truncated SAW tree. For each vertex $u$ in $T_k$, let $d_u^c$ be the number of children of $u$. Note that we have removed cycle-closing vertices, but the pinned vertices remain.  The vertices and edges in $T_k$ inherit the parameters of the original spin system on $T$.
    \item For each vertex $w \in L_{<k}(v) \cap \Lambda$, the value of $w$ is fixed by the pinning $\sigma^*$. Using self-reducibility in \Cref{obs:self-reducibility}, we remove the leaf vertex $w$ and modify the external field of its parent.
\end{itemize}
\end{definition}
By \Cref{obs:self-reducibility}, if the original system is a $(\beta,\gamma,\lambda)$-ferromagnetic two-spin system, then $\+S_k$ is also a $(\beta,\gamma,\lambda)$-ferromagnetic two-spin system. The same holds for $(\gamma,\lambda)$-RBM, 

For each vertex $w \in L_k(v)$, we use $\sigma^w_k$ to denote the pinning $\sigma^w$ restricted on $L_k(v)$. Hence, $\sigma^w_k$ is a pinning on all leaf vertices of the tree $T_k$ except the vertex $w$. Let $R_{v,k}^{\sigma^w_k \land w \gets \infty}$ and $R_{v,k}^{\sigma^w_k \land w \gets 0}$ be the ratio computed via the tree recursion in the spin system $\+S_k$. By definition, we have
\begin{align}\label{eq:influence-bound-level-k}
    \text{Inf}(k) = \sum_{w \in L_k(v) \cap \Lambda} \abs{R_{v,k}^{\sigma^w_k \land w \gets \infty}-R_{v,k}^{\sigma^w_k \land w \gets 0}}.
\end{align}

Recall that $D_1$ and $D_2$ are defined in \eqref{eq:parameters-for-S_v-construction}. Let $n$ denote the number of vertices in the original graph $G$.
We have the following two lemmas. 

\begin{lemma}\label{lemma:influence-bound-level-1}
If $k > (\log \log n)^3$, then $\text{Inf}(k)  \leq C' \cdot (1-\delta)^{k} \cdot (\log n)^3$, where $\delta < 1$ and $C'  > 0$ are constants.
\end{lemma}

\begin{lemma}\label{lemma:influence-bound-level-2}
If $1 \leq k \leq (\log \log n)^3$, then $\text{Inf}(k) < \frac{1}{\log n}$.
\end{lemma}

Assuming \Cref{lemma:influence-bound-level-1} and \Cref{lemma:influence-bound-level-2} hold, we can bound the sum of the influence as follows.
\begin{align*}
\sum_{k \geq 1} \text{Inf}(k) &\leq  \frac{(\log \log n)^3}{\log n} + \sum_{k > (\log \log n)^3} C' \cdot (1-\delta)^{k} \cdot (\log n)^3\\
&\leq o(1) + \frac{C' (1-\delta)^{(\log \log n)^3}}{\delta} (\log n)^3 = o(1) < \frac{1}{20}.
\end{align*}
The last equality holds because when $n$ is sufficiently large, we have $(1-\delta)^{(\log \log n)^3} \ll \frac{1}{(\log n)^3}$. 
The above analysis shows that $\sum_{w\in \Lambda} b_w \leq \frac{1}{20}$.
Combining it with \Cref{lemma:mapping-influence-bounds} proves \Cref{lemma:assm-with-good-boundary}.

\subsection{Analysis of the sum of the influence}\label{subsec:analysis-of-the-sum-of-the-influence}
We prove \Cref{lemma:influence-bound-level-1} and \Cref{lemma:influence-bound-level-2} in this subsection.
We consider the following setting. Fix an integer $k \geq 1$. Let $\+S_k$ be the ferromagnetic two-spin system defined in \Cref{def:ferromagnetic-two-spin-system-for-level-k} in the tree $T_k$, where $T_k$ is a tree with $k$ levels rooted at $v$. 
Recall that $T_k$ is constructed by the following procedure. 
First, let $T_{\partial S_v} = T_{\textnormal{SAW}}(G,v,\partial S_v)$.
After pruning all cycle-closing vertices in $T_{\partial S_v}$, we obtain a tree $T$.
Finally, we truncate the tree $T$ and keep levels $0,1,\ldots,k$, and then prune all vertices in $\Lambda \cap L_{<k}(v)$. When pruning a vertex, we modify the external field of its parent using self-reduction.
\begin{lemma}\label{lemma:property-of-vertex-in-tree-T_k}
Let $u \in L_{k'}(v)$ be a vertex at level $k'$ of the tree $T_k$, where $k'\leq k-2$. 
\begin{itemize}
    \item The number of children $|\text{cld}_{T_k}(u)|$ of $u$ in $T_k$ satisfies $|\text{cld}_{T_k}(u)| = F_{T_{\partial S_v}}(u)$, where $F_{T_{\partial S_v}}$ is defined in \eqref{eq:F_T} and $T_{\partial S_v} = T_{\textnormal{SAW}}(G,v,\partial S_v)$.
    \item If the number of non-cycle-closing children of $u$ in $T_{\partial S_v}$ is at least $D_2$, then either $u$ has at least $D_2/2$ children in $T_k$ or $\lambda_u \leq \lambda (1/\gamma)^{D_2/(5 \log n)}$, where $\lambda_u$ is the external field of $u$ in $\+S_k$.
\end{itemize}
\end{lemma}
\begin{proof}
By the construction of $T_k$, for vertex $u$, we have pruned all its cycle-closing children and children in $\Lambda$ from the tree $T_{\partial S_v}$. The first property holds from the definition of $F_{T_{\partial S_v}}(u)$.

For the second property, if the number of non-cycle-closing children of $u$ in $T_{\partial S_v}$ is at least $D_2$, then one of the following two conditions must hold:
\begin{itemize}
    \item $u$ has at least $D_2/2$ children in $T_{\partial S_v}$ that are copies of vertices in $S_v$. All of them remain in $T_k$. Hence, $u$ has at least $D_2/2$ children in $T_k$.
    \item $u$ has at least $D_2/2$ children in $T_{\partial S_v}$ that are copies of vertices in $\partial S_v$. Hence, $u$ has at least $D_2/2$ children in $T$ that belong to $\Lambda$. By the definition of $\sigma^*$ in \Cref{def:universal-pinning}, at least $\lfloor |N^T_\Lambda(u)|/(\log n)  \rfloor$ children in $\nbd{\Lambda}{T}(u)$ satisfy $\sigma^*(w_i) = 0$. Note that when we prune a vertex and modify the external field of its parent using self-reduction, we can only decrease the external field of the parent because $\beta_e \leq 1$ and $\gamma_e \geq \gamma > 1$ for all edges $e$. Hence, the external field of $u$ in $T_k$ can be bounded by
    \begin{align*}
    \lambda_u \leq \lambda \cdot \tp{\frac{\beta 0 + 1}{0 + \gamma}}^{\lfloor |N^T_\Lambda(u)|/\log n \rfloor} &\leq \lambda \cdot \tp{\frac{1}{\gamma}}^{\lfloor |N^T_\Lambda(u)|/ \log n \rfloor}\\
     &\leq \lambda \cdot \tp{\frac{1}{\gamma}}^{\lfloor D_2/(2 \log n) \rfloor} \leq \lambda \cdot \tp{\frac{1}{\gamma}}^{D_2/(5 \log n)}.
    \end{align*}
\end{itemize}
Hence, the second property holds.
\end{proof}

For any vertex $w \in L_k(v) \cap \Lambda$, there is an associated pinning $\sigma^w_k$ on $L_k(v) \setminus \{w\}$. By \Cref{lemma:monotone-potential}, the pinning $\sigma^w_k$ satisfies the following condition.
\begin{lemma}\label{lemma:property-of-pinning-sigma^w_k}
Let $w \in L_k(v) \cap \Lambda$ be a vertex at level $k$ of the tree $T_k$. Let $u$ be the parent of $w$ in $T_k$, where $u$ is at level $k-1$.
The following two properties hold for the pinning $\sigma^w_k$.
\begin{itemize}
    \item For any sibling $w' \notin \Lambda$ of $w$, $\sigma^w_k(w') \in [0,\lambda)$.
    \item If $u$ has more than $D_2/3$ children in $\Lambda$ (i.e., $|\nbd{\Lambda}{T_k}(u)| > D_2 / 3$), then at least $\lfloor |N^{T_k}_\Lambda(u)|/\log n \rfloor$ siblings $w'$ of $w$ satisfy $\sigma^w_k(w') = 0$.
\end{itemize}
\end{lemma}
\begin{proof}
The first property follows directly from \Cref{lemma:monotone-potential}. For the second property, if $|\nbd{\Lambda}{T_k}(u)| > D_2 / 3$, then in the pinning $\rho^w$ from \Cref{lemma:monotone-potential}, at least $\lfloor |N^{T_k}_\Lambda(u)|/ \log n \rfloor+1$ children $w'$ of $u$ satisfy $\rho^w(w') = 0$. This is because $\rho^w$ is obtained from a good pinning $\tau^w \in \Omega_\Lambda$; see \eqref{eq:rho-w}. 
Note that in $T_k$ and $T_{\partial S_v}$, the children of $u$ in $\Lambda$ are the same.
By the definition of a good boundary pinning in \Cref{def:good-boundary-configuration-for-saw-tree}, at least $\lfloor |N^{T_k}_\Lambda(u)|/\log n \rfloor+1$ children of $u$ satisfy $\tau^w(w') = 1$, and thus $\rho^w(w') = 0$. Using \Cref{lemma:monotone-potential}, all siblings $w' \in \Lambda$ of $w$ satisfy $\sigma^w_k(w') = \rho^w(w')$. Hence, at least $\lfloor |N^{T_k}_\Lambda(u)|/\log n \rfloor$ siblings $w' \in \Lambda$ of $w$ satisfy $\sigma^w_k(w') = 0$.
\end{proof}

Recall that the influence we need to bound is 
\begin{align*}
    \text{Inf}(k) = \sum_{w \in L_k(v) \cap \Lambda} \abs{R_{v,k}^{\sigma^w_k \land w \gets \infty}-R_{v,k}^{\sigma^w_k \land w \gets 0}},
\end{align*}
where $R_{v,k}^{\cdot}$ is the ratio computed by tree recursion in $T_k$ rooted at $v$. 
Moreover, the spin system $\+S_k$ on the tree $T_k$ still satisfies \Cref{condition}. Indeed, $T_k$ is obtained from the pruned SAW tree by taking an induced subtree and absorbing some leaf pinnings into external fields via self-reducibility. Thus the edge parameters are inherited, and the external fields can only decrease. Hence the parameter assumptions in \Cref{condition} are preserved; the potential-function bounds and the SSM property hold with the same constants.
Similarly to~\eqref{eq:K-v-w-potential} and~\eqref{eq:K-u-w-potential}, we can define the potential-based influence $K_{v,k}^w$ and $K_{u,k}^w$, where we add a subscript $k$ to emphasise that the quantity is defined on the tree $T_k$.
We will use the general results \Cref{lemma:general-influence-decay-1,lemma:general-influence-decay-2} to bound the influence. However, those two lemmas bound the sum over all $w \in L_k(v)$ instead of the sum over $w \in L_k(v) \cap \Lambda$. We will use the following modified version of the two lemmas.

\begin{lemma}\label{lemma:general-influence-decay-1-modi}
    Let $u \in L_\ell(v)$ be a vertex at level $\ell$, where $0 \le \ell \leq k-2$. Let $u_1,u_2,\ldots,u_d$ be the children of $u$. Then
\begin{align*}
\sum_{w \in L_{k-\ell}(u) \cap \Lambda} K_{u,k}^w 
 \leq C_\text{trl} \lambda_u d \exp(-C_{\text{decay}}d)\max_{1 \leq i \leq d} \sum_{w \in L_{k-\ell-1}(u_i) \cap \Lambda} K_{u_i,k}^w,
\end{align*}
where $L_j(u)$ denotes the set of vertices at level $j$ in the subtree rooted at $u$.
\end{lemma}

\begin{lemma}\label{lemma:general-influence-decay-2-modi}
There exist constants $\ell_0$ and $0< \delta < 1$ such that if $k > \ell_0$, then for any $0 \leq \ell \leq k - \ell_0$, for any vertex $u \in L_\ell(v)$ with children $u_1,\cdots,u_d$, it holds that 
\begin{align*}
  \sum_{w \in L_{k-\ell}(u) \cap \Lambda} K_{u,k}^{w} \leq (1 - \delta) \max_{1 \leq i \leq d} \sum_{w \in L_{k-\ell-1}(u_i) \cap \Lambda} K_{u_i,k}^w.
\end{align*}
\end{lemma}

\Cref{lemma:general-influence-decay-1-modi} and \Cref{lemma:general-influence-decay-2-modi} follows from almost the same proof as \Cref{lemma:general-influence-decay-1} and \Cref{lemma:general-influence-decay-2}, respectively. The only difference is that we sum over all $w \in L_{k-\ell}(u) \cap \Lambda$ instead of the sum over $w \in L_{k-\ell}(u)$ in~\eqref{eq:sum-K-u} and~\eqref{eq:sum-K-u-bound}, respectively.



\subsubsection{Proof of \Cref{lemma:influence-bound-level-1}}

Suppose $k \geq (\log \log n)^3$. We use \Cref{lemma:general-influence-decay-2-modi} $(k - \ell_0 + 1)$ times and then use \Cref{lemma:general-influence-decay-1-modi} $(\ell_0 - 2)$ times,
where $\ell_0$ is from \Cref{lemma:general-influence-decay-2-modi}. We arrive at a vertex $u$ at level $k-1$ with children $u_1,u_2,\ldots,u_d$ in $\Lambda$. Let the path be $v=v_0,v_1,\ldots,v_{k-1}=u$, and let $d_j$ be the number of children of $v_j$ in $T_k$. Then
\begin{align*}
    \sum_{w \in L_k(v) \cap \Lambda } K_{v,k}^w \leq (1-\delta)^{k-\ell_0+1} \cdot \prod_{j=k-\ell_0+1}^{k-2} \tp{C_{\text{trl}} \cdot \lambda_{v_j} d_j \exp(-C_{\text{decay}}d_j)} \cdot \sum_{i=1}^d K_{u,k}^{u_i}.
\end{align*} 
Note that $\lambda_{v_j}\leq \lambda$ and $d_j \exp(-C_{\text{decay}}d_j)=O(1)$ uniformly over all $j$, and $\delta = \delta(\beta,\gamma,\lambda)$ is a constant. Hence, we have
\begin{align}\label{eq:bound-on-influence-level-k-pf}
    \sum_{w \in L_k(v) \cap \Lambda} K_{v,k}^w \leq O_{\beta,\gamma,\lambda}(1) \cdot (1-\delta)^{k} \cdot \sum_{i=1}^d K_{u,k}^{u_i}.
\end{align}

We need the following lemma to bound the influence coming from the last level.
\begin{lemma}\label{lemma:bound-on-last-level-influence}
Let $u$ be a vertex at level $k-1$ with $\bar d$ children such that $u_1,u_2,\ldots,u_{d}$ are in $\Lambda$, where $\bar{d} \geq d$ because some children may not in $\Lambda$. Then
\begin{align*}
    \sum_{i=1}^{d} K_{u,k}^{u_i} \leq \begin{cases}
    \lambda C_{\max} \cdot (\log n)^3 & \text{if } \bar d < D_2 = (\log n)^3;\\
    \exp(-\bar d/(C_0 \log n)) \leq \exp(-d/(C_0 \log n)) & \text{if } \bar d \geq D_2 = (\log n)^3.
    \end{cases}
\end{align*}
where $C_{\max}$ is the constant in \Cref{lemma:dtv-bound-potential}, 
and $C_0 > 1$ is a sufficiently large constant depending on $\beta,\gamma,\lambda$.
\end{lemma}

Assuming \Cref{lemma:bound-on-last-level-influence} holds, the last-level influence is at most $O_{\beta,\gamma,\lambda}(1) \cdot (\log n)^3$. Combining with~\eqref{eq:bound-on-influence-level-k-pf}, 
\begin{align*}
  \sum_{w \in L_k(v) \cap \Lambda} K_{v,k}^w \leq O_{\beta,\gamma,\lambda}(1) \cdot (1-\delta)^{k} \cdot (\log n)^3.
\end{align*}
Combining the above bound with~\Cref{lemma:dtv-bound-potential} proves \Cref{lemma:influence-bound-level-1}.
We now prove \Cref{lemma:bound-on-last-level-influence}.

\begin{proof}[Proof of \Cref{lemma:bound-on-last-level-influence}]
Consider the two possible cases of the parameter $\bar d$. If $\bar d < D_2 = (\log n)^3$, then by the definition of the tree recursion, the influence 
\begin{align*}
\abs{R_{u,k}^{\sigma^{u_i}_k \land u_i \gets \infty} - R_{u,k}^{\sigma^{u_i}_k \land u_i \gets 0}} \leq \lambda,
\end{align*}
because the recursion function has the image space in $[0,\lambda)$. Note that 
\begin{align*}
K_{u,k}^{u_i} = \abs{\Phi(R_{u,k}^{\sigma^{u_i}_k \land u_i \gets \infty}) - \Phi(R_{u,k}^{\sigma^{u_i}_k \land u_i \gets 0})}.
\end{align*}
Using \Cref{lemma:dtv-bound-potential}, we have $K_{u,k}^{u_i} \leq C_{\max} \lambda$. Summing up all $u_i$ (at most $d \leq bar d < D_2$) gives the first bound.

Suppose $\bar d \geq D_2 = (\log n)^3$. 
Then either $u$ has at least $\bar d / 2 \geq D_2/2$ children in $\Lambda$ or at least $\bar d / 2 \geq D_2/2$ children not in $\Lambda$. Suppose we are in the first case.
By \Cref{lemma:property-of-pinning-sigma^w_k}, at least $\lfloor |N^{T_k}_\Lambda(u)|/\log n \rfloor \geq \bar d/(5 \log n)$ siblings $w$ of $u_i$ satisfy $\sigma^{u_i}_k(w) = 0$.
Note that $\frac{\beta x + 1}{x + \gamma} \leq 1$ for all $x \geq 0$. Hence,
\begin{align*}
    \abs{R_{u,k}^{\sigma^{u_i}_k \land u_i \gets \infty} - R_{u,k}^{\sigma^{u_i}_k \land u_i \gets 0}} \leq \lambda \tp{\frac{1}{\gamma}}^{\bar d/(5 \log n)} \cdot \beta \leq \exp \tp{-\frac{\bar d}{C_1 \log n}}, 
\end{align*}
for some constant $C_1 > 1$ large enough. Here, the first factor $\lambda$ comes from the external field $\lambda_u \leq \lambda$ of $u$; the second factor $(\frac{1}{\gamma})^{\bar d/(5 \log n)}$ comes from the siblings $w$ of $u_i$ with $\sigma^{u_i}_k(w) = 0$; and the last factor $\beta$ bounds $\left|\beta_{u,u_i} - \frac{1}{\gamma_{u,u_i}}\right|$, using $\beta_{u,u_i}\leq \beta$. 
For the second case, $u$ has at least $\bar d / 2 \geq D_2/2$ children not in $\Lambda$. By \Cref{lemma:property-of-pinning-sigma^w_k}, the pinning values at these children are at most $\lambda$. Then
\begin{align*}
    \abs{R_{u,k}^{\sigma^{u_i}_k \land u_i \gets \infty} - R_{u,k}^{\sigma^{u_i}_k \land u_i \gets 0}} \leq \lambda \tp{\frac{\beta \lambda + 1}{\lambda + \gamma}}^{\bar d/3} \cdot \beta \leq \exp \tp{-\frac{\bar d}{C_1}}, 
\end{align*}
where we use $\bar d/3$ to exclude the possibly distinguished child $u_i$, and the last inequality holds for some constant $C_1 > 1$ large enough. 
Finally, summing over all $u_i$ and using the bound in \Cref{lemma:dtv-bound-potential} gives 
\begin{align*}
    \sum_{i=1}^{d} K_{u,k}^{u_i} \leq C_{\max} \bar d \exp \tp{-\frac{\bar d}{C_1 \log n}} \leq \exp\tp{-\frac{\bar d}{C_0 \log n}},
\end{align*}
for some constant $C_0 > 1$ large enough.
\end{proof}

\subsubsection{Proof of \Cref{lemma:influence-bound-level-2}}
Let $\ell_1\defeq\max\{-1,k-\ell_0-1\}$,
where $\ell_0$ is from \Cref{lemma:general-influence-decay-2-modi}. 
By applying \Cref{lemma:general-influence-decay-1-modi} and \Cref{lemma:general-influence-decay-2-modi}, we go through a path from the root $v$ to a vertex $u$ at level $k-1$ with children $u_1,u_2,\ldots,u_d$ in $\Lambda$. Let 
the path be $v = v_0,v_1,\ldots,v_{k-1} = u$, and the number of children of $v_i$ is $d_i$, where $d_{k-1}\geq d$ because some children may not be in $\Lambda$. We have
\begin{align}\label{eq:influence-bound-level-2-pf}
    \begin{split}
    \sum_{w\in L_k(v) \cap \Lambda }K_{v,k}^{w}\leq &\prod_{i=0}^{\ell_1}\min\left \{\tp{C_{\text{trl}} \cdot \lambda_{v_i} d_i \exp(-C_{\text{decay}}d_i)},1-\delta\right \}\\
    \cdot&\prod_{i=\ell_1+1}^{k-2}\tp{C_{\text{trl}} \cdot \lambda_{v_i} d_i \exp(-C_{\text{decay}}d_i)} \cdot \sum_{i=1}^d K_{u,k}^{u_i}.
    \end{split}
\end{align} 
For any $0 \leq i \leq k-2$, we have $d_i=F_{T_{\partial S_v}}(v_i)$, where $F_{T_{\partial S_v}}$ is defined in \eqref{eq:F_T} and $T_{\partial S_v} = T_{\textnormal{SAW}}(G,v,\partial S_v)$. 
Recall the definition of $d_{v_j}^c$ in \Cref{def:ferromagnetic-two-spin-system-for-level-k}. We have $d_i \leq d_{v_i}^c$ for all $i$. 
If there exists $j\in [0,k-2]$ such that $d_{v_j}^c \geq D_2 = (\log n)^3$, then by \Cref{lemma:property-of-vertex-in-tree-T_k}, either $d_j \geq D_2/2$ or $\lambda_{v_j} \leq \lambda (1/\gamma)^{D_2/(5 \log n)}$. If $d_j \geq D_2/2$, then
similar to the proof of \Cref{lemma:bound-on-last-level-influence}, we have
\begin{align*}
    C_{\text{trl}} \cdot \lambda_{v_j} d_j \exp(-C_{\text{decay}}d_j)\leq \exp \left(-\frac{d_j}{C_2}+C_3\right),
\end{align*}
for sufficiently large constants $C_2,C_3>0$. 
Note that here we choose $C_2$ and $C_3$ large so that the estimate above holds for any integer $d_j\ge 1$,
although with sufficiently large $n$ we could absorb $C_3$ into $C_2$.
This is because a similar estimate will be used again later when we do not have the assumption that $d_j\ge D_2$.

Next we have
\begin{align}\label{eq:large-degree-decay-case1}
\begin{split}
    \sum_{w\in L_k(v) \cap \Lambda}K_{v,k}^{w}\leq &\exp\left(-\frac{d_j}{C_2}+C_3\right)\prod_{i=0,i\neq j}^{\ell_1}\min\left \{\tp{C_{\text{trl}} \cdot \lambda_{v_i} d_i \exp(-C_{\text{decay}}d_i)},1-\delta\right \}\\
    \cdot&\prod_{i=\ell_1+1,i\neq j}^{k-2}\tp{C_{\text{trl}} \cdot \lambda_{v_i} d_i \exp(-C_{\text{decay}}d_i)} \cdot \sum_{i=1}^d K_{u,k}^{u_i}\\
    \leq &\exp\left(-\frac{d_j}{C_2}+C_3\right) \cdot O_{\beta,\gamma,\lambda}(1) \cdot (\log n)^3 < \frac{1}{(\log n)^2},
\end{split}
\end{align} 
where the second inequality holds because every factor in the first product is at most $ 1 - \delta < 1$, the second product is no larger than $C^{\ell_0+1}$ for some constant $C>0$ and $\ell_0=O_{\beta,\gamma,\lambda}(1)$, and the term $\sum_{i=1}^d K_{u,k}^{u_i}$ is bounded by \Cref{lemma:bound-on-last-level-influence}.
The last inequality holds for large enough $n$ as $\exp\left(-\frac{d_j}{C_2}+C_3\right)\leq \exp\left(-\frac{(\log n)^3}{2C_2}+C_3\right)\leq \frac{1}{n}$. Since $\sum_{w\in L_k(v) \cap \Lambda}K_{v,k}^{w}$ is at most $\frac{1}{(\log n)^2}$, using \Cref{lemma:dtv-bound-potential}, the sum of the influence without potential function is at most $O(\frac{1}{(\log n)^2}) < \frac{1}{\log n}$.

If $\lambda_{v_j} \leq \lambda (1/\gamma)^{D_2/(5 \log n)}$, then
\begin{align*}
    C_{\text{trl}} \cdot \lambda_{v_j} d_j \exp(-C_{\text{decay}}d_j)&\leq C_{\text{trl}} \cdot \lambda (1/\gamma)^{D_2/(5 \log n)} \cdot d_j \exp(-C_{\text{decay}}d_j)\\
    &\leq \exp\left(-\frac{\log^2 n}{C_2}+C_3\right), 
\end{align*}
where we may choose $C_2$ and $C_3$ to be larger than previous values so that the second inequality holds. Similar to~\eqref{eq:large-degree-decay-case1}, we have
\begin{align*}
    \sum_{w\in L_k(v) \cap \Lambda}K_{v,k}^{w}\leq &\exp\left(-\frac{\log^2 n}{C_2}+C_3\right) \cdot O_{\beta,\gamma,\lambda}(1) \cdot (\log n)^3 < \frac{1}{(\log n)^2}.
\end{align*}

We finish the analysis for the case when there exists $j\in [0,k-2]$ such that $d_{v_j}^c \geq D_2$. 
If $d_{k-1}=d_{v_{k-1}}^c\geq D_2 = (\log n)^3$ (we do not prune pinnings at level $k$), then by \Cref{lemma:bound-on-last-level-influence}, we have $\sum_{i=1}^{d} K_{u,k}^{u_i} \leq \exp(-\frac{d_{k-1}}{C_0 \log n})$. Therefore, 
\begin{align*}
    \sum_{w\in L_k(v) \cap \Lambda}K_{v,k}^{w}\leq O_{\beta,\gamma,\lambda}(1)\sum_{i=1}^{d} K_{u,k}^{u_i} \leq O_{\beta,\gamma,\lambda}(1) \cdot \exp\left(-\frac{d_{k-1}}{C_0 \log n}\right) < \frac{1}{(\log n)^2},
\end{align*}
where the first inequality holds because every factor in the first product in~\eqref{eq:influence-bound-level-2-pf} is at most $ 1 - \delta < 1$, and the second product is no larger than $C^{\ell_0+1}$ for some constant $C>0$ and $\ell_0=O_{\beta,\gamma,\lambda}(1)$. The last inequality holds for sufficiently large $n$ because $\exp\left (-\frac{d_{k-1}}{C_0 \log n}\right )\leq \exp\left(-\frac{(\log n)^3}{C_0 \log n}\right) = \exp\left(-\frac{(\log n)^2}{C_0}\right) \leq \frac{1}{n}$. Again, using \Cref{lemma:dtv-bound-potential}, the sum of the influence without potential function is at most $O(\frac{1}{(\log n)^2}) < \frac{1}{\log n}$.

For the remaining case, we have $d_{v_i}^c < D_2$ for all $i \in [0,k-1]$. Hence, by \Cref{lemma:saw-tree-construction-nonsimple}, we have $\sum_{i=0}^{k-2}d_i\geq D_1$.
Let $C_4$ be a large enough constant such that $(1-\delta)^{C_4/2}\leq \exp(-5)$ and $C_5$ be a large enough constant such that $\exp(-C_5/2C_2)\leq \exp(-5)$. We set 
\begin{align}\label{eq:C_D-definition}
  C_D\defeq2C_2C_3C_4+C_5,
\end{align}
and recall that $D_1=C_D\log\log n$.
There are two subcases: 
\begin{enumerate}
    \item $|\{i\in [0,k-2]:d_i< 2C_2 C_3 \}|\geq C_4\cdot \log \log n$, we have $k\geq C_4\cdot \log \log n$ and
    \begin{align*}
        \sum_{w\in L_k(v) \cap \Lambda}K_{v,k}^{w}\leq &(1-\delta)^{k-\ell_0+1} \cdot O_{\beta,\gamma,\lambda}(1) \cdot (\log n)^3\\
        \leq &(1-\delta)^{C_4\cdot \log \log n/2} \cdot O_{\beta,\gamma,\lambda}(1) \cdot (\log n)^3\\
        \leq &\exp(-5 \cdot \log\log n) \cdot O_{\beta,\gamma,\lambda}(1) \cdot (\log n)^3 < \frac{1}{(\log n)^{1.5}};
    \end{align*}
    \item $|\{i\in [0,k-2]:d_i< 2C_2 C_3\}|< C_4\cdot \log \log n$, then 
    $\sum_{0\leq i\leq k-2:d_i\geq 2C_2C_3}d_i\geq D_1 - 2C_2 C_3C_4\cdot \log \log n=C_5\log \log n$. We have $\exp\left (-\frac{x}{C_2}+C_3\right )\leq \exp\left (-\frac{x}{2C_2}\right )$ for $x\geq 2C_2 C_3$. 
    Hence,
    \begin{align*}
        \sum_{w\in L_k(v) \cap \Lambda}K_{v,k}^{w}\leq &\prod_{i=0}^{\ell_1}\min\left \{\tp{C_{\text{trl}} \cdot\lambda_{v_i}d_i \exp(-C_{\text{decay}}d_i)},1-\delta\right \}\\
        \cdot&\prod_{i=\ell_1+1}^{k-2}\tp{C_{\text{trl}} \cdot \lambda_{v_i} d_i \exp(-C_{\text{decay}}d_i)} \cdot \sum_{i=1}^d K_{u,k}^{u_i}\\
        \leq &\prod_{0\leq i\leq k-2:d_i\geq 2C_2C_3}\exp\left (-\frac{d_i}{2C_2}\right ) \cdot O_{\beta,\gamma,\lambda}(1) \cdot (\log n)^3\\
        =&\exp\left(-\frac{\sum_{0\leq i\leq k-2:d_i\geq 2C_2C_3}d_i}{2C_2}\right)\cdot O_{\beta,\gamma,\lambda}(1) \cdot (\log n)^3\\
        \leq& \exp\left (-\frac{C_5\log \log n}{2C_2}\right )\cdot O_{\beta,\gamma,\lambda}(1) \cdot (\log n)^3 < \frac{1}{(\log n)^{1.5}}.
    \end{align*}
    where the last inequality holds because $\exp(-\frac{C_5\log \log n}{2C_2}) \leq \exp(-5\log \log n)$.
\end{enumerate}
We have shown that $\sum_{w\in L_k(v) \cap \Lambda}K_{v,k}^{w} < \frac{1}{(\log n)^{1.5}}$ for all cases. Using \Cref{lemma:dtv-bound-potential}, the sum of the influence without potential function is at most $O\left(\frac{1}{(\log n)^{1.5}}\right) < \frac{1}{\log n}$.

\subsection{Find the worst pinning}\label{subsec:proof-of-lemma:monotone-potential}

We now give the proof of \Cref{lemma:monotone-potential}.
First, we show the following property of the pinning $\sigma^*$ constructed in \Cref{def:universal-pinning}.
Recall that $L_{<k}(v) = \cup_{j < k} L_j(v)$ and $L_{\geq k}(v) = \cup_{j \geq k} L_j(v)$.

\begin{lemma}\label{lem:sigma-star-property}
Let $\sigma: \Lambda \to \{0,\infty\}$, where $\sigma \in \Omega_\Lambda$\footnote{By definition, $\Omega_\Lambda$ contains all pinnings $\sigma$ such that $\sigma$ fixes the value of each $w \in \Lambda$ to either $0$ or $1$, which is equivalent to fixing the ratio at each $w$ to either $\infty$ or $0$.}.
Let $w \in \Lambda$ and $c \in \{0,\infty\}$.
Let $k \geq 1$ be an integer. 
Define a pinning $\tau: \Lambda \to \{0,\infty\}$ such that
\begin{align*}
\forall u \in \Lambda \setminus \{w\}, \quad \tau(u) = \begin{cases}
\sigma^*(u) & \text{if } u \in \Lambda \cap L_{<k}(v),\\
\sigma(u) & \text{if } u \in \Lambda \cap L_{\geq k}(v).
\end{cases}
\end{align*}
For any non-leaf vertex $u$,
\begin{align*}
R^{\sigma \land w \gets c}_u \leq R^{\tau \land w \gets c}_u,
\end{align*}
where $\sigma \land w \gets c$ is the pinning obtained from $\sigma$ by overwriting the value of $w$ to $c$.
\end{lemma}

\begin{remark}
By \Cref{def:pinning-on-computation-tree}, the ratio
$R^{\sigma \land w \gets c}_u$ is computed via tree recursion given the initial value $\sigma \land w \gets c$ at leaves $\Lambda$. Note that $R^{\sigma \land w \gets c}_u = R^{\bar{\sigma}}_u$, where $\bar{\sigma}$ is the pinning obtained from $\sigma \land w \gets c$ by removing the pinning outside the subtree of $u$.
This is because the value computed at $u$ is independent of the pinning outside the subtree of $u$.
\end{remark}

\begin{proof}[Proof of \Cref{lem:sigma-star-property}]
We prove it by induction on $u$ from bottom to top. For the base case, all children of $u$ are leaf vertices. If $u \in L_{\geq k-1}(v)$, then for $\sigma$ and $\tau$, the pinning on the subtree of $u$ is the same, and hence $R^{\sigma \land w \gets c}_u = R^{\tau \land w \gets c}_u$. Suppose $u \in L_{<k-1}(v)$.
If $|\nbd{\Lambda}{T}(u)| \leq D_2 / 3$, then for all children $x \in \Lambda$ of $u$, we have $\sigma(x) \leq \tau(x) = \sigma^*(x) = \infty$, and $w$ has the same value in the two pinnings (if $w$ is a child of $u$). Since the tree recursion is monotone, we have $R^{\sigma \land w \gets c}_u \leq R^{\tau \land w \gets c}_u$. If $|\nbd{\Lambda}{T}(u)| > D_2 / 3$, let $w_1,w_2,\ldots,w_d \in \Lambda$ be the children of $u$ in the SAW tree $T$, where $d = |N_\Lambda^T(u)|$. Let $\gamma_i = \gamma_{u,w_i}$ and $\beta_i = \beta_{u,w_i}$. Suppose all $w_i$ are sorted in decreasing order of $\beta_i\gamma_i$ (breaking ties arbitrarily). Let $w'_1,w'_2,\ldots,w'_{d'} \notin \Lambda$ be the other children of $u$ in the SAW tree $T$. Let $\gamma'_i = \gamma_{u,w'_i}$ and $\beta'_i = \beta_{u,w'_i}$. Note that $w'_1,\ldots,w'_{d'}$ must be unpinned leaves. Using the tree recursion, we have 
\begin{align}
R^{\sigma \land w \gets c}_u &= \lambda_u \prod_{\substack{1\leq i \leq d:\\ w_i \neq w,\ \sigma(w_i) = 0}} \frac{1}{\gamma_i} \prod_{\substack{1 \leq i \leq d:\\ w_i \neq w,\ \sigma(w_i) = \infty}} \beta_i \prod_{1 \leq j \leq d'} \frac{\beta'_j \lambda_{w'_j} + 1}{\lambda_{w'_j} + \gamma'_j} \cdot W\notag\\
&= \lambda_u \tp{\prod_{1 \leq i \leq d:\ w_i \neq w} \frac{1}{\gamma_i}} \prod_{\substack{1 \leq i \leq d:\\ w_i \neq w,\ \sigma(w_i) = \infty}} \beta_i \gamma_i \prod_{1 \leq j \leq d'} \frac{\beta'_j \lambda_{w'_j} + 1}{\lambda_{w'_j} + \gamma'_j}\cdot W,
\label{eqn:R_u^sigma-wedge-w}
\end{align}
where $W = 1$ if $w$ is not a child of $u$ and $W = \frac{\beta_{u,w} c + 1}{\gamma_{u,w}  + c}$ if $w$ is a child of $u$.
Similarly,
\begin{align}
    R^{\tau \land w \gets c}_u = \lambda_u \tp{\prod_{1 \leq i \leq d:\ w_i \neq w} \frac{1}{\gamma_i}} \prod_{\substack{1 \leq i \leq d:\\ w_i \neq w,\ \sigma^*(w_i) = \infty}} \beta_i \gamma_i \prod_{1 \leq j \leq d'} \frac{\beta'_j \lambda_{w'_j} + 1}{\lambda_{w'_j} + \gamma'_j}\cdot W.
    \label{eqn:R_u^tau-wedge-w}
\end{align}

Let $\nbd{\Lambda}{T}(u)=\{w_1,\ldots,w_d\}$ be the set of children of $u$ that are in $\Lambda$.
If $w$ is a child of $u$, then $w$ appears in this list, but it is excluded from the products in~\eqref{eqn:R_u^sigma-wedge-w} and~\eqref{eqn:R_u^tau-wedge-w} and is represented separately by $W$.
By the definition of $\Omega_\Lambda$, at least $\lfloor |N^T_\Lambda(u)|/(\log n) \rfloor + 1$ children in $\nbd{\Lambda}{T}(u)$ have $\sigma(w_i) = 0$ (that is, the value of $w_i$ is pinned to $1$). Hence, at most $|N_\Lambda^T(u)| - \lfloor |N^T_\Lambda(u)|/(\log n) \rfloor - 1$ children in $\nbd{\Lambda}{T}(u)$ have $\sigma(w_i) = \infty$.
Let us consider two cases.
\begin{itemize}
  \item Case I: $w$ is not a child of $u$. Note that $\beta_i \gamma_i \geq 1$ for all $i$. By definition, $\sigma^*$ picks exactly $|N_\Lambda^T(u)| - \lfloor |N^T_\Lambda(u)|/(\log n) \rfloor$ children $w_i$ with the largest $\beta_i \gamma_i$ and sets $\sigma^*(w_i) = \infty$. By \eqref{eqn:R_u^sigma-wedge-w} and \eqref{eqn:R_u^tau-wedge-w}, $R^{\sigma \land w \gets c}_u \leq R^{\tau \land w \gets c}_u$.
    \item Case II: $w$ is a child of $u$. Note that $W$ is the same factor in both $R^{\sigma \land w \gets c}_u$ and $R^{\tau\land w \gets c}_u$ by \eqref{eqn:R_u^sigma-wedge-w} and \eqref{eqn:R_u^tau-wedge-w}. In $R^{\sigma \land w \gets c}_u$, at most $|N_\Lambda^T(u)| - \lfloor |N^T_\Lambda(u)|/(\log n) \rfloor - 1$ children among $\{w_1,w_2,\ldots,w_d\} \setminus \{w\}$ contribute a factor $\beta_i\gamma_i$ because the pinning on $w$ has been overwritten. In $R^{\tau \land w \gets c}_u$, we may set $\sigma^*(w) = \infty$, but we set $\sigma^*(w') = \infty$ for at least $|\nbd{\Lambda}{T}(u)| - \lfloor |N^T_\Lambda(u)|/(\log n) \rfloor$ children $w' \in \{w_1,w_2,\ldots,w_d\}$. At least $|\nbd{\Lambda}{T}(u)| - \lfloor |N^T_\Lambda(u)|/(\log n) \rfloor-1$ children among $\{w_1,w_2,\ldots,w_d\} \setminus \{w\}$ satisfy $\sigma^*(w_i) = \infty$. These children contribute the $|\nbd{\Lambda}{T}(u)| - \lfloor |N^T_\Lambda(u)|/(\log n) \rfloor-1$ largest factors $\beta_i \gamma_i$ among all children in $\{w_1,w_2,\ldots,w_d\} \setminus \{w\}$. Hence, $R^{\sigma \land w \gets c}_u \leq R^{\tau \land w \gets c}_u$.
\end{itemize}


For a general non-leaf vertex $u$, where $u$ may have non-leaf children $w'$ and children $w_i$ in the set $\Lambda$, the induction hypothesis gives $R^{\sigma \land w \gets c}_{w'} \leq R^{\tau \land w \gets c}_{w'}$ for every non-leaf child $w'$. For all children $w_i \in \Lambda$ of $u$, we can use the same analysis as in the base case. Since the recursion is monotone, it follows that $R^{\sigma \land w \gets c}_u \leq R^{\tau \land w \gets c}_u$.
\end{proof}

We next prove the following technical lemma.
\begin{lemma}\label{lemma:one-step-relation-preserve}
Let $\beta \leq 1 < \gamma$, $\beta\gamma > 1$ and $\lambda < \lambda_0(\beta,\gamma) \defeq \sqrt{\gamma/\beta}$.
Let $\lambda\geq x>y>0$ and $\lambda\geq x'>y'>0$, satisfying $x\geq x'$, $y\geq y'$, and $x/y\geq x'/y'$. Then 
\begin{align}\label{eq:one-step-relation-1}
\frac{\beta x + 1}{x + \gamma} \cdot \frac{y + \gamma}{\beta y + 1} \geq \frac{\beta x' + 1}{x' + \gamma} \cdot \frac{y' + \gamma}{\beta y' + 1}.
\end{align}
\end{lemma}

\begin{proof}
Subtracting 1 from both the left and right sides of~\eqref{eq:one-step-relation-1}, we only need to show that
\begin{align}\label{eq:one-step-relation-2}
\frac{(\beta\gamma-1)(x-y)}{(x + \gamma)(\beta y + 1)} \geq \frac{(\beta\gamma-1)(x'-y')}{(x' + \gamma)(\beta y' + 1)}.
\end{align}
It is easy to see that the right-hand side is monotone decreasing in $y'$. We only need to consider the case $y'=x'y/x$. In this case, we can set $1\ge c=x'/x=y'/y$.
Then~\eqref{eq:one-step-relation-2} is equivalent to
\begin{align*}
\frac{1}{(x + \gamma)(\beta y + 1)} \geq \frac{c}{(cx + \gamma)(\beta cy + 1)},
\end{align*}
which, in turn, is equivalent to $ (1-c)(\gamma-c\beta xy)\geq 0$.
The last inequality holds because $\gamma-c\beta xy \geq \gamma - \beta \lambda^2 > 0$.
\end{proof}

Now, we are ready to prove \Cref{lemma:monotone-potential}.
\begin{proof}[Proof of \Cref{lemma:monotone-potential}]

We first consider the following definition of the pinning $\sigma^w$ on $\Lambda \setminus \{w\}$:
\begin{align}\label{eq:sigma-w-definition}
\forall u \in \Lambda \setminus \{w\}, \quad \sigma^w(u) = \begin{cases}
\sigma^*(u) & \text{if } u \in \Lambda \cap L_{<k}(v),\\
\rho^w(u) & \text{if } u \in \Lambda \cap (L_{\geq k}(v) \setminus \{w\}).
\end{cases}
\end{align}
We first show that~\eqref{eq:lemma:monotone-inequality} holds for this pinning $\sigma^w$.
Note that the pinning $\sigma^w$ in the lemma is a pinning on the subset $(L_k(v) \setminus \{w\}) \cup (\Lambda \cap L_{<k}(v))$. After proving~\eqref{eq:lemma:monotone-inequality}, we explain how to modify $\sigma^w$ so that it satisfies the condition in the lemma.

Let the path from $w$ to $v$ in the SAW tree $T$ be $w = u_0,u_1,\ldots,u_{k-1},u_k = v$. 
By the monotonicity of the recursion function, for all $1\leq j \leq k$, we have 
\begin{align}\label{eq:pf-mono}
&x_j\defeq R_{u_j}^{\sigma^w \land w \gets \infty} > y_j \defeq R_{u_j}^{\sigma^w \land w \gets 0},\quad x'_j \defeq R_{u_j}^{\rho^w \land w \gets \infty} > y'_j \defeq R_{u_j}^{\rho^w \land w \gets 0}.
\end{align}
By definition, $x_j = R_{u_j}^{\sigma^w_j \land w \gets \infty}$, where $\sigma^w_j$ is the pinning $\sigma^w$ projected on vertices in $\cup_{\ell \geq k -j + 1}L_\ell(v)$. This is because when computing the tree recursion for $u_j$, we only need to use all pinnings at the subtree rooted at $u_j$. Note that the vertex $u_j$ is in $L_{k-j}(v)$. Hence, the value of $x_j$ depends only on $\sigma^w_j$. 
Similar results apply to $y_j,x'_j,y'_j$.
By applying \Cref{lem:sigma-star-property} to $u_1,\cdots,u_k$, we have 
\begin{align*}
\forall 1 \leq j \leq k, \quad x_j \geq x'_j \text{ and } y_j \geq y'_j.
\end{align*}
We claim that
\begin{align}\label{eq:pf-ratio-inequality}
\forall 1 \leq j \leq k, \quad \frac{x_j}{y_j}\geq \frac{x'_j}{y'_j}.
\end{align}
We prove inequality~\eqref{eq:pf-ratio-inequality} by induction on $j$.
For $j=1$, note that $x_1,y_1$ depend only on $\sigma^w$ projected on vertices in $L_{\geq k}(v)$ (denoted by $\sigma^w_1$), and $y_1,y_1'$ depend only on $\rho^w$ projected on vertices in $L_{\geq k}(v)$ (denoted by $\rho^w_1$). By \eqref{eq:sigma-w-definition}, $\sigma^w_1 = \rho^w_1$. Hence, $x_1 = x'_1$ and $y_1 = y'_1$, so the claim holds. 
Now fix $1< j\leq k$ and assume the claim holds for $j-1$.
Note that $x_j,y_j,x'_j,y'_j$ can all be computed by tree recursion. Let $\beta_j$ and $\gamma_j$ be the parameters on the edge $\{u_j,u_{j-1}\}$. By comparing the tree recursion for $x_j$ and $y_j$, we have 
\begin{align*}
    \frac{x_j}{y_j} &= \frac{\beta_j x_{j-1} + 1}{x_{j-1} + \gamma_j} \cdot \frac{y_{j-1} + \gamma_j}{\beta_j y_{j-1} + 1}.
\end{align*}
Similarly, we can write
\begin{align*}
    \frac{x'_j}{y'_j} &= \frac{\beta_j x'_{j-1} + 1}{x'_{j-1} + \gamma_j} \cdot \frac{y'_{j-1} + \gamma_j}{\beta_j y'_{j-1} + 1}.
\end{align*}
By the definition of the recursion function, all $x_{j-1},y_{j-1},x'_{j-1},y'_{j-1}\leq \lambda$.
Note $x_{j-1} > y_{j-1}$, $x'_{j-1} > y'_{j-1}$, $x_{j-1} \geq x'_{j-1}$, and $y_{j-1} \geq y'_{j-1}$. By induction hypothesis, $\frac{x_{j-1}}{y_{j-1}} \geq \frac{x'_{j-1}}{y'_{j-1}}$.
Note that the hypothesis of \Cref{lemma:one-step-relation-preserve} is verified
with $(\beta_j,\gamma_j,\lambda)$ in place of $(\beta,\gamma,\lambda)$:
indeed $\beta_j\leq\beta\leq 1<\gamma\leq\gamma_j$, $\beta_j\gamma_j>1$, and
$\lambda<\sqrt{\gamma/\beta}\leq\sqrt{\gamma_j/\beta_j}$.
Using \Cref{lemma:one-step-relation-preserve},
we have
\begin{align*}
\frac{x_j}{y_j} &= \frac{\beta_j x_{j-1} + 1}{x_{j-1} + \gamma_j} \cdot \frac{y_{j-1} + \gamma_j}{\beta_j y_{j-1} + 1}\geq \frac{\beta_j x'_{j-1} + 1}{x'_{j-1} + \gamma_j} \cdot \frac{y'_{j-1} + \gamma_j}{\beta_j y'_{j-1} + 1} = \frac{x'_j}{y'_j}.
\end{align*}

Finally, we have $\frac{R_v^{\sigma^w \land w \gets \infty}}{R_v^{\sigma^w \land w \gets 0}} = \frac{x_k}{y_k} \geq \frac{x'_k}{y'_k} = \frac{R_v^{\rho^w \land w \gets \infty}}{R_v^{\rho^w \land w \gets 0}}$. 
We can compute that $$|x_k-y_k|=y_k\left | \frac{x_k}{y_k}-1 \right |\geq y'_k\left | \frac{x'_k}{y'_k}-1 \right |=|x'_k-y'_k|,$$
where the inequality holds because $y_k \geq y'_k$, $\frac{x_k}{y_k},\frac{x'_k}{y'_k} \geq 1$, and $\frac{x_k}{y_k} \geq \frac{x'_k}{y'_k}$.

To obtain the pinning $\sigma^w$ in the lemma, we compute the tree recursion from the bottom up to the level $k$ conditional on $\sigma^w$, except for vertex $w$ (note that $w \in \Lambda$ is a leaf at level $k$). After the computation, every vertex $u \in L_k(v) \setminus \{w\}$ gets a ratio. We set this value as the pinning value of $\sigma^w(u)$ and remove all the pinnings below the level $k$. 
Therefore, we get a pinning $\sigma^w$ defined on the subset $(L_k(v) \setminus \{w\}) \cup (\Lambda \cap L_{<k}(v))$. By definition, for all $u \in \Lambda \cap L_{<k}(v)$, we have $\sigma^w(u) = \sigma^*(u)$. For all siblings $u \in \Lambda$ of the vertex $w$, note that $u$ is in the level $k$ and $u$ must be a leaf node because $u \in \Lambda$. When computing the tree recursion for $u$, we simply let $u$ take the pinning value $\rho^w(u)$. 
For all siblings $u \not\in \Lambda$ of the vertex $w$, their values are not fixed by $\rho^w$, 
\begin{itemize}
  \item if $u$ is a leaf, then the ratio value at $u$ is $\lambda_u < \lambda$ (note that $u$ cannot be a cycle-closing vertex because we have pruned all cycle-closing vertices when constructing the tree $T$);
    \item if $u$ is not a leaf, then the ratio value at $u$ is computed by tree recursion. The range of the tree recursion function implies that $\sigma^w(u) \in (0,\lambda)$.
\end{itemize}
In both cases, we have $\sigma^w(u) \in (0,\lambda)$.
This verifies the two properties of $\sigma^w$ in the lemma.
\end{proof}

\section{Proof of main results}\label{sec:proof-of-main-results}

In this section we show the main theorems, namely \Cref{thm:alternating-scan-mixing-ferromagnetic-two-spin-system}, \Cref{thm:glauber-mixing-2}, and \Cref{thm:glauber-mixing-1}.
Note that \Cref{thm:alternating-scan-mixing-BM} is implied by \Cref{thm:alternating-scan-mixing-ferromagnetic-two-spin-system}.
We first show the slightly easier \Cref{thm:glauber-mixing-2} in \Cref{sec:glauber-mixing-2}.
Then, in \Cref{sec:alternating-scan-proof}, we show \Cref{thm:alternating-scan-mixing-ferromagnetic-two-spin-system} via a similar approach.
We conclude by proving \Cref{thm:glauber-mixing-1} in \Cref{sec:glauber-mixing-1}.

\subsection{Mixing of Glauber dynamics when \texorpdfstring{$\lambda < \lambda_0$}{lambda < lambda0}}\label{sec:glauber-mixing-2}
\Cref{thm:glauber-mixing-2} is proved by applying \Cref{thm:mixing}.
Recall that Glauber dynamics is a special case of the heat-bath block dynamics in  \Cref{thm:mixing}, where each block is a single vertex.
We verify the conditions in \Cref{def:correlation-decay} and~\eqref{eq:burn-in} in \Cref{thm:mixing} for a $(\beta,\gamma,\lambda)$-ferromagnetic two-spin system on a graph $G$ with $\beta \leq 1 < \gamma$, $\beta\gamma > 1$, and $\lambda < \lambda_0 \defeq \sqrt{\gamma/\beta}$.
The definition of good boundary conditions is given in \Cref{def:good-boundary-configuration}. We first show the following lemma.

\begin{lemma}\label{lemma:close-under-shortest-path}
  For any $v \in V$, any $S_v\ni v$, and any $\sigma,\tau \in \Omega_{\partial S_v}$, where $\Omega_{\partial S_v}$ is defined in \Cref{def:good-boundary-configuration}, there exists a path $\eta_0,\eta_1,\ldots,\eta_t \in \Omega_{\partial S_v}$ such that $\eta_0 = \sigma$, $\eta_t = \tau$, and for any $0 \leq i < t$, $\eta_i$ and $\eta_{i+1}$ differ at exactly one vertex, where
$t = |\{u \in \partial S_v : \sigma(u) \neq \tau(u)\}|$
is the Hamming distance between $\sigma$ and $\tau$.
\end{lemma}
\begin{proof}
To move from $\sigma$ to $\tau$, define the following two sets of vertices:
\begin{align*}
    S_1 &= \{u \in \partial S_v: \sigma(u) = 0, \tau(u) = 1\},\\
    S_2 &= \{u \in \partial S_v: \sigma(u) = 1, \tau(u) = 0\}.
\end{align*}
Starting from $\sigma$, we first change all $v \in S_1$ from the value 0 to the value 1, and then change all $v \in S_2$ from the value 1 to the value 0. For any $\eta_i$ in the path, it is straightforward to see that for any $u \in S_v$ with $|\nbd{\partial S_v}{G}(u)|>D_2/3$, it satisfies
\begin{align*}
    |\{w\in \nbd{\partial S_v}{G}(u): \eta_i(w) = 1\}| &\geq  \min\{|\{w\in \nbd{\partial S_v}{G}(u): \sigma(w) = 1\}| , |\{w\in \nbd{\partial S_v}{G}(u): \tau(w) = 1\}|  \} \\
    &\geq |\nbd{\partial S_v}{G}(u)|/(\log n) + 2.
\end{align*}  
Hence, $\eta_i$ is a good boundary configuration. The length of the path is  $|S_1| + |S_2| = t$.
\end{proof}

\Cref{lemma:close-under-shortest-path} proves the first property of \Cref{def:correlation-decay}. The second property of \Cref{def:correlation-decay} is proved by \Cref{lemma:assm-with-good-boundary}. 
We next verify the condition~\eqref{eq:burn-in} in \Cref{thm:mixing}.
Consider the monotone coupling $(X_t^+,X_t^-)_{t \geq 0}$ of the Glauber dynamics in \Cref{def:monotone-coupling}. 
We show that there exists
\[
T_{\textnormal{burn-in}} = O(n \log n)
\]
such that for any $t \geq T_{\textnormal{burn-in}}$ and any $v \in V$, it holds that
\begin{align}\label{eq:burn-in-pf}
     \Pr[X^+_t(\partial S_v) \notin \Omega_{\partial S_v} \lor X^-_t(\partial S_v) \notin \Omega_{\partial S_v}] \leq \frac{1}{n^3}.
\end{align}
Fix any time $t \geq T_{\textnormal{burn-in}}$. If $T_{\textnormal{burn-in}}$ is a sufficiently large multiple of $n \log n$, then with probability at least $1 - \frac{1}{n^{10}}$, each vertex $u \in V$ has been updated at least once during the time interval $[t-T_{\textnormal{burn-in}},t]$. 
For each vertex $u \in \partial S_v$, consider the last time in the interval $[t-T_{\textnormal{burn-in}},t]$ at which $u$ is updated, and denote this time by $t_u$. 
For every edge $e \in E$, we have $\beta_e \leq 1$ and $\gamma_e \geq 1$. Hence, whenever $u$ is updated, the conditional probability that it is set to $1$ is at least $\frac{1}{1+\lambda_u} \geq \frac{1}{1+\lambda} = \Omega(1)$. Consider a vertex $w \in S_v$ with $d > D_2 / 3 = (\log n)^3 / 3$ neighbors in $\partial S_v$. Since a good boundary configuration requires at least $d/\log n+2$ neighbors of $w$ in state $1$, a Chernoff bound shows that, with probability at least $1 - \frac{1}{n^{10}}$, at least $d/\log n + 2$ neighbors $u$ of $w$ are set to $1$ at their respective times $t_u$. Taking a union bound over the two chains $X^+_t$ and $X^-_t$, and over all relevant vertices $w \in S_v$, yields \eqref{eq:burn-in-pf}.

Finally, we claim the local mixing time for censored Glauber dynamics on $\mu^\sigma_{S_v}$ is 
\begin{align}\label{eq:local-mixing-time-bound-pf}
    T_{\textnormal{local}} = n \cdot (\log n)^{C''},
\end{align}
where $C'' = C''(\beta,\gamma,\lambda) > 0$ is a constant depending on $\beta,\gamma,\lambda$.
Assume the above local mixing time bound holds.
Let $t_{\textnormal{mix}}^{\textnormal{Glauber}}$ denote the mixing time of Glauber dynamics.
By \Cref{thm:mixing}, we have
\begin{align*}
t_{\textnormal{mix}}^{\textnormal{Glauber}}\left(\frac{1}{4e}\right) &= O\left(T_{\textnormal{burn-in}} + T_{\textnormal{local}} \cdot \max_{v \in V}\log |R_v|\cdot \log n\right), \quad \text{where } |R_v| = |S_v \cup \partial S_v| \leq n\\
&\leq n \cdot (\log n)^{C(\beta,\gamma,\lambda)}.
\end{align*}
Then, \Cref{thm:glauber-mixing-2} follows from the standard decay in $\epsilon$ for mixing times, namely $t_{\textnormal{mix}}^{\textnormal{Glauber}}(\epsilon) \leq t_{\textnormal{mix}}^{\textnormal{Glauber}}(\frac{1}{4e})\log \frac{1}{\epsilon}$.

We use the following result to show the local mixing bound in \eqref{eq:local-mixing-time-bound-pf}.

\begin{theorem}\label{thm:local-mixing-bound}
    Let $\beta,\gamma,\lambda > 0$ be three constants such that $\beta \leq 1 < \gamma$, $\beta\gamma > 1$, and $\lambda < \lambda_c \defeq (\gamma/\beta)^{\frac{\sqrt{\beta \gamma}}{\sqrt{\beta \gamma}-1}} $. 
    For any $(\beta,\gamma,\lambda)$-ferromagnetic two-spin system with vertex set $V$, the spectral gap of the Glauber dynamics on the Gibbs distribution $\mu$ is at least $\frac{1}{|V|^{C}}$, where $C = C(\beta,\gamma,\lambda) > 0$ is a constant depending on $\beta,\gamma,\lambda$.
\end{theorem}

\begin{remark}
The above theorem only requires a weaker condition $\lambda < \lambda_c$. Note that 
\begin{align*}
\lambda_c = (\gamma/\beta)^{\frac{\sqrt{\beta \gamma}}{\sqrt{\beta \gamma}-1}} > \sqrt{\gamma/\beta} = \lambda_0.
\end{align*}
Hence, we can use the above theorem to prove the local mixing bound when $\lambda < \lambda_0$. \Cref{thm:local-mixing-bound} can also be viewed as a weaker version of \Cref{thm:glauber-mixing-1} when $\lambda < \lambda_c$ as it only provides a $\mathrm{poly}(n) \cdot \log \frac{1}{\mu(\sigma)}$ mixing time bound instead of the $n^3 \cdot \mathrm{polylog}(n)$ mixing time bound in \Cref{thm:glauber-mixing-1}.
\end{remark}

To prove \Cref{thm:local-mixing-bound}, we need the following mixing result obtained from the spectral independence.

\begin{proposition}[\cite{ALO24}]\label{prop:spectral-independence}
Let $\mu$ be a distribution over $\{0,1\}^V$. If there exists a constant $\eta > 0$ such that for any pinning $\sigma \in \{0,1\}^\Lambda$, the conditional distribution $\mu^\sigma_{V \setminus \Lambda}$ has $\eta$-bounded all-to-one influence, then, the spectral gap of the Glauber dynamics on $\mu$ is at least $\frac{1}{n^{O(\eta)}}$. 
\end{proposition}

\begin{proof}[Proof of \Cref{thm:local-mixing-bound}]
Using \Cref{obs:self-reducibility}, any conditional distribution $\mu^\sigma$ also induces a Gibbs distribution of a $(\beta,\gamma,\lambda)$-ferromagnetic two-spin system on a subgraph. By \Cref{thm:all-to-one-influence}, all conditional distributions have $C_{\text{inf}}$-bounded all-to-one influence for some constant $C_{\text{inf}}  = C_{\text{inf}}(\beta,\gamma,\lambda) > 0$ depending on $\beta,\gamma,\lambda$.
The theorem then follows from \Cref{prop:spectral-independence}.
\end{proof}

We use \Cref{thm:local-mixing-bound} to prove the local mixing bound. Fix any vertex $v \in V$ and any outside configuration $\sigma \in \{0,1\}^{V \setminus S_v}$. The censored Glauber dynamics on $\mu^\sigma_{S_v}$ updates as follows: in each step, it picks a vertex $u \in V$ uniformly at random; if $u \notin S_v$, then the dynamics does nothing; otherwise, it resamples the value at $u$ conditional on the current configuration of the other variables. It is straightforward to see that the censored Glauber dynamics on $\mu^\sigma_{S_v}$ is at most a factor of $n$ slower than the Glauber dynamics on $\pi = \mu^\sigma_{S_v}$, where in each step, the Glauber dynamics picks a vertex $u \in S_v$ uniformly at random and resamples the value.
Using \Cref{lemma:saw-tree-construction-nonsimple} and~\eqref{eq:parameters-for-S_v-construction}, we know that $|S_v| \leq (\log n)^{C'}$, where $C' = C'(\beta,\gamma,\lambda) > 0$ is a constant. We prove the following mixing result. Note that~\eqref{eq:local-mixing-time-bound-pf} is a simple corollary of this lemma.

\begin{lemma}\label{lemma:warm-start-configuration}
Let $\pi = \mu^\sigma_{S_v}$. Let $P_{\pi}^{\textnormal{Glauber}}$ be the Glauber dynamics on $\pi$. Starting from an arbitrary configuration in $\{0,1\}^{S_v}$, after running $P_{\pi}^{\textnormal{Glauber}}$ for $(\log n)^{C''}$ steps, the total variation distance between the resulting distribution and the stationary distribution $\pi$ is at most $\frac{1}{4e}$, where $C'' = C''(\beta,\gamma,\lambda) > 0$ is a constant.
\end{lemma}

By \Cref{obs:self-reducibility}, the conditional distribution $\pi$ is a Gibbs distribution of a $(\beta,\gamma,\lambda)$-ferromagnetic two-spin system on $G[S_v]$.
If we directly apply \Cref{thm:local-mixing-bound} and~\eqref{eq:mixing-time-glauber-gap}, then we need to bound $\log \frac{1}{\pi_{\min}}$, where $\pi_{\min} = \min_{x \in \{0,1\}^{S_v}} \pi(x)$. 
However, for some edge $e$ and vertex $u$, the parameters $\beta_e$ and $\lambda_u$ can be arbitrarily small and the parameter $\gamma_e$ can be arbitrarily large. Hence, $\log \frac{1}{\pi_{\min}}$ can be larger than $\mathrm{polylog}(n)$. To resolve this issue, we use \Cref{thm:local-mixing-bound} after reaching a warm-start configuration.
We give the following general result.
\begin{lemma}\label{lemma:warm-start-configuration-general}
Let $\lambda > 0$ be a constant. 
Let $\mu$ be a Gibbs distribution of a ferromagnetic two-spin system on a graph $G = (V,E)$ with $(\beta_e,\gamma_e)_{e\in E}$ and $(\lambda_v)_{v \in V}$. 
Suppose $\beta_e \gamma_e > 1$ and $\beta_e \leq 1 \leq \gamma_e$ for all $e \in E$, and $\lambda_v < \lambda$ for all $v \in V$.
Let $P_{\mu}^{\textnormal{Glauber}}$ be the Glauber dynamics on $\mu$. Suppose the spectral gap of the Glauber dynamics is at least $0 < g < 1$.  Then, the mixing time of the Glauber dynamics on $\mu$ satisfies
\begin{align*}
t_{\textnormal{mix}}^{\textnormal{Glauber}}\tp{\frac{1}{4e}} \leq O_{\lambda}\left(|V| \log |V| + \frac{|V|^2}{g} \log |V|\right).
\end{align*}
\end{lemma}
We will now use \Cref{lemma:warm-start-configuration-general} to prove the mixing result for $(\beta,\gamma,\lambda)$-ferromagnetic two spin systems.
Note that the lemma does not require any upper bound on $\beta_e \gamma_e$ for $e \in E$. In \Cref{sec:alternating-scan-proof}, we can also use it to prove the mixing of $(\gamma,\lambda)$-RBMs.
Assume that \Cref{lemma:warm-start-configuration-general} holds. We apply \Cref{lemma:warm-start-configuration-general} to the distribution $\pi$ defined on the subgraph $G[S_v]$. Note that $|S_v| \leq (\log n)^{C'}$, where $C' = C'(\beta,\gamma,\lambda) > 0$ is a constant.
Using \Cref{thm:local-mixing-bound} on the subgraph $G[S_v]$, the spectral gap of the Glauber dynamics on $\pi$ is at least $\frac{1}{(\log n)^{C}}$, where $C = C(\beta,\gamma,\lambda)$ is a constant.
Hence, the mixing time of the Glauber dynamics on $\pi$ is at most $(\log n)^{C''}$. This proves \Cref{lemma:warm-start-configuration}.
Finally, we prove \Cref{lemma:warm-start-configuration-general}.

\begin{proof}[Proof of \Cref{lemma:warm-start-configuration-general}]
Let $N = |V|$. 
Let $N_0(\lambda)$ be a sufficiently large constant depending only on $\lambda$. First we consider the case when $N \leq N_0(\lambda) =  O_\lambda(1)$.
In each update of the Glauber dynamics, we have a chance at least $\frac{1}{1+\lambda}$ to update the value of a vertex to 1. 
We run the Glauber dynamics for some $O_\lambda(1)$ steps, so that with probability $\Omega_\lambda(1)$, all vertices take the value 1. 
Let $T_0 = O_\lambda(1)$ be a sufficiently large constant. With probability at least $1 - \frac{1}{10e}$, we can find a time $t < T_0$ such that all vertices take the value 1.
For each edge, $\gamma_e > 1 \geq \beta_e$. It holds that $\mu(\boldsymbol{1}) = \Omega_\lambda(1)$ if $N \leq N_0(\lambda)$. Using~\eqref{eq:mixing-time-glauber-gap}, starting from all-1 configuration, we only need to run Glauber dynamics for $O_\lambda(1/g)$ steps to get a configuration with total variation distance at most $\frac{1}{10e}$ to the stationary distribution $\mu$. 
A simple coupling argument shows that the total variation distance between the resulting distribution and $\mu$ is at most $\frac{1}{4e}$ after $T_0 + O_\lambda(1/g) = O_\lambda(1/g)$ steps.

Now, we assume $N \geq N_0(\lambda)$ is large enough.
Fix $\tau \in \{0,1\}^V$.
We say that a vertex $u \in V$ is \emph{bad} in $\tau$ if
\begin{align*}
    \lambda_u \leq \frac{1}{100 N^5}
    \qquad\text{and}\qquad
    \tau(u) = 0.
\end{align*}
For any edge $e = \{u,w\} \in E$, we say that $e$ is \emph{bad} in $\tau$ if
\begin{align*}
\gamma_e \geq 100N^5 \text{ and } (\tau(u) = 0 \text{ or } \tau(w) = 0),
\end{align*}
and we say that $\tau$ is a \emph{warm-start configuration} if no vertex or edge is bad in $\tau$.

We prove the following two claims.
\begin{itemize}
    \item Starting from an arbitrary configuration $X_0 \in \{0,1\}^{V}$, after running $P_{\mu}^{\textnormal{Glauber}}$ for $T_0 = O_\lambda(N (\log N)^2)$ steps, with probability at least $1 - \frac{1}{10e}$, the configuration $X_{T_0}$ is a warm-start configuration.
    \item Starting from any warm-start configuration $X_{T_0}$, after running the Glauber dynamics for $T_1 = O_\lambda\left(\frac{N^2}{g} \log N\right)$ steps, where $g$ is a lower bound of the spectral gap, the total variation distance between the resulting distribution and $\mu$ is at most $\frac{1}{10e}$.
\end{itemize}
If these two claims hold, we can construct a coupling between the law of $X_{T_0 + T_1}$ and the stationary distribution $\mu$ such that the coupling fails with probability at most
\begin{align*}
&\Pr[X_{T_0} \text{ is not a warm-start configuration}] + \Pr[\text{coupling fails} \mid X_{T_0} \text{ is a warm-start configuration}] \\
&= \frac{1}{10e} + \frac{1}{10e} < \frac{1}{4e},
\end{align*}
which finishes the proof.

Now we prove the first claim. Let $M = C_1 N \log N$ and $L = C_0 \log N$, where $C_1 > 0$ is a sufficiently large absolute constant and $C_0 = C_0(\lambda) > 0$ is a sufficiently large constant depending only on $\lambda$. Set
\begin{align*}
T_0 = LM = O_\lambda(N(\log N)^2).
\end{align*}
Partition the time interval $[T_0]$ into $L$ consecutive blocks, each of length $M$.
We list the sequence of updated vertices as
\begin{align*}
  v_1,v_2,\ldots,v_{T_0}.
\end{align*}
An update sequence is \emph{good} if every vertex is updated at least once in every block.
By the coupon collector bound and a union bound over all $L$ blocks, the update sequence is good with probability at least $1 - \frac{1}{20e}$.

Fix a good update sequence.
We first bound the probability that a vertex is bad in $X_{T_0}$. Fix any vertex $u \in V$, and let $t_u$ be the last time at which $u$ is updated. We must have $\lambda_u \leq \frac{1}{100N^5}$, since otherwise $u$ cannot be bad. Fix all the updates before time $t_u$. Let $u_1,u_2,\ldots,u_d$ denote all neighbors of $u$, let $\beta_i,\gamma_i$ denote the parameters of the edge $\{u_i,u\}$, and let $\rho$ denote the configuration of the other variables at time $t_u$. Then
\begin{align}\label{eq:lambda-u-bound}
 \frac{\Pr[ u \text{ is updated to } 0]}{\Pr[u \text{ is updated to }1]} = \lambda_u \prod_{i \in [d]: \rho(u_i) = 1} \frac{1}{\gamma_i}\prod_{i \in [d]: \rho(u_i) = 0} \beta_i.
\end{align}
Since $\frac{1}{\gamma_i} \leq 1$ and $\beta_i \leq 1$, we have
\begin{align*}
 \Pr[X_{t_u}(u) = 0] \leq \lambda_u \leq \frac{1}{100N^5}.
\end{align*}
Hence, the probability that $u$ is bad in $X_{T_0}$ is at most \smash{$\frac{1}{100N^5}$}.

Now fix a bad edge $e = \{u,w\} \in E$ with  $\gamma_e \geq 100N^5$, as otherwise, $e$ cannot be bad.
We call a pair of times $(t,t')$ a \emph{clean pair} for $e$ if $t < t'$, $\{v_t,v_{t'}\} = \{u,w\}$, $v_t \neq v_{t'}$, and for all $t < \ell < t'$ we have $v_\ell \notin \{u,w\}$.
Since the update sequence is good, both $u$ and $w$ are updated at least once in every block. Fix any block. 
Since both vertices appear in the block, there must be a clean pair of times for $e$. 
We list all clean pairs in the update sequence: $(t_j,t'_j)_{j=1}^K$ with $t'_j < t_{j+1}$,
where $K \geq L \geq C_0 \log N$ since there is at least one clean pair in each block.

Fix all randomness used to update vertices in $V \setminus \{u,w\}$.
Let $p_\lambda \defeq \frac{1}{1+\lambda}$. For each clean pair $(t_b,t_b')$, define the event $A_b$ by
\begin{align*}
A_b = \{X_{t_b'}(u) = X_{t_b'}(w) = 1\}.
\end{align*}
By~\eqref{eq:lambda-u-bound}, at every update of either $u$ or $w$, the chosen vertex is updated to $1$ with probability at least $p_\lambda$. Therefore, conditional on all past updates on $\{u,w\}$ before time $t_b$, the probability of the event $A_b$ is at least $p_\lambda^2$.
By iterating this bound over all clean pairs and choosing $C_0$ sufficiently large as a function of $\lambda$, we obtain
\begin{align*}
\Pr\Big[\bigcap_{b = 1}^K \overline{A_b}\Big] \leq (1-p_\lambda^2)^K \leq N^{-6}.
\end{align*}
If any event $A_b$ occurs, then the following event $A$ holds:
\begin{itemize}
\item $A$: there exists the first time $t_e < T_0$ such that $X_{t_e}(u) = X_{t_e}(w) = 1$.
\end{itemize}
Hence, $t_e$ exists with probability at least $1 - N^{-6}$. Furthermore, the random variable $t_e$ 
is independent from the updates after $t_e$.
Suppose $t_e = s'$ and let $s > s'$ be the first time after $t_e$ at which the edge $u$ or $w$ is updated to $0$.
At time $s$, one of the endpoints, say $u$, is updated to $0$ while the other endpoint is still equal to $1$. Hence, by~\eqref{eq:lambda-u-bound},
\begin{align*}
\Pr[X_s(u) = 0 \mid X_{s-1}(w) = 1] \leq \frac{\lambda_u}{\gamma_e} \leq \frac{\lambda}{100N^5}.
\end{align*}
A union bound over all times $s'$ for $t_e = s'$ and all times $s$ for $s > s'$ yields
\begin{align*}
\Pr[e \text{ is bad in } X_{T_0} \mid A] \leq \frac{\lambda T_0^2}{100N^5}.
\end{align*}
Therefore, since $T_0 = O_\lambda(N(\log N)^2)$, we have
\begin{align*}
\Pr[e \text{ is bad in } X_{T_0}] \leq \Pr[\neg A] + \Pr[e \text{ is bad in } X_{T_0} \mid A] \leq N^{-6} + \frac{\lambda T_0^2}{100N^5} \leq \frac{1}{100N^{2.5}}.
\end{align*}

Taking a union bound over all vertices and edges, conditioned on the update sequence fixed above, the probability that $X_{T_0}$ is not a warm-start configuration is at most
\begin{align}\label{eq:warm-start-configuration-probability}
\frac{N}{100N^5} + \frac{N^2}{100N^{2.5}} < \frac{1}{20e}.
\end{align}
Combining this with the probability $\frac{1}{20e}$ that the update sequence is not good proves the first claim.

For the second claim, we show a lower bound on $\mu(\tau)$ for each warm-start configuration $\tau$. For any configuration $\tau' \in \{0,1\}^{V}$, not necessarily a warm-start configuration, we give a lower bound on the ratio $\frac{\mu(\tau)}{\mu(\tau')}$. We analyze the contribution of every vertex and every edge in $G$. Formally, the ratio $\frac{\mu(\tau)}{\mu(\tau')}$ can be written as the following ratio of products:
\begin{align*}
  \frac{\mu(\tau)}{\mu(\tau')} = \frac{ \prod_{u \in V} a_u(\tau(u)) \prod_{e \in E} b_e(\tau(e)) }{ \prod_{u \in V} a_u(\tau'(u)) \prod_{e \in E} b_e(\tau'(e)) },
\end{align*}
where, for each vertex $u \in V$,
\begin{align*}
  a_u(\tau(u))\defeq
  \begin{cases}
    \lambda_u & \text{ if }\tau(u) = 0;\\
    1 & \text{ if }\tau(u) = 1,
  \end{cases}
\end{align*}
and for each edge $e = \{u,w\} \in E$,
\begin{align*}
  b_e(\tau(e))\defeq
  \begin{cases}
    \beta_e & \text{ if }\tau(u) = \tau(w) = 0;\\
    \gamma_e & \text{ if }\tau(u) = \tau(w) = 1;\\    
    1 & \text{ if }\tau(u) \neq \tau(w).
  \end{cases}
\end{align*}
We analyse each ratio as follows.
\begin{itemize}
    \item If $\lambda_u \leq \frac{1}{100N^5}$, then $\tau(u) = 1$ because $\tau$ is warm-start. Hence, $\frac{a_u(\tau(u))}{a_u(\tau'(u))} \geq \min\{1,\lambda^{-1}\}$.
    \item If $\lambda_u > \frac{1}{100N^5}$, then $\frac{a_u(\tau(u))}{a_u(\tau'(u))} \geq \min\{1/(100N^5),\lambda^{-1}\}$.
    \item If $\gamma_e \geq 100N^5$, then $\tau(u) = \tau(w) = 1$ because $\tau$ is warm-start. Therefore, $\frac{b_e(\tau(e))}{b_e(\tau'(e))} \geq 1$.
    \item If $\gamma_e < 100N^5$, then $\beta_e > \frac{1}{\gamma_e} > \frac{1}{100N^5}$ because $\beta_e \gamma_e > 1$. Therefore,
    \begin{align*}
        \frac{b_e(\tau(e))}{b_e(\tau'(e))} \geq \frac{\beta_e}{\gamma_e} > \frac{1}{10^4 N^{10}}.
    \end{align*}
\end{itemize}
The total number of edges in $E$ is at most $N^2$. Hence, the ratio $\frac{\mu(\tau)}{\mu(\tau')}$ can be bounded as follows:
\begin{align*}
\frac{\mu(\tau)}{\mu(\tau')} \geq \tp{\min\{1/(100N^5),\lambda^{-1}\}}^{N} \cdot \tp{\frac{1}{10^4 N^{10}}}^{N^2} \geq \exp(-O_\lambda(N^2 \log N)).
\end{align*}
Since the above lower bound holds for every $\tau' \in \{0,1\}^V$, summing over all $2^N$ choices of $\tau'$ gives
\begin{align}\label{eq:mu-tau-lower-bound}
 \mu(\tau) \geq \exp(-O_\lambda(N^2 \log N)) \cdot 2^{-N} = \exp(-O_\lambda(N^2 \log N)).
\end{align}
Let
\begin{align*}
 T_1 \defeq O\left( \max_{\text{warm-start } \tau} \frac{1}{g} \log \frac{1}{(1/10e)^2 \mu(\tau)}\right) = O_\lambda\tp{\frac{N^2}{g} \log N}.
\end{align*}
The second claim follows from~\eqref{eq:mixing-time-glauber-gap} with $\epsilon = \frac{1}{10e}$ for the warm-start configuration $\tau$.
\end{proof}

\subsection{Mixing of alternating-scan sampler}\label{sec:alternating-scan-proof}
In this section, we prove the alternating-scan mixing bound for $(\gamma,\lambda)$-RBMs (\Cref{thm:alternating-scan-mixing-ferromagnetic-two-spin-system}), which implies \Cref{thm:alternating-scan-mixing-BM}. Let $\gamma > 1$ and $\lambda < \lambda_0(1,\gamma) \defeq \sqrt{\gamma}$ be two constants, and let $\mu$ be the Gibbs distribution of a $(\gamma,\lambda)$-RBM on a bipartite graph $G=(V_0,V_1,E)$ with $V = V_0 \uplus V_1$. We prove that the alternating-scan sampler on $\mu$ has mixing time $(\log n)^{O_{\gamma,\lambda}(1)} \log \frac{1}{\epsilon}$.

The proof strategy here is the same as that in \Cref{sec:glauber-mixing-2}.
The alternating-scan sampler is a special case of the systematic-scan block dynamics in \Cref{thm:mixing} with two blocks, namely $\+B = \{V_0,V_1\}$. 
The definition of good boundary conditions is given in \Cref{def:good-boundary-configuration}. 
\Cref{lemma:close-under-shortest-path} proves the first property of \Cref{def:correlation-decay}. The second property of \Cref{def:correlation-decay} is proved by \Cref{lemma:assm-with-good-boundary}. 
For the burn-in estimate in~\eqref{eq:burn-in-pf}, we can simply set $T_{\textnormal{burn-in}} \defeq 2$. 
In the alternating-scan sampler, after two steps all vertices have been updated exactly once. The bound in~\eqref{eq:burn-in-pf} follows from the same Chernoff-bound argument used in \Cref{sec:glauber-mixing-2}.

Finally, we claim that the local mixing time for the censored alternating-scan sampler on $\mu^\sigma_{S_v}$ is 
\begin{align}\label{eq:local-mixing-time-bound-pf-1}
    T_{\textnormal{local}} = (\log n)^{C''},
\end{align}
where $C'' = C''(\gamma,\lambda) > 0$ is a constant depending on $\gamma,\lambda$. 
Let $t_{\textnormal{mix}}^{\textnormal{AS}}$ denote the mixing time of the alternating-scan sampler.
Assuming this local mixing bound, \Cref{thm:mixing} implies
\begin{align*}
t_{\textnormal{mix}}^{\textnormal{AS}}\left(\frac{1}{4e}\right) &= O\left(T_{\textnormal{burn-in}} + T_{\textnormal{local}} \cdot \max_{v \in V}\log |R_v|\cdot\log n\right), \quad \text{where } |R_v| = |S_v \cup \partial S_v| \leq n\\
&\leq (\log n)^{C(\gamma,\lambda)}.
\end{align*}
\Cref{thm:alternating-scan-mixing-ferromagnetic-two-spin-system} then follows from the standard $\epsilon$ decay in mixing times $t_{\textnormal{mix}}^{\textnormal{AS}}(\epsilon) \leq t_{\textnormal{mix}}^{\textnormal{AS}}\left(\frac{1}{4e}\right)\log \frac{1}{\epsilon}$, and \Cref{thm:alternating-scan-mixing-BM} follows as a special case.

Fix $v \in V$, let $S = S_v$, and set $N = |S| \leq (\log n)^{C'}$, where $C' = C'(\gamma,\lambda) > 0$. Fix a boundary configuration $\sigma \in \{0,1\}^{\partial S}$. Let
\begin{align*}
    \pi \defeq \mu^\sigma_{S}.
\end{align*}
By \Cref{obs:self-reducibility}, $\pi$ is the Gibbs distribution of a $(\gamma,\lambda)$-RBM on $G[S]$. Write $S_i \defeq S \cap V_i$ for $i \in \{0,1\}$, and let $Q_\pi$ be the alternating-scan sampler on $\pi$ with blocks $S_0$ and $S_1$.

Every conditional distribution of $\pi$ is again a $(\gamma,\lambda)$-RBM. Since $\lambda < \lambda_0(1,\gamma) < \lambda_c(\gamma,1)$, \Cref{thm:all-to-one-influence} implies that these conditional distributions have $C_{\textnormal{inf}}$-bounded all-to-one influence for some $C_{\textnormal{inf}} = C_{\textnormal{inf}}(\gamma,\lambda) > 0$. By \Cref{prop:spectral-independence}, the spectral gap of the Glauber dynamics on $\pi$ is at least $N^{-O(C_{\textnormal{inf}})}$. Then \Cref{prop:relaxation-time-alternating-scan,prop:mixing-time-alternating-scan} imply that, starting from any configuration $\tau \in \{0,1\}^{S}$, after running $Q_\pi$ for
\[
2N^{O(C_{\textnormal{inf}})} \log \frac{4e^2}{\epsilon^2 \pi(\tau)}
\]
steps, the total variation distance from $\pi$ is at most $\epsilon$. We prove the local mixing bound in \eqref{eq:local-mixing-time-bound-pf-1} using a warm-start argument similar to that in \Cref{sec:glauber-mixing-2}. The case $N = O_{\lambda}(1)$ can be handled by the same argument. For large $N$, recall the definition of a warm-start configuration from the proof of \Cref{lemma:warm-start-configuration-general}. Let
\begin{align*}
T_{\textnormal{warm}}^{\textnormal{AS}} = O_\lambda(\log N).
\end{align*}
In every two consecutive steps of $Q_\pi$, every vertex in $S$ is updated exactly once, and every edge receives a clean ordered pair of endpoint updates. Therefore, the same argument as in the proof of \Cref{lemma:warm-start-configuration-general} shows that, starting from any configuration $X_0 \in \{0,1\}^{S}$, after running $Q_\pi$ for $T_{\textnormal{warm}}^{\textnormal{AS}}$ steps, the probability that $X_{T_{\textnormal{warm}}^{\textnormal{AS}}}$ is a warm-start configuration is at least $1 - \frac{1}{10e}$.

For any warm-start configuration $\tau \in \{0,1\}^{S}$, by~\eqref{eq:mu-tau-lower-bound}, we have $\pi(\tau) \geq \exp(-O_\lambda(N^2 \log N)) \geq \exp(-(\log n)^{O_{\gamma,\lambda}(1)})$. Let
\begin{align*}
T_{\textnormal{post}}^{\textnormal{AS}} \defeq 2N^{O(C_{\textnormal{inf}})} \max_{\substack{\tau \in \{0,1\}^{S}:\\ \tau \textnormal{ is warm-start}}} \log \frac{4e^2}{(1/10e)^2 \pi(\tau)} \leq (\log n)^{O_{\gamma,\lambda}(1)}.
\end{align*}
Starting from any warm-start configuration, after $T_{\textnormal{post}}^{\textnormal{AS}}$ additional steps, the resulting distribution is within $\frac{1}{10e}$ in total variation distance from $\pi$.

Hence, starting from any configuration $X_0 \in \{0,1\}^{S}$, we can couple $X_{T_{\textnormal{warm}}^{\textnormal{AS}} + T_{\textnormal{post}}^{\textnormal{AS}}}$ with the stationary distribution $\pi$ successfully with probability at least $1 - 1/(10e) - 1/(10e) > 1 - 1/(4e)$. By the coupling inequality,
\begin{align*}
\DTV{X_{T_{\textnormal{warm}}^{\textnormal{AS}} + T_{\textnormal{post}}^{\textnormal{AS}}}}{\pi} \leq \frac{1}{4e}.
\end{align*}
This proves the local mixing time bound in \eqref{eq:local-mixing-time-bound-pf-1}.

\subsection{Mixing of Glauber dynamics when \texorpdfstring{$\lambda < \lambda_c$}{lambda < lambdac}}\label{sec:glauber-mixing-1}
To prove \Cref{thm:glauber-mixing-1}, we use the field dynamics technique introduced in~\cite{CFYZ21}. Let $\mu$ be a distribution over $\{0,1\}^V$, and let $\boldsymbol{\theta} = (\theta_v)_{v \in V}$ be a vector of real numbers. The tilted distribution $\boldsymbol{\theta} * \mu$ is defined by
\begin{align*}
\forall \sigma \in \{0,1\}^V, \quad (\boldsymbol{\theta} * \mu)(\sigma) \propto \mu(\sigma) \cdot \prod_{v \in V: \sigma_v = 0} \theta_v.
\end{align*}
In particular, if $\theta_v = \theta$ for all $v \in V$, then we denote $\boldsymbol{\theta} * \mu = \theta * \mu$.

The field dynamics on $\mu$ is defined as follows. Let $\theta \in (0,1)$. Starting from an arbitrary configuration $X \in \{0,1\}^V$, in each step, it updates the current configuration $X$ as follows:
\begin{itemize}
    \item construct a random subset $S \subseteq V$ by selecting each vertex $v \in V$ independently with probability $p_v$, where $p_v = 1$ if $X(v) = 1$ and $p_v = \theta$ if $X(v) = 0$;
    \item resample $X(S) \sim (\theta * \mu)_S^{X(V \setminus S)}$, where $(\theta * \mu)_S^{X(V \setminus S)}$ is the marginal distribution on $S$ induced by $(\theta * \mu)$ conditioned on the configuration $X(V \setminus S)$ on the variables outside $S$.
\end{itemize}
Compared with the original version of the field dynamics in~\cite{CFYZ21}, the above definition swaps the roles of 0 and 1. The two versions are essentially equivalent.
The spectral gap of the field dynamics can be analyzed using the \emph{complete spectral independence} property. We have the following proposition.

\begin{proposition}[\text{\cite{CFYZ21}}]\label{prop:field-dynamics}
Let $\eta > 0$ be a constant. 
If the distribution $\mu$ over $\{0,1\}^V$ satisfies the following condition: for any $\boldsymbol{\phi} \in (0,1]^V$ and any pinning $\sigma \in \{0,1\}^\Lambda$, the conditional distribution $(\boldsymbol{\phi} * \mu)^\sigma_{V \setminus \Lambda}$ has $\eta$-bounded all-to-one influence, then for any $\theta \in (0,1)$, the spectral gap of the field dynamics on $\mu$ with parameter $\theta$ is at least $\theta^{O(\eta)}$. 
\end{proposition}


Let $\mu$ be a Gibbs distribution of a $(\beta,\gamma,\lambda)$-ferromagnetic two-spin system on a graph $G$, where $\lambda < \lambda_c \defeq (\gamma/\beta)^{\frac{\sqrt{\beta \gamma}}{\sqrt{\beta \gamma}-1}}$.
By \Cref{def:ferromagnetic-two-spin-system}, the tilted distribution $\boldsymbol{\theta} * \mu$ is again a ferromagnetic two-spin system with the same edge parameters and external fields satisfying $\lambda_v \theta_v < \lambda < \lambda_c$. By \Cref{obs:self-reducibility} and \Cref{thm:all-to-one-influence}, the distribution $\mu$ satisfies the condition in \Cref{prop:field-dynamics} with $\eta = C_{\text{inf}}$. Let $\gamma_{\text{field}}^{}(\mu,\theta)$ denote the spectral gap of the field dynamics on $\mu$ with parameter $\theta$. Then
\begin{align}\label{eq:gamma-field-theta}
    \gamma_{\text{field}}^{}(\mu,\theta) \geq \theta^{O(C_{\text{inf}})}.
\end{align}

To relate the field dynamics to the Glauber dynamics, we need the following definition. Let $\sigma \in \{0,1\}^\Lambda$ be a configuration, where $\Lambda \subseteq V$ is a subset of vertices. Consider the distribution $(\theta * \mu)^\sigma$, obtained by pinning all variables in $\Lambda$ according to $\sigma$. The Glauber dynamics on $(\theta * \mu)^\sigma$ is defined as follows. Starting from an arbitrary configuration $X \in \{0,1\}^{V}$ with $X(\Lambda) = \sigma$, in each step, pick a vertex $v \in V$ uniformly at random. If $v \in \Lambda$, then do nothing; if $v \notin \Lambda$, then resample
$X(v) \sim (\theta * \mu)^{X(V \setminus \{v\})}_v$.
In particular, we take the parameter $\theta$ as follows:
\begin{align}\label{eq:theta-definition}
\theta = \frac{1}{2\lambda_c} = \Theta_{\beta,\gamma,\lambda}(1).
\end{align}
Note that $(\theta * \mu)^\sigma$ coincides with the conditional distribution $(\theta * \mu)^\sigma_{V \setminus \Lambda}$ because all variables in $\Lambda$ are pinned. Furthermore, $(\theta * \mu)^\sigma_{V \setminus \Lambda}$ is a Gibbs distribution of a ferromagnetic two-spin system on the induced subgraph $G[V \setminus \Lambda]$ with the same edge parameters and with external fields bounded by
\[
\lambda_v \theta < \lambda_c \cdot \theta = \frac{1}{2} < 1 < \lambda_0.
\]
Using \Cref{thm:glauber-mixing-2}, for any $\Lambda \subseteq V$ and any $\sigma \in \{0,1\}^\Lambda$, the mixing time of the Glauber dynamics on $(\theta * \mu)^\sigma$ started from an arbitrary configuration is at most
\begin{align*}
 \forall \epsilon > 0, \quad t_{\textnormal{mix}}^{\textnormal{Glauber}}((\theta * \mu)^\sigma,\epsilon) = O\left(n (\log n)^{C} \log \frac{1}{\epsilon}\right),
\end{align*}
where $C = C(\beta,\gamma,\lambda) > 0$ is a constant depending on $\beta,\gamma,\lambda$. 
As a consequence, the spectral gap of the Glauber dynamics on $(\theta * \mu)^\sigma$ is at least $\Omega(n^{-1} (\log n)^{-C})$ (see \cite[Theorem 12.5]{LP17}).
Define 
\begin{align}\label{eq:gamma-min-theta}
\gamma_{\text{min}}^{}(\theta) \defeq \min\left\{  \gamma_{\text{Glauber}}^{}((\theta * \mu)^\sigma) \mid \sigma \in \{0,1\}^\Lambda, \Lambda \subseteq V \right\} = \Omega \tp{\frac{1}{n (\log n)^{C}}},
\end{align} 
where $\gamma_{\text{Glauber}}^{}((\theta * \mu)^\sigma)$ is the spectral gap of the Glauber dynamics on $(\theta * \mu)^\sigma$.
Let $\gamma_{\text{field}}^{}(\mu,\theta)$ denote the spectral gap of the field dynamics on $\mu$ with parameter $\theta$. The spectral gap of the Glauber dynamics on $\mu$ can be lower-bounded by the following proposition.

\begin{proposition}[\cite{CFYZ21}]\label{prop:gamma-glauber-field-min}
$\gamma_{\textnormal{Glauber}}^{}(\mu) \geq \gamma_{\textnormal{field}}^{}(\mu,\theta) \cdot \gamma_{\textnormal{min}}^{}(\theta)$.
\end{proposition}

Combining \eqref{eq:gamma-field-theta}, \eqref{eq:theta-definition}, \eqref{eq:gamma-min-theta}, and \Cref{prop:gamma-glauber-field-min}, we obtain the following lower bound on the spectral gap of the Glauber dynamics:
\begin{align}\label{eq:gamma-glauber-lower-bound}
    \gamma_{\textnormal{Glauber}}^{}(\mu) \geq \gamma_{\textnormal{field}}^{}(\mu,\theta) \cdot \gamma_{\textnormal{min}}^{}(\theta) = \Omega_{\beta,\gamma,\lambda} \tp{\frac{1}{n (\log n)^{C}}}.
\end{align}

Finally, we bound the mixing time of the Glauber dynamics on $\mu$.
Suppose the starting configuration is the all-1 configuration $X_0 = \boldsymbol{1}$. For any configuration $\tau \in \{0,1\}^V$, it holds that 
\begin{align*}
\frac{\mu(\boldsymbol{1})}{\mu(\tau)} \geq \min\{1,\lambda^{-1}\}^n \geq \lambda_c^{-n}.
\end{align*}
The above inequality holds because $\boldsymbol{1}$ maximizes the factors contributed by all edges; for each vertex, the factor contributed by $\boldsymbol{1}$ is $1$, whereas the factor contributed by $\tau$ is at most $\max\{1,\lambda\}$. Since there are $2^n$ configurations in total, we have
\begin{align}\label{eq:mu-1-lower-bound}
   {\mu(\boldsymbol{1})} \geq (2 \lambda_c)^{-n}.
\end{align}
Combining~\eqref{eq:gamma-glauber-lower-bound} and~\eqref{eq:mixing-time-glauber-gap}, the mixing time of the Glauber dynamics starting from the all-1 configuration is 
\begin{align*}
t^{\textnormal{Glauber}}_{\textnormal{mix-}\boldsymbol{1}}(\epsilon) = O\tp{ \frac{1}{ \gamma_{\textnormal{Glauber}}^{}(\mu) } \log \frac{1}{\epsilon^2 \mu(\boldsymbol{1})}} = O_{\beta,\gamma,\lambda}\left(n^2 (\log n)^C \log \frac{1}{\epsilon}\right).
\end{align*}

To bound the mixing time of the Glauber dynamics on $\mu$ starting from an arbitrary configuration, combine~\eqref{eq:gamma-glauber-lower-bound}, \Cref{lemma:warm-start-configuration-general}, and~\eqref{eq:general-bound-on-mixing-time} to obtain
\begin{align*}
    t^{\textnormal{Glauber}}_{\textnormal{mix}}(\epsilon) = O_{\beta,\gamma,\lambda}\left(n^3 (\log n)^{C+1} \log \frac{1}{\epsilon}\right).
\end{align*}
\Cref{thm:glauber-mixing-1} now follows after increasing the constant $C$ in the theorem by $2$. The extra $\log n$ factor absorbs the constants hidden in the notation $O_{\beta,\gamma,\lambda}(\cdot)$.

\ifdoubleblind
\else
  \section*{Acknowledgement}
    We thank Konrad Anand and Graham Freifeld for useful discussions at an early stage of this paper.

    This project has received funding from the European Research Council (ERC) under the European Union's Horizon 2020 research and innovation programme (grant agreement No.~947778).
    Weiming Feng acknowledges the support of ECS grant 27202725 from Hong Kong RGC.
\fi

\bibliographystyle{alpha}
\bibliography{ref}

@article{shao2021contraction,
  title={Contraction: A unified perspective of correlation decay and zero-freeness of 2-spin systems},
  author={Shao, Shuai and Sun, Yuxin},
  journal={J. Stat. Phys.},
  volume={185},
  number={2},
  pages={12},
  year={2021},
  publisher={Springer}
}

@inproceedings{ChenE22,
  author       = {Yuansi Chen and
                  Ronen Eldan},
  title        = {Localization Schemes: {A} Framework for Proving Mixing Bounds for
                  Markov Chains (extended abstract)},
  booktitle    = {FOCS},
  pages        = {110--122},
  publisher    = {{IEEE}},
  year         = {2022},
}

@inproceedings{Chen0YZ22,
  author       = {Xiaoyu Chen and
                  Weiming Feng and
                  Yitong Yin and
                  Xinyuan Zhang},
  title        = {Optimal mixing for two-state anti-ferromagnetic spin systems},
  booktitle    = {FOCS},
  pages        = {588--599},
  publisher    = {{IEEE}},
  year         = {2022},
}

@inproceedings{LiuLZ14,
  author       = {Jingcheng Liu and
                  Pinyan Lu and
                  Chihao Zhang},
  title        = {The Complexity of Ferromagnetic Two-spin Systems with External Fields},
  booktitle    = {{RANDOM} },
  series       = {LIPIcs},
  pages        = {843--856},
  publisher    = {Schloss Dagstuhl - Leibniz-Zentrum f{\"{u}}r Informatik},
  year         = {2014},
}

@inproceedings{GLL20,
  author       = {Heng Guo and
                  Jingcheng Liu and
                  Pinyan Lu},
  title        = {Zeros of ferromagnetic 2-spin systems},
  booktitle    = {{SODA}},
  pages        = {181--192},
  publisher    = {{SIAM}},
  year         = {2020}
}

@article{DGGJ04,
  author       = {Martin E. Dyer and
                  Leslie Ann Goldberg and
                  Catherine S. Greenhill and
                  Mark Jerrum},
  title        = {The Relative Complexity of Approximate Counting Problems},
  journal      = {Algorithmica},
  volume       = {38},
  number       = {3},
  pages        = {471--500},
  year         = {2004}
}

@article {ES88,
    AUTHOR = {Edwards, Robert G. and Sokal, Alan D.},
     TITLE = {Generalization of the {F}ortuin-{K}asteleyn-{S}wendsen-{W}ang
              representation and {M}onte {C}arlo algorithm},
   JOURNAL = {Phys. Rev. D (3)},
  FJOURNAL = {Physical Review. D. Particles and Fields. Third Series},
    VOLUME = {38},
      YEAR = {1988},
    NUMBER = {6},
     PAGES = {2009--2012},
}

@article{FK72,
    AUTHOR = {Fortuin, Cees M. and Kasteleyn, Piet W.},
     TITLE = {On the random-cluster model. {I}. {I}ntroduction and relation
              to other models},
   JOURNAL = {Physica},
    VOLUME = {57},
      YEAR = {1972},
     PAGES = {536--564},
}

@article{FengGW23,
  author       = {Weiming Feng and
                  Heng Guo and
                  Jiaheng Wang},
  title        = {Swendsen-{W}ang dynamics for the ferromagnetic Ising model with external
                  fields},
  journal      = {Inf. Comput.},
  volume       = {294},
  pages        = {105066},
  year         = {2023},
}

@article {GuoJ18,
    AUTHOR = {Guo, Heng and Jerrum, Mark},
     TITLE = {Random cluster dynamics for the {I}sing model is rapidly
              mixing},
   JOURNAL = {Ann. Appl. Probab.},
  FJOURNAL = {The Annals of Applied Probability},
    VOLUME = {28},
      YEAR = {2018},
    NUMBER = {2},
     PAGES = {1292--1313},
}

@article{JerrumS93,
  author       = {Mark Jerrum and
                  Alistair Sinclair},
  title        = {Polynomial-Time Approximation Algorithms for the Ising Model},
  journal      = {{SIAM} J. Comput.},
  volume       = {22},
  number       = {5},
  pages        = {1087--1116},
  year         = {1993},
}

@article{GoldbergJP03,
  author       = {Leslie Ann Goldberg and
                  Mark Jerrum and
                  Mike Paterson},
  title        = {The computational complexity of two-state spin systems},
  journal      = {Random Struct. Algorithms},
  volume       = {23},
  number       = {2},
  pages        = {133--154},
  year         = {2003},
}

@incollection{Hin12,
  author       = {Geoffrey E. Hinton},
  title        = {A Practical Guide to Training Restricted Boltzmann Machines},
  booktitle    = {Neural Networks: Tricks of the Trade (2nd ed.)},
  series       = {Lecture Notes in Computer Science},
  pages        = {599--619},
  publisher    = {Springer},
  year         = {2012}
}

@misc{Nobel24,
  title = {Press release for the {N}obel prize in physics},
  howpublished = {\url{https://www.nobelprize.org/prizes/physics/2024/press-release/}},
  note = {Accessed: 2026-03-30},
  year = {2024}
}

@inproceedings{Sly10,
  author       = {Allan Sly},
  title        = {Computational Transition at the Uniqueness Threshold},
  booktitle    = {{FOCS}},
  pages        = {287--296},
  publisher    = {{IEEE} Computer Society},
  year         = {2010}
}

@article{GSV16,
	Author = {Galanis, Andreas and {\v{S}}tefankovi\v{c}, Daniel and Vigoda, Eric},
	Journal = {Combin. Probab. Comput.},
	Number = {4},
	Pages = {500--559},
	Title = {Inapproximability of the partition function for the antiferromagnetic {I}sing and hard-core models},
	Volume = {25},
	Year = {2016}}

@article{SS14,
	Author = {Sly, Allan and Sun, Nike},
	Journal = {Ann. Probab.},
	Number = {6},
	Pages = {2383--2416},
	Title = {Counting in two-spin models on {$d$}-regular graphs},
	Volume = {42},
	Year = {2014}}

@article{HOT06,
  author       = {Geoffrey E. Hinton and
                  Simon Osindero and
                  {Yee Whye} Teh},
  title        = {A Fast Learning Algorithm for Deep Belief Nets},
  journal      = {Neural Comput.},
  volume       = {18},
  number       = {7},
  pages        = {1527--1554},
  year         = {2006}
}

@article{Hin02,
  author       = {Geoffrey E. Hinton},
  title        = {Training Products of Experts by Minimizing Contrastive Divergence},
  journal      = {Neural Comput.},
  volume       = {14},
  number       = {8},
  pages        = {1771--1800},
  year         = {2002}
}

@inproceedings{MH10,
  author       = {Abdel{-}rahman Mohamed and
                  Geoffrey E. Hinton},
  title        = {Phone recognition using Restricted {B}oltzmann Machines},
  booktitle    = {{ICASSP}},
  pages        = {4354--4357},
  publisher    = {{IEEE}},
  year         = {2010}
}

@inproceedings{SMH07,
  author       = {Ruslan Salakhutdinov and
                  Andriy Mnih and
                  Geoffrey E. Hinton},
  title        = {Restricted {B}oltzmann machines for collaborative filtering},
  booktitle    = {{ICML}},
  pages        = {791--798},
  publisher    = {{ACM}},
  year         = {2007}
}

@article{KQWW26,
  author       = {Youngwoo Kwon and
                  Qian Qin and
                  Guanyang Wang and
                  Yuchen Wei},
  title        = {A phase transition in sampling from Restricted {B}oltzmann Machines},
  journal      = {Ann. Appl. Probab.},
  year         = {2026+},
  note         = {to appear}
}

@inproceedings{Tos16,
  author       = {Christopher Tosh},
  title        = {Mixing Rates for the Alternating {G}ibbs Sampler over Restricted {B}oltzmann
                  Machines and Friends},
  booktitle    = {{ICML}},
  pages        = {840--849},
  publisher    = {JMLR.org},
  year         = {2016}
}

@article{MS13,
    AUTHOR = {Mossel, Elchanan and Sly, Allan},
     TITLE = {Exact thresholds for {I}sing-{G}ibbs samplers on general
              graphs},
   JOURNAL = {Ann. Probab.},
    VOLUME = {41},
      YEAR = {2013},
    NUMBER = {1},
     PAGES = {294--328},
}

@article{fill1991eigenvalue,
  title={Eigenvalue bounds on convergence to stationarity for nonreversible Markov chains, with an application to the exclusion process},
  author={Fill, James Allen},
  journal={The annals of applied probability},
  pages={62--87},
  year={1991},
  publisher={JSTOR}
}

@inproceedings{GuoKZ18,
  author       = {Heng Guo and
                  Kaan Kara and
                  Ce Zhang},
  title        = {Layerwise Systematic Scan: Deep Boltzmann Machines and Beyond},
  booktitle    = {
                  {AISTATS}},
  series       = {Proceedings of Machine Learning Research},
  volume       = {84},
  pages        = {178--187},
  publisher    = {{PMLR}},
  year         = {2018},
}

@inproceedings{FY26,
  title={Rapid Mixing of Glauber Dynamics for Monotone Systems via Entropic Independence},
  author={Feng, Weiming and Yang, Minji},
  booktitle={SODA},
  pages={4894--4929},
  year={2026},
  organization={SIAM}
}

@article{AHS85,
  author       = {David H. Ackley and
                  Geoffrey E. Hinton and
                  Terrence J. Sejnowski},
  title        = {A Learning Algorithm for {B}oltzmann Machines},
  journal      = {Cogn. Sci.},
  volume       = {9},
  number       = {1},
  pages        = {147--169},
  year         = {1985}
}

@article{Smo86,
 author = {Paul Smolensky},
 journal = {Information Processing in Dynamical Systems: Foundations of Harmony Theory},
 title = {Parallel Distributed Processing: Explorations in the Microstructure of Cognition},
 year = {1986},
 pages = {194--281},
}

@article{GuoL18,
  author       = {Heng Guo and
                  Pinyan Lu},
  title        = {Uniqueness, Spatial Mixing, and Approximation for Ferromagnetic 2-Spin
                  Systems},
  journal      = {{ACM} Trans. Comput. Theory},
  volume       = {10},
  number       = {4},
  pages        = {17:1--17:25},
  year         = {2018},
}

@article{BlancaCV20,
  author       = {Antonio Blanca and
                  Zongchen Chen and
                  Eric Vigoda},
  title        = {Swendsen-{W}ang dynamics for general graphs in the tree uniqueness region},
  journal      = {Random Struct. Algorithms},
  volume       = {56},
  number       = {2},
  pages        = {373--400},
  year         = {2020},
}

@inproceedings {weitz2006counting,
  author       = {Dror Weitz},
  title        = {Counting independent sets up to the tree threshold},
  booktitle    = {{STOC}},
  pages        = {140--149},
  publisher    = {{ACM}},
  year         = {2006}
}

@article{ALO24,
  author       = {Nima Anari and
                  Kuikui Liu and
                  Shayan {Oveis Gharan}},
  title        = {Spectral Independence in High-Dimensional Expanders and Applications
                  to the Hardcore Model},
  journal      = {{SIAM} J. Comput.},
  volume       = {53},
  number       = {6},
  pages        = {S20--1},
  year         = {2024}
}

@article{CLV23,
  author       = {Zongchen Chen and
                  Kuikui Liu and
                  Eric Vigoda},
  title        = {Rapid Mixing of {G}lauber Dynamics up to Uniqueness via Contraction},
  journal      = {{SIAM} J. Comput.},
  volume       = {52},
  number       = {1},
  pages        = {196--237},
  year         = {2023}
}

@article{CLV21,
  author = {Chen, Zongchen and Liu, Kuikui and Vigoda, Eric},
title = {Optimal Mixing of Glauber Dynamics: Entropy Factorization via High-Dimensional Expansion},
journal = {SIAM J. Comput.},
volume = {0},
number = {0},
pages = {STOC21-104-STOC21-153},
year = {2023},
}

@inproceedings{CFYZ21,
  author       = {Xiaoyu Chen and
                  Weiming Feng and
                  Yitong Yin and
                  Xinyuan Zhang},
  title        = {Rapid mixing of {G}lauber dynamics via spectral independence for all
                  degrees},
  booktitle    = {{FOCS}},
  pages        = {137--148},
  publisher    = {{IEEE}},
  year         = {2021}
}

@inproceedings{CZ23,
  author       = {Xiaoyu Chen and
                  Xinyuan Zhang},
  title        = {A Near-Linear Time Sampler for the {I}sing Model with External Field},
  booktitle    = {{SODA}},
  pages        = {4478--4503},
  publisher    = {{SIAM}},
  year         = {2023}
}

@inproceedings{AJKPV22,
  author       = {Nima Anari and
                  Vishesh Jain and
                  Frederic Koehler and
                  Huy Tuan Pham and
                  Thuy{-}Duong Vuong},
  title        = {Entropic independence: optimal mixing of down-up random walks},
  booktitle    = {{STOC}},
  pages        = {1418--1430},
  publisher    = {{ACM}},
  year         = {2022}
}

@book{Bar16,
  author       = {Alexander I. Barvinok},
  title        = {Combinatorics and Complexity of Partition Functions},
  series       = {Algorithms and combinatorics},
  volume       = {30},
  publisher    = {Springer},
  year         = {2016}
}

@article{PR17,
  author       = {Viresh Patel and
                  Guus Regts},
  title        = {Deterministic Polynomial-Time Approximation Algorithms for Partition
                  Functions and Graph Polynomials},
  journal      = {{SIAM} J. Comput.},
  volume       = {46},
  number       = {6},
  pages        = {1893--1919},
  year         = {2017}
}

@book{LP17,
    AUTHOR = {Levin, David A. and Peres, Yuval},
     TITLE = {Markov chains and mixing times},
   EDITION = {Second},
 PUBLISHER = {American Mathematical Society},
      YEAR = {2017},
     PAGES = {xvi+447},
}

@article{fill2013comparison,
  title={Comparison inequalities and fastest-mixing Markov chains},
  author={Fill, James Allen and Kahn, Jonas},
  journal={The Annals of Applied Probability},
  pages={1778--1816},
  year={2013},
  publisher={JSTOR}
}

\appendix

\section{One-step decay in general settings}\label{sec:decay-general}
In this section, we prove \Cref{lemma:potential-function-bound-GL} by generalising the proof in \cite{GuoL18}. Our proof primarily focuses on ferromagnetic two-spin systems, and the RBM case is treated separately whenever necessary. 
Recall \eqref{eq:phi-definition} and \eqref{eq:phi-definition-1}, the definitions of the potential functions.
The constant $t$ is specified as follows:
\begin{itemize}
  \item for ferromagnetic two-spin systems, $t\defeq\frac{(1-\alpha)\gamma}{\beta\gamma-1}\log\frac{\lambda+\gamma}{\beta \lambda+1}$, where $\alpha = \alpha(\beta,\gamma,\lambda)$ will be specified later in the proof;
  \item for RBMs, we set $\beta=1$ in the above. Namely, $t\defeq\frac{(1-\alpha)\gamma}{\gamma-1}\log\frac{\lambda+\gamma}{\lambda+1}$, where $\alpha = \alpha(1, \gamma, \lambda)$.
\end{itemize}

Thus, $t$ is a constant depending on $\beta$, $\gamma$, and $\lambda$ in the ferromagnetic two-spin system case, and a constant depending on only $\gamma$ and $\lambda$ in the RBM case. We also have $\frac{1}{x\log \frac{\lambda}{x}}\geq \frac{\mathrm{e}}{\lambda}$.
It implies
\begin{align*}
  \frac{1}{t}\ge \phi(x) \ge \min\left\{\frac{1}{t},\frac{\mathrm{e}}{\lambda}\right\},
\end{align*}
namely, \Cref{condition:phi-bound} holds for both cases. 

We verify \Cref{condition:potential-function-bound} next.
For any edge $e=(u,u_i)$, define the function $g_{\lambda,e}(x)$ for $x \in (0,\lambda)$ by:
\begin{align*}
g_{\lambda,e}(x):=\frac{(\beta_e\gamma_e-1)x\log\frac{\lambda}{x}}{(\beta_ex+1)(x+\gamma_e)\log\frac{x+\gamma_e}{\beta_e x+1}}.
\end{align*}
A useful property about the function $g_{\lambda,e}(x)$ is the following lemma.
Recall that the critical threshold $\lambda_c=\lambda_c(\beta,\gamma)=\left(\frac{\gamma}{\beta}\right)^{\frac{\sqrt{\beta\gamma}}{\sqrt{\beta\gamma}-1}}$.
\begin{lemma}[Lemma 3.3, \text{\cite{GuoL18}}]\label{lemma:guo-lemma-3.3}
  For any $\beta_e,\gamma_e$ such that $\beta_e\gamma_e>1$ and $\beta_e\le \gamma_e$,
  and any $x\in[0,\lambda_c(\beta_e,\gamma_e)]$, $g_{\lambda_c(\beta_e,\gamma_e),e}(x)\leq 1$.
\end{lemma}
Moreover, notice that $g_{\lambda,e}(x)$ is monotone increasing in $\lambda$ for any $x>0$.

We first handle the ferromagnetic two-spin system case, namely for any $\beta_e\leq \beta\leq 1<\gamma\leq \gamma_e$ and $\beta\gamma\geq \beta_e\gamma_e>1$. 
There are two main ingredients that will differ in the two cases:
\begin{enumerate}
  \item\label{item:g-bound} there exists a constant $\alpha\defeq\alpha(\beta,\gamma,\lambda)$ such that $\alpha\in(0,1)$ and $g_{\lambda,e}(x) \leq 1-\alpha$ for all $x \in (0,\lambda)$;
  \item\label{item:t-bound} for any $x \in (0,\lambda)$, we have
    \begin{align}\label{eq:t-bound}
      t<\frac{(1-\alpha)(\beta_e x+1)(x+\gamma_e)}{\beta_e\gamma_e-1}\log\frac{x+\gamma_e}{\beta_e x+1},
    \end{align}
    where $t=\frac{(1-\alpha)\gamma}{\beta\gamma-1}\log\frac{\lambda+\gamma}{\beta \lambda+1}$ by our choice.
\end{enumerate}

Property \eqref{item:t-bound} is straightforward for ferromagnetic 2-spin systems,
because $(\beta_e x+1)(x+\gamma_e)>\gamma$, $\beta_e\gamma_e-1\leq \beta\gamma-1$ and $\frac{\lambda+\gamma}{\beta \lambda+1}\leq \frac{\lambda+\gamma_e}{\beta_e \lambda+1}<\frac{x+\gamma_e}{\beta_e x+1}$.

For Property \eqref{item:g-bound},
we have $\log \frac{x+\gamma_e}{\beta_e x+1}\geq \log \frac{\lambda+\gamma_e}{\beta_e \lambda+1}\geq \log \frac{\lambda+\gamma}{\lambda+1}$ for $x\in (0,\lambda)$, where the first inequality holds because $\log \frac{x+\gamma_e}{\beta_e x+1}$ is monotone decreasing in $x$.
Thus,
\begin{align*}
    g_{\lambda,e}(x):=&\frac{(\beta_e\gamma_e-1)x\log\frac{\lambda}{x}}{(\beta_ex+1)(x+\gamma_e)\log\frac{x+\gamma_e}{\beta_e x+1}}\\
    \leq & \frac{(\beta\gamma-1)x\log\frac{\lambda}{x}}{\log\frac{x+\gamma_e}{\beta_e x+1}}\leq \frac{(\beta\gamma-1)x\log\frac{\lambda}{x}}{\log \frac{\lambda+\gamma}{\lambda+1}},
\end{align*}
where the first inequality holds because $\beta_ex+1\geq 1$, $x+\gamma_e\geq 1$ and $0<\beta_e\gamma_e-1\leq \beta\gamma-1$. 
As $(x\log \frac{\lambda}{x})' = \log \frac{\lambda}{x} - 1\geq 0$ for $0\leq x\leq \frac{\lambda}{e}$, we also have $x\log \frac{\lambda}{x}\to 0^+$ as $x\to 0^+$. 
Hence there exists a constant $x_0=x_0(\lambda,\beta,\gamma)\in (0,\lambda)$ such that for any $x\in (0,x_0]$, we have
\begin{align*}
    g_{\lambda,e}(x)\leq \frac{(\beta\gamma-1)x\log\frac{\lambda}{x}}{\log \frac{\lambda+\gamma}{\lambda+1}}\leq \frac{1}{2}.
\end{align*}
We denote $\lambda_c=\lambda_c(\beta,\gamma)$ and $\lambda_{c,e}\defeq\lambda_c(\beta_e,\gamma_e)$ for any edge $e$.
Since $\gamma_e\ge \gamma$, $\beta_e\le\beta$, and $\beta\gamma\ge\beta_e\gamma_e$,
$\lambda_c\le \lambda_{c,e}$.
For $x\in [x_0,\lambda)\subseteq [0,\lambda_{c,e})$, we have
\begin{align*}
  g_{\lambda,e}(x) & = g_{\lambda_c,e}(x)\cdot \frac{\log \lambda - \log x}{\log \lambda_c - \log x}
  \leq g_{\lambda_{c,e},e}(x)\cdot \frac{\log \lambda - \log x_0}{\log \lambda_c - \log x_0}\\
  &\le  \frac{\log \lambda - \log x_0}{\log \lambda_c - \log x_0} < 1,
\end{align*}
where the first inequality follows from the monotonicity of $g_{\lambda,e}(x)$ in $\lambda$ and $\lambda_c\le \lambda_{c,e}$,
and the second is by \Cref{lemma:guo-lemma-3.3}.
We set $\alpha = \alpha(\beta,\gamma,\lambda) \defeq 1 - \max\left\{\frac{1}{2}, \frac{\log \lambda - \log x_0}{\log \lambda_c - \log x_0}\right\}$ so that $g_{\lambda,e}(x) \leq 1-\alpha < 1$ for all $x \in (0,\lambda)$. 
This proves Property \eqref{item:g-bound}.

Next we claim that
\begin{align}\label{eq:property-for-decay}
    \forall x \in (0,\lambda), \quad \frac{(\beta_{e}\gamma_{e}-1)}{(\beta_{e}x+1)(x+\gamma_{e})} \frac{1}{\phi(x)}\leq (1-\alpha)\log \frac{x+\gamma_e}{\beta_e x+1},
\end{align}
We consider two different possible choices of $\phi(x)$ separately.
\begin{enumerate}
  \item If $\phi(x)=\frac{1}{x\log\frac{\lambda}{x}}$, this follows from Property \eqref{item:g-bound}.
  \item If $\phi(x)=\frac{1}{t}$, this follows from Property \eqref{item:t-bound}, namely \eqref{eq:t-bound}, directly.    
\end{enumerate}

Having verified~\eqref{eq:property-for-decay},
we have that
\begin{align*}
C_{\phi,d}(\boldsymbol{x}) =& \phi(F_u(\boldsymbol{x})) \sum_{i=1}^d \left\vert  \frac{\partial F_u}{\partial x_i}(\boldsymbol{x})  \right\vert \frac{1}{\phi(x_i)}\\
=&\phi(F_u(\boldsymbol{x})) \sum_{i=1}^d F_u(\boldsymbol{x}) \frac{(\beta_{e_i}\gamma_{e_i}-1)}{(\beta_{e_i}x_i+1)(x_i+\gamma_{e_i})} \frac{1}{\phi(x_i)}\\
\leq &\phi(F_u(\boldsymbol{x})) \sum_{i=1}^d F_u(\boldsymbol{x})(1-\alpha)\log \frac{x_i+\gamma_{e_i}}{\beta_{e_i} x_i+1}\qquad \text{by~\eqref{eq:property-for-decay}}\\
=&(1-\alpha)\phi(F_u(\boldsymbol{x})) F_u(\boldsymbol{x})\log \frac{\lambda_u}{F_u(\boldsymbol{x})}\\
\leq& (1-\alpha)\phi(F_u(\boldsymbol{x})) F_u(\boldsymbol{x})\log \frac{\lambda}{F_u(\boldsymbol{x})}\\
\leq& 1-\alpha,
\end{align*}
where the last inequality holds because $\phi(F_u(\boldsymbol{x}))=\min\left \{ \frac{1}{t},\frac{1}{F_u(\boldsymbol{x})\log \frac{\lambda}{F_u(\boldsymbol{x})}}\right \}\leq \frac{1}{F_u(\boldsymbol{x})\log \frac{\lambda}{F_u(\boldsymbol{x})}}$.

For the RBM case, all we need to do is to verify Properties \eqref{item:g-bound} and \eqref{item:t-bound} without using $\beta\gamma\ge\beta_e\gamma_e$.
The rest of the proof is identical.

For Property \eqref{item:g-bound},
we first claim that $g_{\lambda,e}(x)$ is monotone decreasing w.r.t.\ $\gamma_e$. 
To verify this, we show that $h(\gamma_e):=\frac{\gamma_e-1}{(x+\gamma_e)\log \frac{x+\gamma_e}{x+1}}$ is monotone decreasing when $\gamma_e\geq \gamma$.

Let $z\defeq\frac{x+\gamma_e}{x+1}>1$. Then $h(\gamma_e)=\frac{z-1}{z\log z}$. We have 
\begin{align*}
    \left(\frac{z-1}{z\log z}\right)'= \frac{\log z -z +1}{z^2 \log^2 z} < 0, 
\end{align*}
where we use $\log z < z-1$ for $z>1$. Combining with $\frac{d z}{d \gamma_e}=\frac{1}{x+1}>0$, we conclude that $h(\gamma_e)$ is monotone decreasing w.r.t.\ $\gamma_e$ when $\gamma_e\geq \gamma$. 
Hence, to show Property \eqref{item:g-bound}, 
we only need to consider the extremal case $\gamma_e=\gamma$. 
This extremal case is a special case of ferromagnetic 2-spin systems, where $\beta=1=\beta_e$ and $\gamma_e=\gamma$ for all edges $e$. 
Therefore, the argument above for ferromagnetic 2-spin systems applies, and Property \eqref{item:g-bound} holds. 

For Property \eqref{item:t-bound},
notice that the right hand side of \eqref{eq:t-bound} is in fact $\frac{(1-\alpha)(x+1)}{h(\gamma_e)}$.
Since $h(\gamma_e)$ is decreasing in $\gamma_e$ and $\gamma_e\ge \gamma$,
$\frac{(1-\alpha)(x+1)}{h(\gamma_e)} \ge \frac{(1-\alpha)(x+1)}{h(\gamma)}$,
which is larger than $t$, the left hand side of \eqref{eq:t-bound}.

\section{Monotone coupling}\label{app:censoring}

\subsection{Proof of \texorpdfstring{\Cref{prop:monotone-of-ferro-ising}}{Proposition 29} and \texorpdfstring{\Cref{prop:monotone-coupling}}{Proposition 30}}
We first prove a lemma to show the monotonicity of the conditional marginal distribution induced by a ferromagnetic two-spin system.
\begin{lemma}\label{lemma:monotone-glauber-dynamics}
Let $\mu$ be a Gibbs distribution of a ferromagnetic two-spin system on graph $G = (V,E)$. Consider any $\Lambda\subseteq V$ and two 
partial configurtions $\sigma\preceq \tau\in \{0,1\}^{\Lambda}$, it holds that
\begin{align}\label{eq:monotone-block-dynamics}
\forall v\in V\setminus \Lambda, \quad \mu_v^{\sigma_{V\setminus \Lambda}}(1) \leq \mu_v^{\tau_{V\setminus \Lambda}}(1).
\end{align}
\end{lemma}
\begin{proof}
We fix a vertex $v\in V\setminus \Lambda$ and prove~\eqref{eq:monotone-block-dynamics}. Consider the SAW tree $T_{\sigma}=T_{\text{SAW}}(G,v,\sigma)$ and $T_{\tau}=T_{\text{SAW}}(G,v,\tau)$, which differ only in the pinning of the leaf nodes.
For any vertex $w$ in the SAW tree, define $p_w^{T_{\sigma}}$ and $p_w^{T_{\tau}}$ as the marginal probabilities of the vertex $w$ in the sub-trees $T_{w,\sigma}$ and $T_{w,\tau}$ rooted at $w$, respectively. Because $\sigma\preceq \tau$, for any leaf node $u$ with pinning in $T_{\sigma},T_{\tau}$, the pinning in $T_{\sigma}$ is at most\footnote{Here, we compare the value $\{0,1\}$ of pinning. The value 0 is smaller than the value 1} the pinning in $T_{\tau}$.
We have $R_u^{T_{\sigma}}=\frac{p_u^{T_{\sigma}}(0)}{p_u^{T_{\sigma}}(1)}\geq R_u^{T_{\tau}}=\frac{p_u^{T_{\tau}}(0)}{p_u^{T_{\tau}}(1)}$. For each parameter of the recursion function $F_w(\cdot)$ in~\eqref{eq:tree-recursion}, where $w$ is any non-leaf node in the SAW tree, the recursion function is monotone increasing with the parameter. We can recursively prove that for any non-leaf node $w$ in the SAW tree, we have $R_w^{T_{\sigma}}\geq R_w^{T_{\tau}}$, from bottom to top. Using a inductive proof, for the root node $v$, we can show that $R_v^{T_{\sigma}}\geq R_v^{T_{\tau}}$, which implies $\mu_v^{\sigma_{V\setminus \Lambda}}(1) \leq \mu_v^{\tau_{V\setminus \Lambda}}(1)$. 
\end{proof}

Now, we first prove \Cref{prop:monotone-coupling} and then \Cref{prop:monotone-of-ferro-ising}.

Let us consider the heat-bath block dynamics at first.
Let $b = |\+B|$. Assume $V=\{v_1,v_2,\ldots,v_n\}$. We construct the monotone coupling $f$ as follows. For any configuration $\sigma \in \Omega$ and $r=(r_0r_1\ldots r_n)\in [0,1]^{n+1}$, we determine the configuration $f(\sigma,r) \in \{0,1\}^V$.
There exists $i\in [1,b]$ such that $r_0\in [(i-1)/b,i/b)$, and we choose the $i$-th block $B_i=\{v_{i_1},v_{i_2},\ldots,v_{i_j}\}$, where $1\leq i_1<i_2<\ldots<i_j\leq n$. 
To simplify the notation, let $\rho = f(\sigma,r)$.
We set $\rho(V\setminus B_i)=\sigma(V\setminus B_i)$. Let $B_i^k={\{v_{i_k},v_{i_2},\ldots,v_{i_j}\}}$ for $1\leq k \leq j+1$. We need to resample vertices in $B_i$ conditioned on $\sigma(V\setminus B_i)$, and we recursively decide the value
of $\rho(v_{i_k})$ in increasing order of $k$, such that 
\begin{align}\label{eq:recursive-v_i}
    \rho(v_{i_k})\sim \mu_{v_{i_k}}^{\rho(V\setminus B_i^k)}.
\end{align}
Assume that we have decided the value of $\rho(V\setminus B_i^k)$.
If $r_k\leq \mu_{v_{i_k}}^{\rho(V\setminus B_i^k)}(1)$, we set $\rho(v_{i_k})=1$. Otherwise, we set $\rho(v_{i_k})=0$.
It is easy to verify that for any $\sigma \in \Omega$, the distribution of $f(\sigma,\boldsymbol{r}) = \rho$ is exactly the distribution of one step of the heat-bath block dynamics on $\mu$ starting from $\sigma$. This proves that $f$ is a valid coupling.
We only need to check that $f(\sigma,r)\preceq f(\tau,r)$ with probability 1 if $\sigma \preceq \tau$. 
To simplify the notation, let $\rho = f(\sigma,r)$ and $\rho' = f(\tau,r)$.
We first have $\sigma(V \setminus B_i) = \rho(V\setminus B_i)\preceq \rho'(V\setminus B_i) = \tau(V \setminus B_i)$. 
Assume that we have $\rho(V\setminus B_i^k)\preceq \rho'(V\setminus B_i^k)$, for some $0\leq k\leq j-1$. By \Cref{lemma:monotone-glauber-dynamics}, we have $\mu_{v_{i_k}}^{\rho(V\setminus B_i^k)}(1)\leq \mu_{v_{i_k}}^{\rho'(V\setminus B_i^k)}(1)$. 
By our construction, $\rho(v_{i_k})\leq \rho'(v_{i_k})$, so we have $\rho(V\setminus B_i^{k+1})\preceq \rho'(V\setminus B_i^{k+1})$. By induction, we have $\rho(V\setminus B_i^{j+1})\preceq \rho'(V\setminus B_i^{j+1})$, which implies $\rho\preceq \rho'$. Hence, $f$ is a monotone coupling of the heat-bath block dynamics on $\mu$. Since our analysis holds for any $B_i\subseteq V$, we can couple two dynamics to pick the same block, then the block dynamics part in \Cref{prop:monotone-coupling} holds.

For systematic scan block dynamics, we can construct the monotone coupling in the same way as above, except that we choose the block $B_i$ according to the systematic scan order instead of the random choice. Using the above argument for each block $B_i$ and doing an induction on all blocks  shows that there exists a monotone coupling of the systematic-scan block dynamics on $\mu$. This proves the systematic scan block dynamics part in \Cref{prop:monotone-coupling}.

Finally, \Cref{prop:monotone-of-ferro-ising} is a simple consequence of the above prove. The above proof works for all blocks $B_i\subseteq V$. Given two configurations $\sigma,\tau \in \{0,1\}^\Lambda$, where $\Lambda \subseteq V$ is any subset, if $\sigma \preceq \tau$, we can use the same process (with $B_i = V \setminus \Lambda$) to couple $ X \sim \mu^\sigma_{V \setminus \Lambda}$ and $Y \sim \mu^\tau_{V \setminus \Lambda}$ such that $X \preceq Y$ with probability 1. This proves \Cref{prop:monotone-of-ferro-ising}.

\subsection{Proof of \texorpdfstring{\Cref{claim:censoring}}{Claim 35}}
We first list some definitions about the comparison between Markov chains.
\begin{definition}[Increasing function]
    We say a function $f:\{0,1\}^V\to \mathbb{R}$ is increasing if for any $\sigma,\tau\in \{0,1\}^V$ with $\sigma\preceq \tau$, it holds that $f(\sigma)\leq f(\tau)$.
\end{definition}

\begin{definition}[Monotone Markov chain]
    We say a Markov chain with transition matrix $P$ on $\{0,1\}^V$ is monotone if for any increasing function $f:\{0,1\}^V\to \mathbb{R}$, $Pf$ is also an increasing function.
\end{definition}

\begin{definition}
    Let $\mu$ be a distribution over $\{0,1\}^V$.
For two monotone Markov chains $P$ and $Q$  on $\{0,1\}^V$, we say $P\preceq_{mc} Q$ if for any increasing function $f,g:\{0,1\}^V\to \mathbb{R}_+$, we have
\begin{align*}
    \left \langle Pf,g \right \rangle_{\mu} \leq  \left \langle Qf,g \right \rangle_{\mu},
\end{align*}
where $\left \langle f_1,f_2 \right \rangle_{\mu}:=\sum_{x\in \{0,1\}^V} f_1(x)f_2(x)\mu(x)$ for any functions $f_1,f_2:\{0,1\}^V\to \mathbb{R}$.
\end{definition}

Fix a distribution $\mu$ over $\{0,1\}^V$.
For any block $B \subseteq V$, let $P_B$ be the transition matrix of the block update on $B$: given any $\sigma \in \{0,1\}^V$, $P_B$ resamples the configuration on $B$ conditional on the current configuration of other variables: $\sigma(B) \sim \mu_{B}^{\sigma(V \setminus B)}$. Similarly, $P_{B \cap S}$ is the transition matrix of the block update on $B\cap S$. The following monotonicity result is known.
\begin{lemma}[\text{\cite{BlancaCV20}}, Proof of Lemma 15]\label{lemma:monotone-each-block}
For any block $B\subseteq V$ and subset $S\subseteq V$, 
\begin{align*}
    P_{B}\preceq_{mc} P_{B\cap S}.
\end{align*}
\end{lemma}

Next, recall that $\preceq_D$ is the stochastic dominance relation for two distributions defined in \Cref{claim:censoring}.
\begin{proposition}[\text{\cite[Proposition 22.7]{LP17}}]\label{prop:monotone-equiv-definition}
For any Markov chain $P$ on $\{0,1\}^V$, the following three statements are equivalent:
\begin{itemize}
    \item $P$ is a monotone Markov chain;
    \item For any two configurations $\sigma,\sigma'\in \{0,1\}^V$ with $\sigma\preceq \sigma'$, we have $P(\sigma,\cdot)\preceq_D P(\sigma',\cdot)$;
    \item for any two distributions $\nu_0\preceq_D \nu_1$, we have $\nu_0P\preceq_D \nu_1P$.
\end{itemize}
\end{proposition}

By the proof of \Cref{prop:monotone-of-ferro-ising}, the second condition in \Cref{prop:monotone-equiv-definition} holds for transition matrix $P_B$ corresponding to any block $B\subseteq V$. Hence, $P_B$ is a monotone Markov chain. 

\begin{lemma}[\text{\cite{fill2013comparison}}, Prop. 2.3, 2.4]\label{lemma:monotone-mc-property}
Let $P_i$, $Q_i$ are Markov chains that are reversible w.r.t. $\mu$ and monotone for $i\in \{1,2,\ldots,\ell\}$, the following statements hold:
\begin{itemize}
    \item If $P_i\preceq_{mc} Q_i$ for each $i\in \{1,2,\ldots,\ell\}$, then $\frac{1}{\ell}\sum_{i=1}^\ell P_i \preceq_{mc} \frac{1}{\ell}\sum_{i=1}^\ell Q_i$; 
    \item If $P_i\preceq_{mc} Q_i$ for each $i\in \{1,2,\ldots,\ell\}$, then $P_1P_2\cdots P_\ell \preceq_{mc} Q_1Q_2\cdots Q_\ell$.
\end{itemize}
\end{lemma}

The following lemma holds for both the heat-bath and the systematic scan block dynamics.

\begin{lemma}\label{lemma:monotone-censoring}
Let $P$ be the transition matrix of a block dynamics on $\mu$ with a set of blocks $B = \{B_1, B_2, \ldots, B_r\}$. Let $P_S^{\textnormal{censored}}$ be the transition matrix of the censored block dynamics on $\mu$ w.r.t. $S\subseteq V$, then $P\preceq_{mc} P_S^{\textnormal{censored}}$.
\end{lemma}

\begin{proof}
By \Cref{lemma:monotone-each-block}, we have $P_{B_i}\preceq_{mc} P_{B_i\cap S}$ for each $i\in [r]$. 
We also have $P_{B_i}$ and $P_{B_i\cap S}$ are reversible with stationary distribution $\mu$ for each $i\in [r]$. For the heat-bath block dynamics, by the first statement of \Cref{lemma:monotone-mc-property}, we have 
$$P = \frac{1}{r}\sum_{i=1}^r P_{B_i} \preceq_{mc} \frac{1}{r}\sum_{i=1}^r P_{B_i\cap S} = P_S^{\textnormal{censored}}.$$
 Hence, the result holds for heat-bath block dynamics.
For systematic scan block dynamics. By the second statement of \Cref{lemma:monotone-mc-property}, we have 
\begin{align*}
    P = P_{B_r}P_{B_{r-1}}\cdots P_{B_1} \preceq_{mc} P_{B_r\cap S}P_{B_{r-1}\cap S}\cdots P_{B_1\cap S} = P_S^{\textnormal{censored}}. 
\end{align*}
Hence, the result holds for systematic scan block dynamics.
\end{proof}

Let $\mu$ be a distribution over $\{0,1\}^V$.
Let $A \subseteq 2^V$ be a collection of censoring subsets.
For any block dynamics $P$ on $\{0,1\}^V$ with stationary distribution $\mu$, and $P\preceq_{mc} P_S^{\textnormal{censored}}$ for $S\in A$.
Let $S_1, S_2, \ldots$ be a sequence of censoring subsets in $A$.
Let $(X_t)_{t\geq 0}$ be the heat-bath or systematic scan block dynamics on $\mu$ with transition matrix $P$ and block set $\+B =\{B_1, B_2, \ldots, B_r\}$.
Let $(Y_t)_{t\geq 0}$ be the censored block dynamics on $\mu$ with transition matrix $P_{S_i}^{\textnormal{censored}}$ in step $i$.
Formally, the transition matrix of $(Y_t)_{t\geq 0}$ in $i$-th step is
\begin{align*}
\begin{cases}
P_{S_i}^{\textnormal{censored}} = \frac{1}{r}\sum_{j=1}^r P_{B_j\cap S_i} & \text{if $P$ is heat-bath block dynamics}, \\
P_{S_i}^{\textnormal{censored}} = P_{B_r\cap S_i}P_{B_{r-1}\cap S_i}\cdots P_{B_1\cap S_i} & \text{if $P$ is systematic scan block dynamics}.
\end{cases}
\end{align*}

We use the following result in our proof.
\begin{lemma}[\text{\cite{BlancaCV20}}, Theorem 7]\label{lemma:monotone-censoring-to-dominance}
Suppose two initial configurations $X_0$, $Y_0$ are both sampled from the same distribution $\nu$ over $\{0,1\}^V$.
The following properties hold:
\begin{itemize}
    \item If $\nu/\mu$ is increasing, where $\nu/\mu(x) = \frac{\nu(x)}{\mu(x)}$, then for any $t\geq 0$, $X_t\preceq_D Y_t$;
    \item If $-\nu/\mu$ is increasing, then for any $t\geq 0$, $ Y_t\preceq_D X_t$.
\end{itemize}
\end{lemma}


Now we are ready to prove \Cref{claim:censoring}.
For parameters in \Cref{lemma:monotone-censoring-to-dominance}, we set $A = \{V, S_v\}$, $B_i=V$ for $1\leq i\leq s$ and $B_i=S_v$ for $i>s$.
Applying the result in \Cref{lemma:monotone-censoring} we have $P\preceq_{mc} P_{V} = P$ and $P\preceq_{mc} P_{S_v}$.
We first let $\nu=1^V$ and apply the first statement in \Cref{lemma:monotone-censoring-to-dominance} to obtain that for any $j>s$, $X_j^+\preceq_D Y_j^+$. 
Then we let $\nu=0^V$ and apply the second statement in \Cref{lemma:monotone-censoring-to-dominance} to obtain that for any $j>s$, $Y_j^-\preceq_D X_j^-$.
By inductively applying the third statement in \Cref{prop:monotone-equiv-definition}, we have for any $j>s$, $X_j^-\preceq_D X_j^+$. Combining above three relationships, we have
\[\forall j \geq 0, \quad Y_j^-\preceq_D X_j^- \preceq_D X_j^+\preceq_D Y_j^+.\] 
Indeed, the above relationships hold for any $j\geq 0$. 
\end{document}